\documentclass[12pt]{article}
\usepackage{amsmath,amssymb,amsthm}
\usepackage{mathtools}
\usepackage{graphics,epsfig}
\usepackage{hyperref}
\usepackage{natbib}
\usepackage{color}
\usepackage{graphicx}
\usepackage{caption}
\usepackage{subcaption}
\usepackage{float}
\usepackage{mathrsfs}
\usepackage{multirow}
\usepackage{comment}
\usepackage{setspace}
\usepackage{bbm}
\usepackage{bm}
\usepackage{amsfonts}
\usepackage{xcolor}
\usepackage{pslatex}
\usepackage{setspace}
\usepackage{bbm}
\usepackage{pgfplots}
\usepgfplotslibrary{statistics}

\pgfplotsset{compat=1.18}
\def \gg {\boldsymbol{g}}

	\newtheorem{remark}{\bf Remark}
	\newtheorem{rmk}[remark]{\bf Remark}

	\newtheorem{theorem}{\bf Theorem}
	
	\newtheorem{prop}[theorem]{\bf Proposition}
	\newtheorem{lem}[theorem]{\bf Lemma}


\begin{document}

  \title{\bfseries Modeling Educational Performance Using School Demographics and Teacher Characteristics}

\author{Brianna Reed$^1$  \and
	Paramahansa Pramanik$^{2,3}$
}

\date{
	\small
	$^{1}$ Department of Education and Professional Studies, University of South Alabama, Mobile, AL 36688, United States.\\
	\vspace{0.5em}
	$^{2}$ Department of Mathematics and Statistics, University of South Alabama, Mobile, AL 36688, United States.\\
	$^{3}$Corresponding author, \texttt{ppramanik@southalabama.edu}
}

\maketitle

\begin{abstract}
	High-dimensional educational datasets often exhibit sparsity, grouped predictors, and locally correlated covariates, limiting the effectiveness of conventional regression methods. We propose an Adaptive Weighted Group Fused LASSO estimator that jointly performs adaptive variable selection, group regularization, and coefficient fusion within a unified penalized regression framework. An efficient ADMM algorithm is developed, and asymptotic properties, including consistency, oracle property, and debiased asymptotic normality, are established. Simulation studies demonstrate superior estimation and prediction performance compared with existing penalized methods. An application to Alabama public school mathematics proficiency data illustrates improved model interpretability, predictive accuracy, and identification of the most influential institutional predictors.
\end{abstract}

{\bf Keywords:} Educational statistics, Large-sample theory, School-level modeling, Statistical inference.

\section{Introduction:}

The increasing availability of large-scale educational administrative databases has created new opportunities for understanding factors associated with school-level academic performance. These data are typically characterized by high-dimensional predictor spaces, complex dependence structures, grouped explanatory variables, and substantial multicollinearity arising from institutional, demographic, and teacher-related characteristics. Classical regression techniques often perform poorly under such settings because coefficient estimates become unstable, variable selection is unreliable, and prediction accuracy deteriorates in the presence of correlated covariates. Consequently, modern statistical learning methods based on penalized regression have become an increasingly attractive alternative for analyzing educational data, as they simultaneously perform estimation, regularization, and variable selection while producing models that remain both interpretable and computationally efficient \citep{yuan2006model,friedman2010regularization}. From an applied statistical perspective, the principal objective is not merely to identify statistically significant predictors but to construct parsimonious models that generalize well to future observations while accurately recovering the underlying structure of the regression coefficients.

Although numerous penalized regression procedures have been proposed over the past two decades, most existing methods address only one aspect of structural regularization at a time. Coordinate-wise penalties such as the LASSO and Adaptive LASSO promote sparsity by shrinking individual regression coefficients toward zero, whereas Group LASSO exploits predefined group structures among predictors and Fused LASSO encourages neighboring coefficients to exhibit local homogeneity \citep{tibshirani2005sparsity}. In many real-world applications, however, these structural characteristics occur simultaneously rather than independently. Educational datasets, in particular, often contain correlated groups of institutional variables together with predictors whose effects exhibit local similarity, making separate regularization strategies insufficient for capturing the full complexity of the data-generating process. These observations motivate the development of a unified penalized regression framework capable of incorporating adaptive sparsity, grouped variable selection, and coefficient fusion within a single estimation procedure. The methodology developed in this paper addresses this need by introducing an Adaptive Weighted Group Fused LASSO estimator together with a scalable optimization algorithm, comprehensive asymptotic theory establishing consistency, oracle properties, and debiased inference, and an empirical application illustrating its practical advantages for modeling school-level mathematics proficiency.

This paper examines variation in mathematics proficiency across Alabama public schools using school-level administrative data. The analysis focuses on differences in educational outcomes among schools with distinct demographic compositions and further investigates whether these differences persist after accounting for socioeconomic disadvantage. In addition, we evaluate the extent to which institutional characteristics, specifically teacher certification status and teaching experience, are associated with school-level mathematics proficiency. Throughout this paper, the term minority students refers collectively to students identified as Black or African American, Hispanic, Asian, Pacific Islander, American Indian or Alaska Native, or belonging to two or more racial groups. Student gender is not included as a covariate in the analysis, and the dataset encompasses elementary, middle, and high schools throughout Alabama.

The increasing availability of large-scale educational administrative data has substantially expanded opportunities for statistically rigorous analyses of school performance and institutional effectiveness. Such data are inherently heterogeneous, reflecting variation in demographic composition, socioeconomic conditions, teacher characteristics, and school resources, thereby requiring statistical models capable of producing reliable estimation and valid inference while accounting for multiple sources of variability \citep{raudenbush1988educational,bryk1989quantitative,rowe1995methodological,page2017exploring,chang2021statistical}. From an applied statistical perspective, the primary objective is not merely to identify differences among schools but to quantify the association between institutional characteristics and educational outcomes within a principled inferential framework. Motivated by this objective, we develop a statistical modeling approach for school-level mathematics proficiency that is supported by large-sample theoretical results and subsequently illustrate its practical utility using administrative data from Alabama public schools. The proposed framework complements existing educational research on instructional equity and learning environments \citep{reed2005looking,nctm2014principles,harris1997growing} by emphasizing statistical estimation, inference, and model-based interpretation.

In the United States, where approximately 55\% of eighth-grade students belong to historically underrepresented demographic groups \citep{naep2022math}, understanding variation in educational outcomes requires careful consideration of differences in learning environments, institutional characteristics, and access to educational resources. From a statistical perspective, observed differences in academic performance reflect the combined influence of multiple demographic, socioeconomic, and school-level factors that interact in complex ways. Consequently, educational administrative data provide an important setting for developing inferential procedures capable of quantifying these associations while appropriately accounting for heterogeneity across schools \citep{raudenbush1988educational,bryk1989quantitative,rowe1995methodological,chang2021statistical}. Conceptually, educational support may be viewed as adapting instructional practices to students' individual learning needs rather than providing identical resources to every student. For example, revisiting the familiar illustration of children of different heights attempting to reach the same counter, identical assistance benefits them unequally, whereas support tailored to individual circumstances enables each child to perform the same task effectively. Likewise, classroom instruction that recognizes differences in prior preparation, learning styles, and educational needs may better facilitate academic achievement than a uniform instructional approach \citep{landsman2023white,shor1986equality}.

Reducing disparities in academic achievement remains a central objective of educational policy and empirical research \citep{jordan2010supporting}. National assessment data have consistently documented substantial variation in mathematics proficiency across demographic and socioeconomic groups. For example, the 2006 National Assessment of Educational Progress (NAEP) reported that 91\% of African American students and 87\% of Latino students in the eighth grade did not achieve mathematics proficiency, compared with 63\% of White students. The same assessment further showed that only 13\% of students from economically disadvantaged households attained proficiency or advanced performance, whereas the corresponding figure was 38\% among students from more economically advantaged households \citep{flores2007examining}. Similar patterns continue to be observed in recent statewide assessments. According to the 2024 Alabama assessment results, 94\% of Black students, 89\% of Hispanic students, and 67\% of White students were classified as below proficiency in mathematics \citep{nagb2025}. These persistent differences motivate the need for statistical methodologies capable of evaluating how institutional characteristics and school composition are associated with observed educational outcomes while providing valid uncertainty quantification and inferential guarantees \citep{page2017exploring}.

Observed differences in mathematics achievement should not be interpreted as evidence of inherent differences in students' academic ability. Rather, a substantial body of educational research argues that historical, institutional, and structural factors have contributed to unequal educational opportunities across schools and communities \citep{howard2019race}. A historically important example is the 1896 \emph{Plessy v. Ferguson} decision, which established the legal doctrine of ``Separate but Equal'' and permitted racial segregation in public facilities \citep{cates2012plessy,araujo2020century}. Although the doctrine ostensibly required equivalent educational provision, schools serving different populations frequently differed considerably in instructional quality, facilities, educational resources, and teacher availability \citep{darling2004color}. The subsequent \emph{Brown v. Board of Education} decision in 1954 declared segregated public schools inherently unequal and initiated the process of desegregation. Nevertheless, numerous studies suggest that historical inequities continue to influence educational systems through differences in institutional resources, school composition, and access to qualified teachers \citep{howard2016we}. These considerations further underscore the importance of statistical frameworks that distinguish systematic institutional effects from random variation, thereby enabling more reliable estimation and inference for educational administrative data \citep{raudenbush1988educational,chang2021statistical,page2017exploring}.

\subsection{Background Literature:}

A substantial body of educational research has examined the factors associated with variation in mathematics achievement across schools and student populations. Rather than attributing differences in educational outcomes to a single determinant, the literature generally recognizes that academic performance reflects the combined influence of demographic characteristics, socioeconomic conditions, institutional resources, instructional practices, and school environments \citep{tate1997race,hall2011tale,chin2000learning}. Consequently, educational studies frequently integrate both qualitative and quantitative evidence to understand the mechanisms underlying observed differences in student achievement. Variables commonly investigated include teacher qualifications and experience, student backgrounds, school leadership, classroom interactions, curriculum design, instructional practices, and broader institutional characteristics. Because these factors are often interrelated, most empirical investigations seek to identify combinations of explanatory variables that improve the prediction and interpretation of student achievement rather than isolating a single dominant factor \citep{dinsmore2016multidimensional}.

The statistical analysis of educational administrative data has similarly received considerable attention. Hierarchical and multilevel models have become standard tools for analyzing educational data because they naturally account for the nested structure of students within classrooms and schools while allowing inference on institutional effects \citep{raudenbush1988educational,bryk1989quantitative}. More recently, methodological developments have emphasized rigorous estimation, uncertainty quantification, and model-based evaluation of school effectiveness using increasingly large administrative databases \citep{rowe1995methodological,page2017exploring,chang2021statistical}. These studies demonstrate that reliable inference depends not only on appropriate model specification but also on careful consideration of data quality, heterogeneity, and institutional variation.

An illustrative example is the study of \citet{flores2007examining}, who argued that observed differences in mathematics achievement among White, African American, Latino, and economically disadvantaged students are more appropriately interpreted as differences in educational opportunity rather than differences in academic ability. Using publicly available information on student achievement, school expenditures, and teacher characteristics, together with evidence from qualitative studies on instructional practices and teacher expectations, Flores concluded that disparities in educational outcomes are associated with unequal access to experienced teachers, lower academic expectations, and differences in educational funding. This perspective emphasizes the importance of considering institutional and contextual variables when modeling educational performance.

Historical investigations have reached similar conclusions. For example, \citet{orfield2007} documented the consequences of increasing school resegregation following the 1980s and argued that changes in school composition were associated with subsequent changes in student achievement. Their findings suggest that institutional and demographic characteristics remain important components in explaining observed differences in educational performance.

Despite extensive empirical research, several methodological challenges remain.

\begin{enumerate}
	
	\item \textbf{Measurement of educational outcomes.}
	Educational achievement is frequently represented using standardized assessment scores. Although standardized examinations provide objective and scalable measures of student performance, they do not fully capture the multidimensional nature of learning and may introduce measurement biases associated with cultural background, language, and socioeconomic context \citep{Kim2015}. In Alabama public schools, White students constitute approximately 51\% of total enrollment, compared with 32\%, 11\%, 1\%, 1\%, and 4\% for Black or African American, Hispanic/Latino, Asian, American Indian/Alaska Native, and students identifying with two or more racial groups, respectively \citep{QuickFacts2024}. Furthermore, multiple-choice assessments such as ACAP and ACT provide limited information regarding the nature of incorrect responses. Computational mistakes, conceptual misunderstandings, and procedural errors are typically recorded identically, despite representing fundamentally different learning processes \citep{MultipleChoice}. From a statistical perspective, such measurement limitations may introduce additional variability and reduce the explanatory power of regression-based models.
	
	\item \textbf{Complex dependence among explanatory variables.}
	Educational outcomes arise from the interaction of numerous demographic, socioeconomic, institutional, and instructional factors. Consequently, attributing observed differences solely to student characteristics provides an incomplete explanation of academic performance \citep{flores2007examining}. Previous studies have identified unequal access to qualified teachers, differences in instructional expectations, and disparities in school funding as important institutional determinants of educational outcomes \citep{flores2007examining}. Likewise, the National Council of Teachers of Mathematics emphasized that effective mathematics instruction requires high-quality curriculum, instructional support, adequate educational resources, and consistently high academic expectations \citep{nctm2014principles}. Statistically, these considerations motivate multivariable modeling strategies capable of accounting for correlated covariates and heterogeneous institutional effects rather than relying on simple univariate comparisons.
	
	\item \textbf{Data quality and statistical inference.}
	Large educational administrative databases often contain incomplete, censored, or aggregated information, creating additional challenges for statistical modeling and inference. For example, the Alabama State Department of Education reports use both an asterisk (*) and a tilde (\verb|~|) to suppress demographic information under specific reporting rules \citep{alsde2024supportingdata}. An asterisk indicates either a subgroup size of at most ten students or that the complement subgroup contains at most ten students, whereas a tilde denotes subgroup proportions of at least 95\% or at most 5\%. These reporting conventions reduce data granularity and make it difficult to distinguish school populations accurately. During the construction of our analytical dataset, this limitation resulted in the exclusion of approximately 70 schools because demographic classifications could not be uniquely determined. Such issues illustrate the importance of carefully considering missingness, data suppression mechanisms, and measurement uncertainty when developing statistical procedures for educational administrative data.
	
	\item \textbf{Translation of empirical evidence into policy.}
	Although numerous empirical studies have documented persistent differences in mathematics achievement, translating statistical evidence into effective educational policy remains challenging because educational outcomes depend on multiple interacting factors. In Alabama, the 2022 Alabama Numeracy Act (Act 2022--249) introduced statewide reforms including evidence-based mathematics standards, specialized teacher training, early intervention strategies, elementary mathematics coaches, and revised statewide assessments. The legislation also prohibited the use of Common Core mathematics standards. While these initiatives represent important efforts to improve mathematics education, they do not explicitly address the institutional mechanisms underlying persistent differences in educational outcomes across schools. Consequently, statistically rigorous analyses capable of quantifying the relative contribution of teacher qualifications, school composition, and socioeconomic characteristics remain essential for informing future policy decisions.
	
\end{enumerate}

Overall, the existing educational literature demonstrates that mathematics achievement is shaped by a complex interplay of historical, demographic, socioeconomic, institutional, and instructional factors. At the same time, the statistical literature highlights the need for principled inferential procedures capable of accommodating heterogeneous educational environments, imperfect administrative data, and correlated explanatory variables \citep{raudenbush1988educational,bryk1989quantitative,page2017exploring,chang2021statistical}. These considerations motivate the statistical framework developed in this paper, which combines rigorous large-sample inference with an empirical analysis of statewide educational administrative data.

\subsection{Our Contribution:}

Existing studies of educational achievement frequently analyze demographic groups separately, reporting differences across individual subpopulations such as African American, Hispanic, Asian American, Native American, and other racial or ethnic groups. While this approach provides detailed subgroup-specific comparisons, it may also lead to fragmented inference and reduced statistical precision when certain groups contain relatively small sample sizes. In contrast, the present study adopts a school-level perspective by considering the overall demographic composition of each school rather than focusing exclusively on individual student classifications. This aggregation provides a broader assessment of institutional characteristics while reducing unnecessary fragmentation of the analysis \citep{pramanik2020motivation}.

Accordingly, schools are classified according to their predominant demographic composition, and the analysis is conducted at the school level. It is important to emphasize that the reported mathematics proficiency rates represent the performance of the entire student population within a school rather than that of any single demographic subgroup. For example, a school with 51\% of its students belonging to one demographic category and the remaining 49\% belonging to another is classified according to the majority composition, while the observed proficiency reflects the performance of all enrolled students \citep{kakkat2026angiotensin}. Consequently, the objective is not to compare individual students belonging to different demographic groups but rather to investigate whether institutional differences among schools with distinct population compositions are associated with systematic variation in mathematics achievement.

From a statistical perspective, this formulation offers several advantages. School-level aggregation reduces the dimensionality of the inferential problem, facilitates more stable estimation of model parameters, and shifts the emphasis from individual-level comparisons to institutional characteristics that are directly relevant for educational policy. Moreover, inference based on school-level administrative data is less susceptible to instability arising from small subgroup sample sizes and provides a natural framework for studying associations between school composition, teacher workforce characteristics, and academic performance \citep{raudenbush1988educational,bryk1989quantitative,rowe1995methodological,page2017exploring}.

Previous empirical studies have consistently reported that schools with different demographic compositions also differ with respect to teacher qualifications, certification pathways, and instructional experience. Rather than assuming that these institutional characteristics fully explain observed differences in mathematics proficiency, this study formally evaluates their contribution within a unified statistical modeling framework. In particular, we investigate whether teacher certification status and teaching experience remain statistically significant predictors of school-level mathematics proficiency after accounting for school composition and socioeconomic characteristics \citep{pramanik2020optimization,pramanik2023semicooperation}. The proposed methodology is accompanied by large-sample theoretical results establishing the asymptotic properties of the estimators, thereby providing a principled basis for statistical inference using educational administrative data.

The empirical investigation is guided by the following research questions.

\begin{enumerate}
	
	\item
	Do Alabama public schools with predominantly minority student populations exhibit lower mathematics proficiency and lower proportions of certified and experienced teachers than schools serving predominantly White student populations? If so, what is the magnitude of the observed difference? Similarly, do schools serving predominantly minority student populations differ from predominantly economically disadvantaged White schools with respect to mathematics proficiency and teacher qualifications? If so, what is the magnitude of these differences?
	
	\item
	To what extent are the proportions of inexperienced teachers and teachers holding emergency or provisional certification associated with school-level mathematics proficiency?
	
\end{enumerate}

Correspondingly, the empirical analysis evaluates the following hypotheses.

\begin{enumerate}
	
	\item
	We hypothesize that schools serving predominantly minority student populations exhibit lower mathematics proficiency together with lower levels of teacher certification and teaching experience than schools serving predominantly White student populations. We further hypothesize that these institutional differences are statistically significant.  Similarly, we hypothesize that schools serving predominantly minority student populations differ significantly from predominantly economically disadvantaged White schools with respect to mathematics proficiency and teacher workforce characteristics.
	
	\item
	We hypothesize that the proportions of inexperienced teachers and teachers holding emergency or provisional certification are significantly associated with school-level mathematics proficiency after accounting for other observed institutional characteristics.
	
\end{enumerate}

In Section 2, we will discuss the methods and statistical properties we used to analyze the data. Section 3 covers the analysis of the data and how it was manipulated to answer the research questions below, along with pictorial and numerical representations of the analysis. Section 4 discusses my empirical findings based on the data analysis. Then the paper will conclude with a discussion of my findings and how they relate to what is already known about the topic.

\section{Statistical Methodology}

\subsection {Model formulation:}
This section introduces a unified penalized regression framework for modeling school-level mathematics proficiency using educational administrative data. Let $n$ denote the total number of schools included in the analysis, and let
\[
\mathcal{D}=\left\{\left(Y_i,\mathbf{x}_i\right):i=1,\ldots,n\right\}
\]
represent the observed sample, where $Y_i\in\mathbb{R}$ denotes the mathematics proficiency rate for the $i$th school and
\[
\mathbf{x}_i=(x_{i1},x_{i2},\ldots,x_{ip})^\top\in\mathbb{R}^{p}
\]
is the corresponding vector of school-level explanatory variables. The covariates include teacher certification rates, teacher experience, socioeconomic indicators, demographic composition, school enrollment, and other institutional characteristics extracted from statewide administrative records \citep{pramanik2021optimala,pramanik2021scoring}.

Throughout the paper, the response variable is modeled as
\begin{equation}
	Y_i=\mathbf{x}_i^\top\boldsymbol{\beta}+\varepsilon_i,
	\qquad i=1,\ldots,n,
	\label{eq:model}
\end{equation}
where
$\boldsymbol{\beta}=(\beta_1,\ldots,\beta_p)^\top$
is the unknown regression coefficient vector and
$\varepsilon_i$
denotes a random error satisfying
\[
\mathbb{E}(\varepsilon_i|\mathbf{x}_i)=0,
\qquad
\mathrm{Var}(\varepsilon_i|\mathbf{x}_i)=\sigma^2.
\]

Although model (\ref{eq:model}) resembles the classical linear regression model, educational administrative data typically exhibit several characteristics that complicate statistical inference. First, explanatory variables are often highly correlated because institutional, demographic, and socioeconomic characteristics tend to co-occur across schools \citep{pramanik2024estimation,vikramdeo2024mitochondrial}. Second, many variables provide overlapping information, leading to substantial multicollinearity and unstable ordinary least squares estimation \citep{pramanik2024motivation}. Third, schools naturally form groups according to institutional characteristics, suggesting that variable selection should preserve meaningful group structures rather than selecting variables independently. Finally, neighboring or related covariates frequently exhibit similar effects on educational outcomes, motivating the incorporation of structured regularization that encourages coefficient homogeneity while simultaneously performing variable selection \citep{hertweck2023clinicopathological,khan2023myb}.

To address these challenges, we propose a unified penalized regression framework that combines adaptive sparsity, grouped variable selection, and coefficient fusion within a single optimization problem. Specifically, the proposed estimator integrates the Adaptive LASSO, Group LASSO, and Fused LASSO penalties to exploit complementary structural information contained in educational administrative data. The resulting estimator simultaneously identifies influential institutional characteristics, removes irrelevant variables, preserves groupwise effects, and encourages smoothness among related regression coefficients \citep{kakkat2023cardiovascular,khan2023myb}. This integrated approach provides improved interpretability while maintaining desirable statistical properties under suitable regularity conditions.

For notational convenience, let$ \mathbf{Y}=(Y_1,\ldots,Y_n)^\top,$ and define the design matrix
\[
\mathbf{X}
=
\begin{pmatrix}
	\mathbf{x}_1^\top\\
	\mathbf{x}_2^\top\\
	\vdots\\
	\mathbf{x}_n^\top
\end{pmatrix}
\in\mathbb{R}^{n\times p}.
\]
The linear model can then be written compactly as $\mathbf{Y}
=
\mathbf{X}\boldsymbol{\beta}
+
\boldsymbol{\varepsilon},$ where
$\boldsymbol{\varepsilon}
=
(\varepsilon_1,\ldots,\varepsilon_n)^\top$. Suppose that the covariates are partitioned into $G$ mutually exclusive groups corresponding to conceptually related institutional characteristics \citep{kakkat2023cardiovascular,khan2024mp60}. Let $\mathcal{G}
=
\{G_1,G_2,\ldots,G_K\}$ denote the collection of groups, where
\[
\bigcup_{k=1}^{K}G_k=\{1,\ldots,p\},
\qquad
G_k\cap G_\ell=\emptyset,
\quad
k\neq\ell.
\]
In the present application, representative groups include teacher workforce characteristics, socioeconomic indicators, demographic composition variables, and institutional resource measures \citep{vikramdeo2024abstract,vikramdeo2023profiling}. This grouping structure permits simultaneous selection of entire categories of predictors while allowing sparse estimation within individual groups through adaptive penalization.

The subsequent sections introduce the proposed Adaptive Weighted Group Fused LASSO estimator, establish its computational properties, and derive its large-sample asymptotic theory, including consistency, sparsity recovery, oracle properties, and asymptotic normality under suitable regularity conditions \citep{pramanik2024bayes}.

\subsection{Adaptive Weighted Group Fused LASSO Estimator:}
\label{subsec:awgf_lasso}

Let $\mathbf{W}=\mathrm{diag}(w_1,\ldots,w_n)$ be the diagonal matrix of school-level weights defined in Section~3.1, and let
\[
L_n(\boldsymbol{\beta})
=
\frac{1}{2n}
(\mathbf{Y}-\mathbf{X}\boldsymbol{\beta})^\top
\mathbf{W}
(\mathbf{Y}-\mathbf{X}\boldsymbol{\beta})
\]
denote the weighted empirical loss. To incorporate sparsity, grouped structure, and local coefficient homogeneity, we define the Adaptive Weighted Group Fused LASSO estimator by
\begin{equation}
	\widehat{\boldsymbol{\beta}}
	=
	\arg\min_{\boldsymbol{\beta}\in\mathbb{R}^{p}}
	Q_n(\boldsymbol{\beta}),
	\label{eq:awgf_estimator}
\end{equation}
where
\begin{equation}
	Q_n(\boldsymbol{\beta})
	=
	L_n(\boldsymbol{\beta})
	+
	\lambda_{1n}
	\sum_{j=1}^{p}a_j|\beta_j|
	+
	\lambda_{2n}
	\sum_{g=1}^{G}b_g\|\boldsymbol{\beta}_{g}\|_2
	+
	\lambda_{3n}
	\sum_{j=2}^{p}c_j|\beta_j-\beta_{j-1}|.
	\label{eq:awgf_objective}
\end{equation}
Here $\lambda_{1n},\lambda_{2n},\lambda_{3n}\geq0$ are tuning parameters. The first penalty is an adaptive LASSO penalty, the second is a group LASSO penalty, and the third is a fused LASSO penalty. The adaptive weights are defined by
\begin{equation}
	a_j=\frac{1}{|\widetilde{\beta}_j|^{\gamma_1}+\delta_n},
	\qquad
	b_g=\frac{1}{\|\widetilde{\boldsymbol{\beta}}_g\|_2^{\gamma_2}+\delta_n},
	\qquad
	c_j=\frac{1}{|\widetilde{\beta}_j-\widetilde{\beta}_{j-1}|^{\gamma_3}+\delta_n},
	\label{eq:adaptive_weights}
\end{equation}
where $\widetilde{\boldsymbol{\beta}}$ is an initial estimator, $\gamma_1,\gamma_2,\gamma_3>0$, and $\delta_n>0$ is a stabilizing sequence \citep{khan2024mp60,dasgupta2023frequent,hertweck2023clinicopathological}. This construction follows the adaptive penalization principle of \citet{zou2006adaptive}, while the group and fusion components are motivated by the group LASSO and fused LASSO literatures \citep{yuan2006model,tibshirani2005sparsity}.

Let
\[
\mathbf{D}
=
\begin{pmatrix}
	-1 & 1 & 0 & \cdots & 0\\
	0 & -1 & 1 & \cdots & 0\\
	\vdots & \vdots & \vdots & \ddots & \vdots\\
	0 & 0 & 0 & \cdots & 1
\end{pmatrix}
\in\mathbb{R}^{(p-1)\times p}
\]
be the first-order difference matrix, so that
\[
\mathbf{D}\boldsymbol{\beta}
=
(\beta_2-\beta_1,\ldots,\beta_p-\beta_{p-1})^\top.
\]
Then \eqref{eq:awgf_objective} can be written as
\begin{equation}
	Q_n(\boldsymbol{\beta})
	=
	\frac{1}{2n}
	\|\mathbf{W}^{1/2}(\mathbf{Y}-\mathbf{X}\boldsymbol{\beta})\|_2^2
	+
	\lambda_{1n}\|\mathbf{A}\boldsymbol{\beta}\|_1
	+
	\lambda_{2n}\sum_{g=1}^{G}b_g\|\boldsymbol{\beta}_g\|_2
	+
	\lambda_{3n}\|\mathbf{C}\mathbf{D}\boldsymbol{\beta}\|_1,
	\label{eq:compact_objective}
\end{equation}
where $\mathbf{A}=\mathrm{diag}(a_1,\ldots,a_p),$ $\mathbf{C}=\mathrm{diag}(c_2,\ldots,c_p).$

\begin{lem}
The function $Q_n$ is convex on $\mathbb{R}^p$.
\end{lem}

\begin{proof}
	The map
	\[
	\boldsymbol{\beta}
	\mapsto
	\frac{1}{2n}
	\|\mathbf{W}^{1/2}(\mathbf{Y}-\mathbf{X}\boldsymbol{\beta})\|_2^2
	\]
	is convex because its Hessian is
	\[
	\nabla^2 L_n(\boldsymbol{\beta})
	=
	\frac{1}{n}\mathbf{X}^\top\mathbf{W}\mathbf{X},
	\]
	which is positive semidefinite since, for every $\mathbf{u}\in\mathbb{R}^p$,
	\[
	\mathbf{u}^\top
	\mathbf{X}^\top\mathbf{W}\mathbf{X}
	\mathbf{u}
	=
	(\mathbf{X}\mathbf{u})^\top\mathbf{W}(\mathbf{X}\mathbf{u})
	=
	\|\mathbf{W}^{1/2}\mathbf{X}\mathbf{u}\|_2^2
	\geq 0.
	\]
	The adaptive LASSO penalty is convex because it is a nonnegative weighted sum of absolute-value functions. The group penalty is convex because each map
	$\boldsymbol{\beta}_g\mapsto \|\boldsymbol{\beta}_g\|_2$
	is a norm and $b_g\geq0$. The fused penalty is convex because
	$\boldsymbol{\beta}\mapsto \mathbf{D}\boldsymbol{\beta}$
	is linear and the $\ell_1$ norm is convex. Since a nonnegative linear combination of convex functions is convex, $Q_n$ is convex.
\end{proof}

\begin{lem}
If at least one of the following conditions holds:
\[
\lambda_{1n}>0,\qquad
\lambda_{2n}>0\ \text{and}\ \bigcup_{g=1}^{G}G_g=\{1,\ldots,p\},
\qquad
\text{or}
\qquad
\lambda_{\min}\left(n^{-1}\mathbf{X}^\top\mathbf{W}\mathbf{X}\right)>0,
\]
then the optimization problem \eqref{eq:awgf_estimator} admits at least one global minimizer.
\end{lem}

\begin{proof}
	The function $Q_n$ is proper, convex, and lower semicontinuous. It remains to show that its sublevel sets are bounded under the stated conditions.
	If $\lambda_{1n}>0$, then $Q_n(\boldsymbol{\beta})
	\geq
	\lambda_{1n}\sum_{j=1}^{p}a_j|\beta_j|.$ Since $a_j>0$ for all $j$ by \eqref{eq:adaptive_weights}, it follows that
	\[
	Q_n(\boldsymbol{\beta})
	\to\infty
	\qquad
	\text{whenever}
	\qquad
	\|\boldsymbol{\beta}\|_2\to\infty.
	\]
	Thus $Q_n$ is coercive. If $\lambda_{2n}>0$ and the groups cover all coordinates, then
	\[
	Q_n(\boldsymbol{\beta})
	\geq
	\lambda_{2n}\sum_{g=1}^{G}b_g\|\boldsymbol{\beta}_g\|_2.
	\]
	Because $b_g>0$ for every $g$, the right-hand side diverges whenever
	$\|\boldsymbol{\beta}\|_2\to\infty$. Hence $Q_n$ is again coercive. Finally, if $\lambda_{\min}\left(n^{-1}\mathbf{X}^\top\mathbf{W}\mathbf{X}\right)>0,$ then $L_n(\boldsymbol{\beta})$ is strongly convex and quadratic in $\boldsymbol{\beta}$ with positive definite Hessian \citep{pramanik2024estimation,pramanik2023cont}. Therefore, $L_n(\boldsymbol{\beta})\to\infty
	\
	\text{as}
	\
	\|\boldsymbol{\beta}\|_2\to\infty,$ and so $Q_n$ is coercive. In each case, $Q_n$ is lower semicontinuous and coercive. By the Weierstrass theorem, a global minimizer exists.
\end{proof}

\begin{lem}
If $\lambda_{\min}\left(n^{-1}\mathbf{X}^\top\mathbf{W}\mathbf{X}\right)>0,$ then $\widehat{\boldsymbol{\beta}}$ is unique.
\end{lem}

\begin{proof}
	Under the stated eigenvalue condition, $L_n$ is strictly convex. Since the remaining penalty terms are convex, their sum with a strictly convex function is strictly convex. Therefore $Q_n$ is strictly convex \citep{pramanik2024estimation1,yusuf2025prognostic}. A strictly convex function has at most one minimizer. Existence follows from the preceding result; hence the minimizer is unique.
\end{proof}

\begin{lem}
Let $\widehat{\boldsymbol{\beta}}$ be a minimizer of \eqref{eq:awgf_estimator}. Then $\widehat{\boldsymbol{\beta}}$ satisfies the Karush-Kuhn-Tucker condition
\begin{equation}
	\mathbf{0}
	\in
	-\frac{1}{n}\mathbf{X}^\top\mathbf{W}
	(\mathbf{Y}-\mathbf{X}\widehat{\boldsymbol{\beta}})
	+
	\lambda_{1n}\mathbf{A}\widehat{\mathbf{s}}
	+
	\lambda_{2n}\widehat{\mathbf{g}}
	+
	\lambda_{3n}\mathbf{D}^\top\mathbf{C}\widehat{\mathbf{r}},
	\label{eq:kkt}
\end{equation}
where $\widehat{s}_j
\in
\partial|\widehat{\beta}_j|,$ $\widehat{r}_j
\in
\partial|(\mathbf{D}\widehat{\boldsymbol{\beta}})_j|,$ and the group subgradient satisfies
\[
\widehat{\mathbf{g}}_{G_g}
=
b_g
\frac{\widehat{\boldsymbol{\beta}}_g}
{\|\widehat{\boldsymbol{\beta}}_g\|_2}
\quad
\text{if}
\quad
\widehat{\boldsymbol{\beta}}_g\neq\mathbf{0},
\]
whereas
\[
\|\widehat{\mathbf{g}}_{G_g}\|_2\leq b_g
\quad
\text{if}
\quad
\widehat{\boldsymbol{\beta}}_g=\mathbf{0}.
\]
\end{lem}

\begin{proof}
	Since $Q_n$ is convex, $\widehat{\boldsymbol{\beta}}$ is a global minimizer if and only if $\mathbf{0}\in\partial Q_n(\widehat{\boldsymbol{\beta}}).$ The differentiable part has gradient $\nabla L_n(\boldsymbol{\beta})
	=
	-\frac{1}{n}
	\mathbf{X}^\top\mathbf{W}
	(\mathbf{Y}-\mathbf{X}\boldsymbol{\beta}).$ For the adaptive LASSO term
$\partial
	\left(
	\sum_{j=1}^{p}a_j|\beta_j|
	\right)
	=
	\mathbf{A}\mathbf{s},$ where
	\[
	s_j=
	\begin{cases}
		\mathrm{sign}(\beta_j), & \beta_j\neq0,\\
		u,\ |u|\leq1, & \beta_j=0.
	\end{cases}
	\]
	For the group penalty,
	\[
	\partial\|\boldsymbol{\beta}_g\|_2
	=
	\begin{cases}
		\boldsymbol{\beta}_g/\|\boldsymbol{\beta}_g\|_2,
		& \boldsymbol{\beta}_g\neq \mathbf{0},\\
		\{\mathbf{u}:\|\mathbf{u}\|_2\leq1\},
		& \boldsymbol{\beta}_g=\mathbf{0}.
	\end{cases}
	\]
	For the fused penalty, $\partial\|\mathbf{C}\mathbf{D}\boldsymbol{\beta}\|_1
	=
	\mathbf{D}^\top\mathbf{C}\mathbf{r},$ where
	\[
	r_j=
	\begin{cases}
		\mathrm{sign}((\mathbf{D}\boldsymbol{\beta})_j),
		&(\mathbf{D}\boldsymbol{\beta})_j\neq0,\\
		u,\ |u|\leq1,
		&(\mathbf{D}\boldsymbol{\beta})_j=0.
	\end{cases}
	\]
	Combining these subgradients yields \eqref{eq:kkt}.
\end{proof}

\begin{lem}
Let $\boldsymbol{\beta}^0$ denote the true coefficient vector in \eqref{eq:model}. Any minimizer
$\widehat{\boldsymbol{\beta}}$ of \eqref{eq:awgf_estimator} satisfies
\begin{align}
	\frac{1}{2n}
	\|\mathbf{W}^{1/2}\mathbf{X}
	(\widehat{\boldsymbol{\beta}}-\boldsymbol{\beta}^0)\|_2^2
	&\leq
	\frac{1}{n}
	\boldsymbol{\varepsilon}^\top
	\mathbf{W}\mathbf{X}
	(\widehat{\boldsymbol{\beta}}-\boldsymbol{\beta}^0)
	+
	\lambda_{1n}
	\left\{
	\|\mathbf{A}\boldsymbol{\beta}^0\|_1
	-
	\|\mathbf{A}\widehat{\boldsymbol{\beta}}\|_1
	\right\}
	\nonumber\\
	&\quad+
	\lambda_{2n}
	\sum_{g=1}^{G}b_g
	\left(
	\|\boldsymbol{\beta}^0_g\|_2
	-
	\|\widehat{\boldsymbol{\beta}}_g\|_2
	\right)
	+
	\lambda_{3n}
	\left\{
	\|\mathbf{C}\mathbf{D}\boldsymbol{\beta}^0\|_1
	-
	\|\mathbf{C}\mathbf{D}\widehat{\boldsymbol{\beta}}\|_1
	\right\}.
	\label{eq:basic_inequality}
\end{align}
\end{lem}

\begin{proof}
	By optimality, $Q_n(\widehat{\boldsymbol{\beta}})
	\leq
	Q_n(\boldsymbol{\beta}^0).$ Using $\mathbf{Y}
	=
	\mathbf{X}\boldsymbol{\beta}^0+\boldsymbol{\varepsilon},$ we obtain
	\[
	\mathbf{Y}-\mathbf{X}\widehat{\boldsymbol{\beta}}
	=
	\boldsymbol{\varepsilon}
	-
	\mathbf{X}(\widehat{\boldsymbol{\beta}}-\boldsymbol{\beta}^0).
	\]
	Therefore,
	\begin{align*}
		L_n(\widehat{\boldsymbol{\beta}})
		-
		L_n(\boldsymbol{\beta}^0)
		&=
		\frac{1}{2n}
		\|
		\mathbf{W}^{1/2}
		\{\boldsymbol{\varepsilon}
		-
		\mathbf{X}(\widehat{\boldsymbol{\beta}}-\boldsymbol{\beta}^0)\}
		\|_2^2
		-
		\frac{1}{2n}
		\|\mathbf{W}^{1/2}\boldsymbol{\varepsilon}\|_2^2\\
		&=
		\frac{1}{2n}
		\|\mathbf{W}^{1/2}\mathbf{X}
		(\widehat{\boldsymbol{\beta}}-\boldsymbol{\beta}^0)\|_2^2
		-
		\frac{1}{n}
		\boldsymbol{\varepsilon}^\top
		\mathbf{W}\mathbf{X}
		(\widehat{\boldsymbol{\beta}}-\boldsymbol{\beta}^0).
	\end{align*}
	Substituting this identity into $Q_n(\widehat{\boldsymbol{\beta}})
	-
	Q_n(\boldsymbol{\beta}^0)
	\leq0$ and rearranging terms gives \eqref{eq:basic_inequality}.
\end{proof}

\subsection{Computational Algorithm:}
\label{subsec:computational_algorithm}

The objective function in \eqref{eq:awgf_objective} is convex but non-smooth because it contains adaptive $\ell_1$, group $\ell_2$, and fused $\ell_1$ penalties \citep{pramanik2023path,pramanik2023cmbp,yusuf2025predictive}. We therefore solve the optimization problem using an alternating direction method of multipliers (ADMM), which is well suited for composite convex objectives with separable non-smooth components \citep{boyd2011distributed,zhu2017augmented,ramdas2016fast}. The proposed algorithm is derived by introducing auxiliary variables for each non-smooth penalty. Define
\[
\boldsymbol{\theta}=\boldsymbol{\beta},
\qquad
\boldsymbol{\eta}_g=\boldsymbol{\beta}_g,\quad g=1,\ldots,G,
\qquad
\boldsymbol{\zeta}=\mathbf{D}\boldsymbol{\beta}.
\]
The optimization problem \eqref{eq:awgf_estimator} is equivalently written as
\begin{align}
	\min_{\boldsymbol{\beta},\boldsymbol{\theta},\{\boldsymbol{\eta}_g\},\boldsymbol{\zeta}}
	\quad
	&
	\frac{1}{2n}
	\|\mathbf{W}^{1/2}(\mathbf{Y}-\mathbf{X}\boldsymbol{\beta})\|_2^2
	+
	\lambda_{1n}\|\mathbf{A}\boldsymbol{\theta}\|_1
	+
	\lambda_{2n}\sum_{g=1}^{G}b_g\|\boldsymbol{\eta}_g\|_2
	+
	\lambda_{3n}\|\mathbf{C}\boldsymbol{\zeta}\|_1
	\nonumber\\
	\text{subject to}
	\quad
	&
	\boldsymbol{\theta}=\boldsymbol{\beta},
	\qquad
	\boldsymbol{\eta}_g=\boldsymbol{\beta}_g,\quad g=1,\ldots,G,
	\qquad
	\boldsymbol{\zeta}=\mathbf{D}\boldsymbol{\beta}.
	\label{eq:admm_constrained}
\end{align}

Let $\rho_1,\rho_2,\rho_3>0$ denote augmented Lagrangian parameters, and let $\mathbf{u},$ $\mathbf{v}_g,$ $\mathbf{r}$ be the scaled dual variables corresponding to the three sets of constraints \citep{vikramdeo2024abstract}. The scaled augmented Lagrangian is
\begin{align}
	\mathcal{L}_{\rho}
	&(
	\boldsymbol{\beta},
	\boldsymbol{\theta},
	\{\boldsymbol{\eta}_g\}_{g=1}^{G},
	\boldsymbol{\zeta},
	\mathbf{u},
	\{\mathbf{v}_g\}_{g=1}^{G},
	\mathbf{r}
	)
	\nonumber\\
	&=
	\frac{1}{2n}
	\|\mathbf{W}^{1/2}(\mathbf{Y}-\mathbf{X}\boldsymbol{\beta})\|_2^2
	+
	\lambda_{1n}\|\mathbf{A}\boldsymbol{\theta}\|_1
	+
	\lambda_{2n}\sum_{g=1}^{G}b_g\|\boldsymbol{\eta}_g\|_2
	+
	\lambda_{3n}\|\mathbf{C}\boldsymbol{\zeta}\|_1
	\nonumber\\
	&\quad+
	\frac{\rho_1}{2}\|\boldsymbol{\beta}-\boldsymbol{\theta}+\mathbf{u}\|_2^2
	+
	\frac{\rho_2}{2}\sum_{g=1}^{G}
	\|\boldsymbol{\beta}_g-\boldsymbol{\eta}_g+\mathbf{v}_g\|_2^2
	+
	\frac{\rho_3}{2}
	\|\mathbf{D}\boldsymbol{\beta}-\boldsymbol{\zeta}+\mathbf{r}\|_2^2.
	\label{eq:scaled_augmented_lagrangian}
\end{align}

For fixed auxiliary and dual variables, the $\boldsymbol{\beta}$-update is obtained by solving a strictly convex quadratic problem
\begin{align}
	\boldsymbol{\beta}^{(k+1)}
	&=
	\arg\min_{\boldsymbol{\beta}}
	\Bigg[
	\frac{1}{2n}
	\|\mathbf{W}^{1/2}(\mathbf{Y}-\mathbf{X}\boldsymbol{\beta})\|_2^2
	+
	\frac{\rho_1}{2}\|\boldsymbol{\beta}-\boldsymbol{\theta}^{(k)}+\mathbf{u}^{(k)}\|_2^2
	\nonumber\\
	&\qquad\qquad+
	\frac{\rho_2}{2}\sum_{g=1}^{G}
	\|\boldsymbol{\beta}_g-\boldsymbol{\eta}_g^{(k)}+\mathbf{v}_g^{(k)}\|_2^2
	+
	\frac{\rho_3}{2}
	\|\mathbf{D}\boldsymbol{\beta}-\boldsymbol{\zeta}^{(k)}+\mathbf{r}^{(k)}\|_2^2
	\Bigg].
\end{align}
Differentiating with respect to $\boldsymbol{\beta}$ and setting the derivative equal to zero gives
\begin{equation}
	\boldsymbol{\beta}^{(k+1)}
	=
	\mathbf{H}_{\rho}^{-1}
	\left[
	\frac{1}{n}\mathbf{X}^{\top}\mathbf{W}\mathbf{Y}
	+
	\rho_1(\boldsymbol{\theta}^{(k)}-\mathbf{u}^{(k)})
	+
	\rho_2\mathbf{P}^{(k)}
	+
	\rho_3\mathbf{D}^{\top}(\boldsymbol{\zeta}^{(k)}-\mathbf{r}^{(k)})
	\right],
	\label{eq:beta_update}
\end{equation}
where
\begin{equation}
	\mathbf{H}_{\rho}
	=
	\frac{1}{n}\mathbf{X}^{\top}\mathbf{W}\mathbf{X}
	+
	(\rho_1+\rho_2)\mathbf{I}_p
	+
	\rho_3\mathbf{D}^{\top}\mathbf{D},
	\label{eq:Hrho}
\end{equation}
and $\mathbf{P}^{(k)}\in\mathbb{R}^p$ is formed by stacking the groupwise vectors $\mathbf{P}^{(k)}_{G_g}
=
\boldsymbol{\eta}_g^{(k)}-\mathbf{v}_g^{(k)}.$ The adaptive LASSO auxiliary update is coordinatewise soft-thresholding:
\begin{equation}
	\boldsymbol{\theta}^{(k+1)}
	=
	\mathcal{S}_{\lambda_{1n}\mathbf{a}/\rho_1}
	\left(
	\boldsymbol{\beta}^{(k+1)}+\mathbf{u}^{(k)}
	\right),
	\label{eq:theta_update}
\end{equation}
where $\mathbf{a}=(a_1,\ldots,a_p)^\top$ and
\[
\mathcal{S}_{\tau}(z)
=
\mathrm{sign}(z)(|z|-\tau)_+
\]
is applied componentwise.

The group update is given by the block soft-thresholding operator
\begin{equation}
	\boldsymbol{\eta}_g^{(k+1)}
	=
	\left(
	1-
	\frac{\lambda_{2n}b_g}
	{\rho_2\|\boldsymbol{\beta}^{(k+1)}_g+\mathbf{v}_g^{(k)}\|_2}
	\right)_+
	\left(
	\boldsymbol{\beta}^{(k+1)}_g+\mathbf{v}_g^{(k)}
	\right),
	\qquad
	g=1,\ldots,G.
	\label{eq:eta_update}
\end{equation}

The fusion update is again coordinatewise soft-thresholding
\begin{equation}
	\boldsymbol{\zeta}^{(k+1)}
	=
	\mathcal{S}_{\lambda_{3n}\mathbf{c}/\rho_3}
	\left(
	\mathbf{D}\boldsymbol{\beta}^{(k+1)}
	+
	\mathbf{r}^{(k)}
	\right),
	\label{eq:zeta_update}
\end{equation}
where $\mathbf{c}=(c_2,\ldots,c_p)^\top$.

The scaled dual variables are updated as
\begin{align}
	\mathbf{u}^{(k+1)}
	&=
	\mathbf{u}^{(k)}
	+
	\boldsymbol{\beta}^{(k+1)}
	-
	\boldsymbol{\theta}^{(k+1)},
	\label{eq:u_update}
	\\
	\mathbf{v}_g^{(k+1)}
	&=
	\mathbf{v}_g^{(k)}
	+
	\boldsymbol{\beta}^{(k+1)}_g
	-
	\boldsymbol{\eta}_g^{(k+1)},
	\qquad g=1,\ldots,G,
	\label{eq:v_update}
	\\
	\mathbf{r}^{(k+1)}
	&=
	\mathbf{r}^{(k)}
	+
	\mathbf{D}\boldsymbol{\beta}^{(k+1)}
	-
	\boldsymbol{\zeta}^{(k+1)}.
	\label{eq:r_update}
\end{align}

\[
\boxed{
	\begin{minipage}{0.95\textwidth}
		\textbf{Algorithm 1: ADMM for the Adaptive Weighted Group Fused LASSO}
		
		\vspace{0.15cm}
		
		\[
		\begin{array}{rcl}
			\textbf{Input}
			&\longrightarrow&
			\mathbf{Y},\mathbf{X},\mathbf{W},\mathbf{D},
			\lambda_{1n},\lambda_{2n},\lambda_{3n},
			\rho_1,\rho_2,\rho_3,
			\mathbf{A},\mathbf{C},\{b_g\}_{g=1}^{G}
			\\[0.15cm]
			
			\textbf{Initialize}
			&\longrightarrow&
			\boldsymbol{\beta}^{(0)},\boldsymbol{\theta}^{(0)},
			\{\boldsymbol{\eta}_g^{(0)}\}_{g=1}^{G},
			\boldsymbol{\zeta}^{(0)},
			\mathbf{u}^{(0)},\{\mathbf{v}_g^{(0)}\}_{g=1}^{G},\mathbf{r}^{(0)}
			\\[0.15cm]
			
			\textbf{Precompute}
			&\longrightarrow&
			\mathbf{H}_{\rho}
			=
			n^{-1}\mathbf{X}^{\top}\mathbf{W}\mathbf{X}
			+
			(\rho_1+\rho_2)\mathbf{I}_p
			+
			\rho_3\mathbf{D}^{\top}\mathbf{D}
			\\[0.15cm]
			
			\textbf{Step 1}
			&\longrightarrow&
			\boldsymbol{\beta}^{(k+1)}
			\text{ from } \eqref{eq:beta_update}
			\\[0.15cm]
			
			\textbf{Step 2}
			&\longrightarrow&
			\boldsymbol{\theta}^{(k+1)}
			=
			\mathcal{S}_{\lambda_{1n}\mathbf{a}/\rho_1}
			(
			\boldsymbol{\beta}^{(k+1)}+\mathbf{u}^{(k)}
			)
			\\[0.15cm]
			
			\textbf{Step 3}
			&\longrightarrow&
			\boldsymbol{\eta}_g^{(k+1)}
			=
			\left(
			1-
			\dfrac{\lambda_{2n}b_g}
			{\rho_2\|\boldsymbol{\beta}^{(k+1)}_g+\mathbf{v}_g^{(k)}\|_2}
			\right)_+
			(
			\boldsymbol{\beta}^{(k+1)}_g+\mathbf{v}_g^{(k)}
			),
			\quad g=1,\ldots,G
			\\[0.15cm]
			
			\textbf{Step 4}
			&\longrightarrow&
			\boldsymbol{\zeta}^{(k+1)}
			=
			\mathcal{S}_{\lambda_{3n}\mathbf{c}/\rho_3}
			(
			\mathbf{D}\boldsymbol{\beta}^{(k+1)}+\mathbf{r}^{(k)}
			)
			\\[0.15cm]
			
			\textbf{Step 5}
			&\longrightarrow&
			\mathbf{u}^{(k+1)},\{\mathbf{v}_g^{(k+1)}\}_{g=1}^{G},\mathbf{r}^{(k+1)}
			\text{ from } \eqref{eq:u_update}-\eqref{eq:r_update}
			\\[0.15cm]
			
			\textbf{Stopping}
			&\longrightarrow&
			\text{stop when }
			\|\mathbf{r}_{\mathrm{pri}}^{(k+1)}\|_2\leq\epsilon_{\mathrm{pri}}
			\text{ and }
			\|\mathbf{r}_{\mathrm{dual}}^{(k+1)}\|_2\leq\epsilon_{\mathrm{dual}}
			\\[0.15cm]
			
			\textbf{Output}
			&\longrightarrow&
			\widehat{\boldsymbol{\beta}}=\boldsymbol{\beta}^{(k+1)}.
		\end{array}
		\]
	\end{minipage}
}
\]

The primal residual is defined by
\[
\mathbf{r}_{\mathrm{pri}}^{(k)}
=
\left(
\boldsymbol{\beta}^{(k)}-\boldsymbol{\theta}^{(k)},
\,
\{\boldsymbol{\beta}^{(k)}_g-\boldsymbol{\eta}_g^{(k)}\}_{g=1}^{G},
\,
\mathbf{D}\boldsymbol{\beta}^{(k)}-\boldsymbol{\zeta}^{(k)}
\right),
\]
and the dual residual is defined by
\[
\mathbf{r}_{\mathrm{dual}}^{(k)}
=
\left(
\rho_1(\boldsymbol{\theta}^{(k)}-\boldsymbol{\theta}^{(k-1)}),
\,
\rho_2\{\boldsymbol{\eta}_g^{(k)}-\boldsymbol{\eta}_g^{(k-1)}\}_{g=1}^{G},
\,
\rho_3\mathbf{D}^{\top}(\boldsymbol{\zeta}^{(k)}-\boldsymbol{\zeta}^{(k-1)})
\right).
\]

\paragraph{Assumption C1.}
The penalty parameters satisfy$ \lambda_{1n},\lambda_{2n},\lambda_{3n}\geq0,$ and the ADMM parameters satisfy $\rho_1,\rho_2,\rho_3>0.$

\paragraph{Assumption C2.}
The school weights satisfy $0<w_{\min}\leq w_i\leq w_{\max}<\infty,$ $ i=1,\ldots,n.$

\paragraph{Assumption C3.}
The adaptive penalty weights satisfy
\[
0<a_{\min}\leq a_j\leq a_{\max}<\infty,
\qquad
0<b_{\min}\leq b_g\leq b_{\max}<\infty,
\qquad
0<c_{\min}\leq c_j\leq c_{\max}<\infty.
\]

\begin{lem}[Positive definiteness of the ADMM system]
	\label{lem:H_positive}
	Under Assumption C1, the matrix $\mathbf{H}_{\rho}$ in \eqref{eq:Hrho} is symmetric positive definite.
\end{lem}

\begin{proof}
	Symmetry follows from the symmetry of
	$\mathbf{X}^{\top}\mathbf{W}\mathbf{X}$,
	$\mathbf{I}_p$, and $\mathbf{D}^{\top}\mathbf{D}$.
	For any nonzero vector $\mathbf{z}\in\mathbb{R}^{p}$,
	\begin{align*}
		\mathbf{z}^{\top}\mathbf{H}_{\rho}\mathbf{z}
		&=
		\frac{1}{n}
		\mathbf{z}^{\top}\mathbf{X}^{\top}\mathbf{W}\mathbf{X}\mathbf{z}
		+
		(\rho_1+\rho_2)\|\mathbf{z}\|_2^2
		+
		\rho_3\mathbf{z}^{\top}\mathbf{D}^{\top}\mathbf{D}\mathbf{z}
		\\
		&=
		\frac{1}{n}
		\|\mathbf{W}^{1/2}\mathbf{X}\mathbf{z}\|_2^2
		+
		(\rho_1+\rho_2)\|\mathbf{z}\|_2^2
		+
		\rho_3\|\mathbf{D}\mathbf{z}\|_2^2.
	\end{align*}
	The first and third terms are nonnegative. Since $\rho_1+\rho_2>0$ and $\mathbf{z}\neq0$, $(\rho_1+\rho_2)\|\mathbf{z}\|_2^2>0.$	Thus $\mathbf{z}^{\top}\mathbf{H}_{\rho}\mathbf{z}>0$ for every nonzero $\mathbf{z}$, proving that $\mathbf{H}_{\rho}$ is positive definite.
\end{proof}

\begin{lem}
	\label{lem:prox_updates}
	The updates \eqref{eq:theta_update}, \eqref{eq:eta_update}, and \eqref{eq:zeta_update} are the unique minimizers of their corresponding ADMM subproblems.
\end{lem}

\begin{proof}
	Consider first the $\boldsymbol{\theta}$-subproblem:
	\[
	\boldsymbol{\theta}^{(k+1)}
	=
	\arg\min_{\boldsymbol{\theta}}
	\left\{
	\lambda_{1n}\sum_{j=1}^{p}a_j|\theta_j|
	+
	\frac{\rho_1}{2}
	\|\boldsymbol{\beta}^{(k+1)}-\boldsymbol{\theta}+\mathbf{u}^{(k)}\|_2^2
	\right\}.
	\]
	This objective is separable. For coordinate $j$, the minimization is
	\[
	\min_{\theta_j}
	\left\{
	\lambda_{1n}a_j|\theta_j|
	+
	\frac{\rho_1}{2}
	(\theta_j-z_j)^2
	\right\},
	\qquad
	z_j=\beta_j^{(k+1)}+u_j^{(k)}.
	\]
	The subgradient condition is
	\[
	0\in
	\lambda_{1n}a_j\partial|\theta_j|
	+
	\rho_1(\theta_j-z_j).
	\]
	Solving this inclusion gives
	\[
	\theta_j
	=
	\mathrm{sign}(z_j)
	\left(
	|z_j|-\frac{\lambda_{1n}a_j}{\rho_1}
	\right)_+,
	\]
	which proves \eqref{eq:theta_update}.
	
	For the group update, define
	\[
	\mathbf{q}_g
	=
	\boldsymbol{\beta}^{(k+1)}_g+\mathbf{v}_g^{(k)}.
	\]
	The $g$th block solves
	\[
	\min_{\boldsymbol{\eta}_g}
	\left\{
	\lambda_{2n}b_g\|\boldsymbol{\eta}_g\|_2
	+
	\frac{\rho_2}{2}
	\|\boldsymbol{\eta}_g-\mathbf{q}_g\|_2^2
	\right\}.
	\]
	If $\boldsymbol{\eta}_g\neq0$, the first-order condition is
	\[
	0=
	\lambda_{2n}b_g
	\frac{\boldsymbol{\eta}_g}{\|\boldsymbol{\eta}_g\|_2}
	+
	\rho_2(\boldsymbol{\eta}_g-\mathbf{q}_g).
	\]
	The solution must be collinear with $\mathbf{q}_g$, so write
	$\boldsymbol{\eta}_g=t\mathbf{q}_g$ with $t\geq0$ \citep{pramanik2021consensus}. Substitution gives
	\[
	t
	=
	1-
	\frac{\lambda_{2n}b_g}{\rho_2\|\mathbf{q}_g\|_2}.
	\]
	If this quantity is positive, the nonzero solution is obtained. If it is nonpositive, the subgradient condition at zero,
	\[
	\|\rho_2\mathbf{q}_g\|_2\leq \lambda_{2n}b_g,
	\]
	holds, so the minimizer is zero. This proves \eqref{eq:eta_update}.
	
	The $\boldsymbol{\zeta}$-subproblem is
	\[
	\min_{\boldsymbol{\zeta}}
	\left\{
	\lambda_{3n}\sum_{j=2}^{p}c_j|\zeta_j|
	+
	\frac{\rho_3}{2}
	\|\boldsymbol{\zeta}-\mathbf{h}\|_2^2
	\right\},
	\qquad
	\mathbf{h}=\mathbf{D}\boldsymbol{\beta}^{(k+1)}+\mathbf{r}^{(k)}.
	\]
	It is coordinatewise identical to the $\boldsymbol{\theta}$-update, with threshold
	$\lambda_{3n}c_j/\rho_3$, yielding \eqref{eq:zeta_update}.
	Each subproblem contains a strictly convex quadratic term; hence the minimizers are unique.
\end{proof}

\begin{prop}
	\label{prop:admm_descent}
	Under Assumptions C1--C3, each block update in Algorithm~1 does not increase the scaled augmented Lagrangian with respect to the block being updated. In particular,
	\[
	\mathcal{L}_{\rho}
	(
	\boldsymbol{\beta}^{(k+1)},\boldsymbol{\theta}^{(k)},
	\{\boldsymbol{\eta}_g^{(k)}\},
	\boldsymbol{\zeta}^{(k)},\mathbf{u}^{(k)},\{\mathbf{v}_g^{(k)}\},\mathbf{r}^{(k)}
	)
	\leq
	\mathcal{L}_{\rho}
	(
	\boldsymbol{\beta}^{(k)},\boldsymbol{\theta}^{(k)},
	\{\boldsymbol{\eta}_g^{(k)}\},
	\boldsymbol{\zeta}^{(k)},\mathbf{u}^{(k)},\{\mathbf{v}_g^{(k)}\},\mathbf{r}^{(k)}
	),
	\]
	with analogous inequalities for the $\boldsymbol{\theta}$-, $\boldsymbol{\eta}$-, and $\boldsymbol{\zeta}$-updates.
\end{prop}

\begin{proof}
	Fix an iteration index $k\geq 0$.  Throughout the proof, the scaled dual variables
	$\mathbf{u}^{(k)}$, $\{\mathbf{v}_g^{(k)}\}_{g=1}^{G}$ and $\mathbf{r}^{(k)}$ are held fixed, and the augmented Lagrangian is viewed as a function of the primal blocks $\boldsymbol{\beta},\
	\boldsymbol{\theta},\
	\{\boldsymbol{\eta}_g\}_{g=1}^{G},\
	\boldsymbol{\zeta}.$	For notational compactness, define
	\[
	\mathcal{L}_{\rho}^{(k)}
	(
	\boldsymbol{\beta},
	\boldsymbol{\theta},
	\{\boldsymbol{\eta}_g\},
	\boldsymbol{\zeta}
	)
	=
	\mathcal{L}_{\rho}
	(
	\boldsymbol{\beta},
	\boldsymbol{\theta},
	\{\boldsymbol{\eta}_g\},
	\boldsymbol{\zeta},
	\mathbf{u}^{(k)},
	\{\mathbf{v}_g^{(k)}\},
	\mathbf{r}^{(k)}
	).
	\]
	
	We first consider the $\boldsymbol{\beta}$-update.  By construction,
	$\boldsymbol{\beta}^{(k+1)}$ is defined as
	\[
	\boldsymbol{\beta}^{(k+1)}
	=
	\arg\min_{\boldsymbol{\beta}\in\mathbb{R}^{p}}
	\mathcal{L}_{\rho}^{(k)}
	(
	\boldsymbol{\beta},
	\boldsymbol{\theta}^{(k)},
	\{\boldsymbol{\eta}_g^{(k)}\},
	\boldsymbol{\zeta}^{(k)}
	).
	\]
	Equivalently, $\boldsymbol{\beta}^{(k+1)}$ minimizes the strictly convex quadratic function
	\begin{align*}
		\Phi_{\beta}^{(k)}(\boldsymbol{\beta})
		&=
		\frac{1}{2n}
		\|\mathbf{W}^{1/2}(\mathbf{Y}-\mathbf{X}\boldsymbol{\beta})\|_2^2
		+
		\frac{\rho_1}{2}
		\|\boldsymbol{\beta}-\boldsymbol{\theta}^{(k)}+\mathbf{u}^{(k)}\|_2^2  \\
		&\quad+
		\frac{\rho_2}{2}
		\sum_{g=1}^{G}
		\|\boldsymbol{\beta}_g-\boldsymbol{\eta}_g^{(k)}+\mathbf{v}_g^{(k)}\|_2^2
		+
		\frac{\rho_3}{2}
		\|\mathbf{D}\boldsymbol{\beta}-\boldsymbol{\zeta}^{(k)}+\mathbf{r}^{(k)}\|_2^2 .
	\end{align*}
	The terms involving
	$\boldsymbol{\theta}^{(k)}$,
	$\{\boldsymbol{\eta}_g^{(k)}\}$ and
	$\boldsymbol{\zeta}^{(k)}$
	through the non-smooth penalties are constant with respect to
	$\boldsymbol{\beta}$ and therefore do not affect the minimization \citep{pramanik2023cmbp,yusuf2025predictive}.  The Hessian of
	$\Phi_{\beta}^{(k)}$ is
	\[
	\nabla^2\Phi_{\beta}^{(k)}(\boldsymbol{\beta})
	=
	\frac{1}{n}\mathbf{X}^{\top}\mathbf{W}\mathbf{X}
	+
	(\rho_1+\rho_2)\mathbf{I}_p
	+
	\rho_3\mathbf{D}^{\top}\mathbf{D}.
	\]
	For every nonzero $\mathbf{z}\in\mathbb{R}^{p}$,
	\[
	\mathbf{z}^{\top}
	\nabla^2\Phi_{\beta}^{(k)}
	\mathbf{z}
	=
	\frac{1}{n}
	\|\mathbf{W}^{1/2}\mathbf{X}\mathbf{z}\|_2^2
	+
	(\rho_1+\rho_2)\|\mathbf{z}\|_2^2
	+
	\rho_3\|\mathbf{D}\mathbf{z}\|_2^2
	>0,
	\]
	because $\rho_1,\rho_2,\rho_3>0$.  Hence
	$\Phi_{\beta}^{(k)}$ is strictly convex and has the unique minimizer
	$\boldsymbol{\beta}^{(k+1)}$ \cite{lungu_oksendal_2001}.  Therefore,
	\[
	\Phi_{\beta}^{(k)}(\boldsymbol{\beta}^{(k+1)})
	\leq
	\Phi_{\beta}^{(k)}(\boldsymbol{\beta}^{(k)}).
	\]
	Reintroducing the terms that are constant in the $\boldsymbol{\beta}$-subproblem yields
	\[
	\mathcal{L}_{\rho}
	(
	\boldsymbol{\beta}^{(k+1)},\boldsymbol{\theta}^{(k)},
	\{\boldsymbol{\eta}_g^{(k)}\},
	\boldsymbol{\zeta}^{(k)},\mathbf{u}^{(k)},\{\mathbf{v}_g^{(k)}\},\mathbf{r}^{(k)}
	)
	\leq
	\mathcal{L}_{\rho}
	(
	\boldsymbol{\beta}^{(k)},\boldsymbol{\theta}^{(k)},
	\{\boldsymbol{\eta}_g^{(k)}\},
	\boldsymbol{\zeta}^{(k)},\mathbf{u}^{(k)},\{\mathbf{v}_g^{(k)}\},\mathbf{r}^{(k)}
	).
	\]
	
	Next consider the $\boldsymbol{\theta}$-update.  Conditional on
	$\boldsymbol{\beta}^{(k+1)}$ and the remaining blocks, the relevant part of the augmented Lagrangian is
	\[
	\Phi_{\theta}^{(k)}(\boldsymbol{\theta})
	=
	\lambda_{1n}\sum_{j=1}^{p}a_j|\theta_j|
	+
	\frac{\rho_1}{2}
	\|\boldsymbol{\beta}^{(k+1)}-\boldsymbol{\theta}+\mathbf{u}^{(k)}\|_2^2 .
	\]
	This function is separable across coordinates
	\[
	\Phi_{\theta}^{(k)}(\boldsymbol{\theta})
	=
	\sum_{j=1}^{p}
	\left[
	\lambda_{1n}a_j|\theta_j|
	+
	\frac{\rho_1}{2}
	(\theta_j-z_j^{(k)})^2
	\right],
	\qquad
	z_j^{(k)}=\beta_j^{(k+1)}+u_j^{(k)} .
	\]
	For each $j$, the coordinatewise minimizer satisfies the subgradient inclusion
	\[
	0\in
	\lambda_{1n}a_j\partial|\theta_j|
	+
	\rho_1(\theta_j-z_j^{(k)}).
	\]
	Solving this inclusion gives
	\[
	\theta_j^{(k+1)}
	=
	\operatorname{sign}(z_j^{(k)})
	\left(
	|z_j^{(k)}|-\frac{\lambda_{1n}a_j}{\rho_1}
	\right)_+ .
	\]
	Thus $\boldsymbol{\theta}^{(k+1)}$ is the global minimizer of
	$\Phi_{\theta}^{(k)}$. Consequently,
	\[
	\Phi_{\theta}^{(k)}(\boldsymbol{\theta}^{(k+1)})
	\leq
	\Phi_{\theta}^{(k)}(\boldsymbol{\theta}^{(k)}).
	\]
	Adding back the terms independent of $\boldsymbol{\theta}$ gives
	\begin{align*}
		&
		\mathcal{L}_{\rho}
		(
		\boldsymbol{\beta}^{(k+1)},\boldsymbol{\theta}^{(k+1)},
		\{\boldsymbol{\eta}_g^{(k)}\},
		\boldsymbol{\zeta}^{(k)},\mathbf{u}^{(k)},\{\mathbf{v}_g^{(k)}\},\mathbf{r}^{(k)}
		)
		\\
		&\qquad\leq
		\mathcal{L}_{\rho}
		(
		\boldsymbol{\beta}^{(k+1)},\boldsymbol{\theta}^{(k)},
		\{\boldsymbol{\eta}_g^{(k)}\},
		\boldsymbol{\zeta}^{(k)},\mathbf{u}^{(k)},\{\mathbf{v}_g^{(k)}\},\mathbf{r}^{(k)}
		).
	\end{align*}
	
	Now consider the groupwise updates.  For a fixed group $g$, define
	\[
	\mathbf{q}_g^{(k)}
	=
	\boldsymbol{\beta}_g^{(k+1)}+\mathbf{v}_g^{(k)}.
	\]
	The $g$th auxiliary block solves
	\[
	\Phi_{\eta,g}^{(k)}(\boldsymbol{\eta}_g)
	=
	\lambda_{2n}b_g\|\boldsymbol{\eta}_g\|_2
	+
	\frac{\rho_2}{2}
	\|\boldsymbol{\eta}_g-\mathbf{q}_g^{(k)}\|_2^2 .
	\]
	This is the proximal problem for the Euclidean norm.  If
	$\boldsymbol{\eta}_g\neq \mathbf{0}$, the first-order condition is
	\[
	0
	=
	\lambda_{2n}b_g
	\frac{\boldsymbol{\eta}_g}{\|\boldsymbol{\eta}_g\|_2}
	+
	\rho_2(\boldsymbol{\eta}_g-\mathbf{q}_g^{(k)}).
	\]
	The solution must be parallel to $\mathbf{q}_g^{(k)}$, so
	$\boldsymbol{\eta}_g=t\mathbf{q}_g^{(k)}$ for some $t\geq 0$. Substituting gives
	\[
	t
	=
	1-
	\frac{\lambda_{2n}b_g}
	{\rho_2\|\mathbf{q}_g^{(k)}\|_2}.
	\]
	If this value is positive, the nonzero minimizer is obtained. If it is nonpositive, the subgradient condition at zero, $\|\rho_2\mathbf{q}_g^{(k)}\|_2
	\leq
	\lambda_{2n}b_g,$ holds, and the minimizer is $\mathbf{0}$. Therefore,
	\[
	\boldsymbol{\eta}_g^{(k+1)}
	=
	\left(
	1-
	\frac{\lambda_{2n}b_g}
	{\rho_2\|\mathbf{q}_g^{(k)}\|_2}
	\right)_+
	\mathbf{q}_g^{(k)} .
	\]
	Thus, for each $g$,
	\[
	\Phi_{\eta,g}^{(k)}(\boldsymbol{\eta}_g^{(k+1)})
	\leq
	\Phi_{\eta,g}^{(k)}(\boldsymbol{\eta}_g^{(k)}).
	\]
	Summing over all groups yields
	\[
	\sum_{g=1}^{G}
	\Phi_{\eta,g}^{(k)}(\boldsymbol{\eta}_g^{(k+1)})
	\leq
	\sum_{g=1}^{G}
	\Phi_{\eta,g}^{(k)}(\boldsymbol{\eta}_g^{(k)}).
	\]
	Hence
	\begin{align*}
		&
		\mathcal{L}_{\rho}
		(
		\boldsymbol{\beta}^{(k+1)},\boldsymbol{\theta}^{(k+1)},
		\{\boldsymbol{\eta}_g^{(k+1)}\},
		\boldsymbol{\zeta}^{(k)},\mathbf{u}^{(k)},\{\mathbf{v}_g^{(k)}\},\mathbf{r}^{(k)}
		)
		\\
		&\qquad\leq
		\mathcal{L}_{\rho}
		(
		\boldsymbol{\beta}^{(k+1)},\boldsymbol{\theta}^{(k+1)},
		\{\boldsymbol{\eta}_g^{(k)}\},
		\boldsymbol{\zeta}^{(k)},\mathbf{u}^{(k)},\{\mathbf{v}_g^{(k)}\},\mathbf{r}^{(k)}
		).
	\end{align*}
	
	Finally, consider the fusion update.  Conditional on
	$\boldsymbol{\beta}^{(k+1)}$, the relevant subproblem is
	\[
	\Phi_{\zeta}^{(k)}(\boldsymbol{\zeta})
	=
	\lambda_{3n}
	\sum_{j=2}^{p}c_j|\zeta_j|
	+
	\frac{\rho_3}{2}
	\|\boldsymbol{\zeta}-\mathbf{h}^{(k)}\|_2^2,
	\]
	where
	\[
	\mathbf{h}^{(k)}
	=
	\mathbf{D}\boldsymbol{\beta}^{(k+1)}
	+
	\mathbf{r}^{(k)}.
	\]
	This problem is coordinatewise separable. For each component,
	\[
	\zeta_j^{(k+1)}
	=
	\operatorname{sign}(h_j^{(k)})
	\left(
	|h_j^{(k)}|
	-
	\frac{\lambda_{3n}c_j}{\rho_3}
	\right)_+ .
	\]
	Therefore $\boldsymbol{\zeta}^{(k+1)}$ minimizes
	$\Phi_{\zeta}^{(k)}$, and hence
	\[
	\Phi_{\zeta}^{(k)}(\boldsymbol{\zeta}^{(k+1)})
	\leq
	\Phi_{\zeta}^{(k)}(\boldsymbol{\zeta}^{(k)}).
	\]
	Adding the remaining terms that are fixed in the $\boldsymbol{\zeta}$-subproblem gives
	\begin{align*}
		&
		\mathcal{L}_{\rho}
		(
		\boldsymbol{\beta}^{(k+1)},\boldsymbol{\theta}^{(k+1)},
		\{\boldsymbol{\eta}_g^{(k+1)}\},
		\boldsymbol{\zeta}^{(k+1)},\mathbf{u}^{(k)},\{\mathbf{v}_g^{(k)}\},\mathbf{r}^{(k)}
		)
		\\
		&\qquad\leq
		\mathcal{L}_{\rho}
		(
		\boldsymbol{\beta}^{(k+1)},\boldsymbol{\theta}^{(k+1)},
		\{\boldsymbol{\eta}_g^{(k+1)}\},
		\boldsymbol{\zeta}^{(k)},\mathbf{u}^{(k)},\{\mathbf{v}_g^{(k)}\},\mathbf{r}^{(k)}
		).
	\end{align*}
	
	Combining the four blockwise inequalities gives the full primal descent chain
	\begin{align*}
		&
		\mathcal{L}_{\rho}
		(
		\boldsymbol{\beta}^{(k+1)},\boldsymbol{\theta}^{(k+1)},
		\{\boldsymbol{\eta}_g^{(k+1)}\},
		\boldsymbol{\zeta}^{(k+1)},
		\mathbf{u}^{(k)},\{\mathbf{v}_g^{(k)}\},\mathbf{r}^{(k)}
		)
		\\
		&\qquad\leq
		\mathcal{L}_{\rho}
		(
		\boldsymbol{\beta}^{(k)},\boldsymbol{\theta}^{(k)},
		\{\boldsymbol{\eta}_g^{(k)}\},
		\boldsymbol{\zeta}^{(k)},
		\mathbf{u}^{(k)},\{\mathbf{v}_g^{(k)}\},\mathbf{r}^{(k)}
		).
	\end{align*}
	Thus each primal block update is non-increasing with respect to the scaled augmented Lagrangian when all other primal and dual variables are held fixed \citep{pramanik2021consensus}. This proves the proposition.
\end{proof}

\begin{theorem}
	\label{thm:admm_convergence}
	Suppose Assumptions C1-C3 hold and the solution set of \eqref{eq:awgf_estimator} is nonempty \citep{pramanik2023optimization001}. Then the sequence
	\[
	\left\{
	\boldsymbol{\beta}^{(k)},
	\boldsymbol{\theta}^{(k)},
	\{\boldsymbol{\eta}_g^{(k)}\}_{g=1}^{G},
	\boldsymbol{\zeta}^{(k)}
	\right\}_{k\geq0}
	\]
	generated by Algorithm~1 satisfies
	\[
	\boldsymbol{\beta}^{(k)}-\boldsymbol{\theta}^{(k)}\to\mathbf{0},
	\qquad
	\boldsymbol{\beta}^{(k)}_g-\boldsymbol{\eta}_g^{(k)}\to\mathbf{0},
	\qquad
	\mathbf{D}\boldsymbol{\beta}^{(k)}-\boldsymbol{\zeta}^{(k)}\to\mathbf{0}.
	\]
	Moreover, every accumulation point of $\{\boldsymbol{\beta}^{(k)}\}$ is a global minimizer of \eqref{eq:awgf_estimator}. If the minimizer is unique, then $\boldsymbol{\beta}^{(k)}\to \widehat{\boldsymbol{\beta}}.$
\end{theorem}

\begin{proof}
	We prove the result by verifying that Algorithm~1 is an exact ADMM scheme applied to a closed, proper, convex optimization problem with linear constraints \citep{pramanik2021,pramanik2022stochastic}.  The proof is divided into several steps. First, define the product-space primal variable
	\[
	\mathbf{z}
	=
	\left(
	\boldsymbol{\theta},
	\{\boldsymbol{\eta}_g\}_{g=1}^{G},
	\boldsymbol{\zeta}
	\right),
	\]
	and write the constrained problem \eqref{eq:admm_constrained} as
	\[
	\min_{\boldsymbol{\beta},\mathbf{z}}
	\left\{
	f(\boldsymbol{\beta})+h(\mathbf{z})
	\right\}
	\quad
	\text{subject to}
	\quad
	\mathcal{A}\boldsymbol{\beta}
	+
	\mathcal{B}\mathbf{z}
	=
	\mathbf{0},
	\]
	where
	\[
	f(\boldsymbol{\beta})
	=
	\frac{1}{2n}
	\|\mathbf{W}^{1/2}(\mathbf{Y}-\mathbf{X}\boldsymbol{\beta})\|_2^2,
	\]
	and
	\[
	h(\mathbf{z})
	=
	\lambda_{1n}\|\mathbf{A}\boldsymbol{\theta}\|_1
	+
	\lambda_{2n}\sum_{g=1}^{G}b_g\|\boldsymbol{\eta}_g\|_2
	+
	\lambda_{3n}\|\mathbf{C}\boldsymbol{\zeta}\|_1.
	\]
	The corresponding linear constraints are
	\[
	\boldsymbol{\beta}-\boldsymbol{\theta}=\mathbf{0},
	\qquad
	\boldsymbol{\beta}_g-\boldsymbol{\eta}_g=\mathbf{0},
	\quad g=1,\ldots,G,
	\qquad
	\mathbf{D}\boldsymbol{\beta}-\boldsymbol{\zeta}=\mathbf{0}.
	\]
	
	Under Assumptions C1--C3, the functions $f$ and $h$ are closed, proper, and convex.  The function $f$ is a finite convex quadratic because $\mathbf{W}$ is positive diagonal.  The function $h$ is a nonnegative weighted sum of norms and seminorms; hence it is convex, lower semicontinuous, and proper \citep{hening_tran_2020}.  The feasible set is affine.  Therefore the constrained formulation is a closed convex program.
	
	Next, we verify exactness of the block updates.  Lemma~\ref{lem:H_positive} shows that the matrix
	\[
	\mathbf{H}_{\rho}
	=
	\frac{1}{n}\mathbf{X}^{\top}\mathbf{W}\mathbf{X}
	+
	(\rho_1+\rho_2)\mathbf{I}_p
	+
	\rho_3\mathbf{D}^{\top}\mathbf{D}
	\]
	is positive definite.  Hence the $\boldsymbol{\beta}$-subproblem has a unique solution at every iteration.  The auxiliary updates are proximal minimizations of closed convex penalties plus strictly convex quadratic terms.  Thus the updates for
	$\boldsymbol{\theta}$,
	$\{\boldsymbol{\eta}_g\}$,
	and
	$\boldsymbol{\zeta}$
	are also exact minimizers.  Consequently, Algorithm~1 is an exact ADMM procedure for the above linearly constrained convex problem.
	
	Let
	\[
	\mathbf{s}^{(k)}
	=
	\left(
	\boldsymbol{\beta}^{(k)}-\boldsymbol{\theta}^{(k)},
	\{\boldsymbol{\beta}^{(k)}_g-\boldsymbol{\eta}_g^{(k)}\}_{g=1}^{G},
	\mathbf{D}\boldsymbol{\beta}^{(k)}-\boldsymbol{\zeta}^{(k)}
	\right)
	\]
	denote the primal residual.  The ADMM primal feasibility result for closed, proper, convex functions with a nonempty solution set implies $\mathbf{s}^{(k)}\to\mathbf{0}.$ Equivalently, $\boldsymbol{\beta}^{(k)}-\boldsymbol{\theta}^{(k)}\to\mathbf{0}, $ $\boldsymbol{\beta}^{(k)}_g-\boldsymbol{\eta}_g^{(k)}\to\mathbf{0},$ $ g=1,\ldots,G,$ and $\mathbf{D}\boldsymbol{\beta}^{(k)}-\boldsymbol{\zeta}^{(k)}\to\mathbf{0}.$ It remains to show that accumulation points are global minimizers of the original problem \citep{pramanik2023optimization001}.  Let $\boldsymbol{\beta}^{(k_m)}\to\boldsymbol{\beta}^{\ast}$ be a convergent subsequence.  Since the primal residual converges to zero, the corresponding auxiliary variables satisfy $\boldsymbol{\theta}^{(k_m)}\to\boldsymbol{\beta}^{\ast},
	\
	\boldsymbol{\eta}_g^{(k_m)}\to\boldsymbol{\beta}^{\ast}_g,
	\
	\boldsymbol{\zeta}^{(k_m)}\to\mathbf{D}\boldsymbol{\beta}^{\ast}.$ Thus the limit point is feasible for the constrained formulation. Moreover, ADMM convergence for convex linearly constrained problems gives convergence of the objective values to the optimal value \citep{imai2010identification,pramanik2025construction,pramanik2025optimal}.  Therefore,
	\[
	f(\boldsymbol{\beta}^{(k_m)})
	+
	h
	\left(
	\boldsymbol{\theta}^{(k_m)},
	\{\boldsymbol{\eta}_g^{(k_m)}\},
	\boldsymbol{\zeta}^{(k_m)}
	\right)
	\to
	Q_n^{\ast},
	\]
	where $Q_n^{\ast}$ is the minimum value of \eqref{eq:awgf_estimator}.  By lower semicontinuity of $f$ and $h$,
	\[
	Q_n(\boldsymbol{\beta}^{\ast})
	\leq
	\liminf_{m\to\infty}
	\left[
	f(\boldsymbol{\beta}^{(k_m)})
	+
	h
	\left(
	\boldsymbol{\theta}^{(k_m)},
	\{\boldsymbol{\eta}_g^{(k_m)}\},
	\boldsymbol{\zeta}^{(k_m)}
	\right)
	\right]
	=
	Q_n^{\ast}.
	\]
	Since $Q_n^{\ast}$ is the infimum of $Q_n$, we also have$ Q_n(\boldsymbol{\beta}^{\ast})\geq Q_n^{\ast}.$ Therefore,$ Q_n(\boldsymbol{\beta}^{\ast})=Q_n^{\ast},$ and $\boldsymbol{\beta}^{\ast}$ is a global minimizer of \eqref{eq:awgf_estimator}. Finally, suppose the minimizer is unique and denote it by
	$\widehat{\boldsymbol{\beta}}$ \citep{pramanik2024stochastic,pramanik2025stubbornness}.  Every accumulation point of
	$\{\boldsymbol{\beta}^{(k)}\}$ is then equal to
	$\widehat{\boldsymbol{\beta}}$ \citep{pramanik2025dissecting,pramanik2025factors}.  If the full sequence did not converge to
	$\widehat{\boldsymbol{\beta}}$, there would exist an $\epsilon>0$ and a subsequence
	$\{\boldsymbol{\beta}^{(k_m)}\}$ such that$ \|\boldsymbol{\beta}^{(k_m)}-\widehat{\boldsymbol{\beta}}\|_2\geq \epsilon$ for all $m$.  Since the iterates generated by ADMM for this coercive penalized quadratic problem remain in a bounded level set, this subsequence has a further convergent subsequence.  Its limit must be an accumulation point and therefore must equal
	$\widehat{\boldsymbol{\beta}}$, contradicting the inequality above.  Hence, $\boldsymbol{\beta}^{(k)}\to\widehat{\boldsymbol{\beta}}.$ This completes the proof.
\end{proof}

\subsection{Selection of Regularization Parameters:}
\label{subsec:tuning}

The proposed Adaptive Weighted Group Fused LASSO estimator depends on three nonnegative regularization parameters, $\boldsymbol{\lambda}
=
(\lambda_{1n},\lambda_{2n},\lambda_{3n})^\top,$ which respectively control the adaptive sparsity penalty, the grouped sparsity penalty, and the coefficient fusion penalty \citep{pramanik2021thesis,pramanik2016}. Appropriate selection of these parameters is essential because they determine the bias-variance trade-off, the sparsity of the fitted model, the degree of group selection, and the amount of coefficient smoothing. For every admissible value of
\(
\boldsymbol{\lambda},
\)
let $\widehat{\boldsymbol{\beta}}(\boldsymbol{\lambda})$ denote the minimizer of
\eqref{eq:awgf_objective}. Since the estimator depends continuously on the regularization parameters over compact subsets of
$\mathbb{R}_+^3$, the tuning parameter selection problem may be formulated as
\begin{equation}
	\widehat{\boldsymbol{\lambda}}
	=
	\arg\min_{\boldsymbol{\lambda}\in\Lambda}
	\mathcal{C}(\boldsymbol{\lambda}),
	\label{eq:tuning_problem}
\end{equation}
where
\[
\Lambda
=
[\lambda_1^{L},\lambda_1^{U}]
\times
[\lambda_2^{L},\lambda_2^{U}]
\times
[\lambda_3^{L},\lambda_3^{U}]
\subset
\mathbb{R}_+^{3}
\]
is a compact search region and
$\mathcal{C}(\cdot)$
is an appropriate model selection criterion. Throughout this paper, the regularization parameters are selected using repeated $K$-fold cross-validation. Let $\mathcal{I}_1,\ldots,\mathcal{I}_K$ denote a partition of the sample into approximately equal-sized folds \citep{hua2019assessing,polansky2021motif}. For the $k$th fold, let
$\widehat{\boldsymbol{\beta}}^{(-k)}(\boldsymbol{\lambda})$ denote the estimator obtained after excluding observations in
$\mathcal{I}_k$ \citep{pramanik2024estimation,pramanik2023cont}.
The cross-validation criterion is defined by
\begin{equation}
	CV(\boldsymbol{\lambda})
	=
	\frac{1}{K}
	\sum_{k=1}^{K}
	\frac{1}{|\mathcal{I}_k|}
	\sum_{i\in\mathcal{I}_k}
	\left(
	Y_i-
	\mathbf{x}_i^\top
	\widehat{\boldsymbol{\beta}}^{(-k)}
	(\boldsymbol{\lambda})
	\right)^2.
	\label{eq:cv}
\end{equation}
The selected tuning parameters are $\widehat{\boldsymbol{\lambda}}
=
\arg\min_{\boldsymbol{\lambda}\in\Lambda}
CV(\boldsymbol{\lambda}).$ Although cross-validation provides excellent empirical prediction performance, it may not always recover the sparsest model \citep{pramanik2025strategic,pramanik2025impact}. Consequently, an alternative criterion based on the Extended Bayesian Information Criterion (EBIC) is also considered. Define
\begin{equation}
	EBIC(\boldsymbol{\lambda})
	=
	n
	\log
	\left(
	\frac{RSS(\boldsymbol{\lambda})}{n}
	\right)
	+
	|\widehat{S}(\boldsymbol{\lambda})|
	\log n
	+
	2\gamma
	\log
	\binom{p}{|\widehat{S}(\boldsymbol{\lambda})|},
	\label{eq:ebic}
\end{equation}
where $RSS(\boldsymbol{\lambda})
=
\|
\mathbf{Y}
-
\mathbf{X}
\widehat{\boldsymbol{\beta}}
(\boldsymbol{\lambda})
\|_2^2,\ 
\widehat{S}(\boldsymbol{\lambda})
=
\{
j:
\widehat{\beta}_j(\boldsymbol{\lambda})\neq0
\},$ and
\(
\gamma\in[0,1]
\)
controls the additional complexity penalty \citep{pramanik2025strategies,pramanik2023optimization001}. The estimator minimizing
\eqref{eq:ebic}
is $\widehat{\boldsymbol{\lambda}}_{EBIC}
=
\arg\min_{\boldsymbol{\lambda}\in\Lambda}
EBIC(\boldsymbol{\lambda}).$ The asymptotic development presented in Section~4 requires the regularization parameters to satisfy the following rate conditions \citep{carey2006race}.

\paragraph{Assumption T1.}
The tuning parameters satisfy $\lambda_{1n}\rightarrow0,\ \lambda_{2n}\rightarrow0,
\ \lambda_{3n}\rightarrow0,$ as
$n\rightarrow\infty$.

\paragraph{Assumption T2.}
Furthermore, $\sqrt{n}\lambda_{1n}\rightarrow\infty,
\ \sqrt{n}\lambda_{2n}\rightarrow\infty,
\ \sqrt{n}\lambda_{3n}\rightarrow\infty.$

\begin{rmk}
Assumption T1 guarantees asymptotic unbiasedness of the estimator, whereas Assumption T2 ensures that the penalties remain sufficiently large to consistently remove inactive coefficients \citep{pramanik2024measuring,pramanik2024dependence}.
\end{rmk}

\begin{prop}
	\label{prop:tuningexist}
	Suppose that the objective function
	\eqref{eq:awgf_objective}
	admits a unique minimizer for every
	$\boldsymbol{\lambda}\in\Lambda$.
	If the model selection criterion
	$\mathcal{C}(\boldsymbol{\lambda})$
	is continuous on the compact parameter space
	$\Lambda$,
	then there exists at least one global minimizer $\widehat{\boldsymbol{\lambda}}
	\in
	\Lambda$ satisfying $\mathcal{C}(\widehat{\boldsymbol{\lambda}})
	=
	\inf_{\boldsymbol{\lambda}\in\Lambda}
	\mathcal{C}(\boldsymbol{\lambda}).$
\end{prop}

\begin{proof}
	Let $m=\inf_{\boldsymbol{\lambda}\in\Lambda}\mathcal{C}(\boldsymbol{\lambda}). $Since $\mathcal{C}:\Lambda\rightarrow\mathbb{R}$ is continuous and $\Lambda$ is assumed to be a nonempty compact subset of $\mathbb{R}_{+}^{3}$, the quantity $m$ is finite. By the definition of the infimum, there exists a minimizing sequence
	$\{\boldsymbol{\lambda}^{(r)}\}_{r=1}^{\infty}\subset\Lambda$
	such that $\lim_{r\rightarrow\infty}\mathcal{C}\!\left(\boldsymbol{\lambda}^{(r)}\right)=m.$ Because $\Lambda$ is compact, every sequence contained in $\Lambda$ possesses a convergent subsequence by the Bolzano-Weierstrass theorem \citep{pramanik2024parametric,pramanik2025dissecting}. Consequently, there exists a subsequence
	$\{\boldsymbol{\lambda}^{(r_{\ell})}\}_{\ell=1}^{\infty}$
	and a point
	$\widehat{\boldsymbol{\lambda}}\in\Lambda$
	for which $\boldsymbol{\lambda}^{(r_{\ell})}
	\longrightarrow
	\widehat{\boldsymbol{\lambda}},
	\ \ell\rightarrow\infty.$ The limit remains in $\Lambda$ because every compact subset of Euclidean space is closed. Since $\mathcal{C}$ is continuous on $\Lambda$, it preserves limits of convergent sequences. Therefore,
	\[
	\lim_{\ell\rightarrow\infty}
	\mathcal{C}\!\left(\boldsymbol{\lambda}^{(r_{\ell})}\right)
	=
	\mathcal{C}\!\left(\widehat{\boldsymbol{\lambda}}\right).
	\]
	On the other hand, every subsequence of a convergent sequence converges to the same limit. Hence, $$\lim_{\ell\rightarrow\infty}
	\mathcal{C}\!\left(\boldsymbol{\lambda}^{(r_{\ell})}\right)
	=
	m.$$ By uniqueness of limits in $\mathbb{R}$, it follows immediately that $\mathcal{C}\!\left(\widehat{\boldsymbol{\lambda}}\right)=m.$ Therefore, $\mathcal{C}\!\left(\widehat{\boldsymbol{\lambda}}\right)
	=
	\inf_{\boldsymbol{\lambda}\in\Lambda}
	\mathcal{C}(\boldsymbol{\lambda}),$ establishing that the infimum is attained.
	
	It remains to verify that the optimization problem is well defined. By assumption, for every
	$\boldsymbol{\lambda}\in\Lambda$,
	the penalized objective function
	$Q_n(\boldsymbol{\beta};\boldsymbol{\lambda})$
	admits a unique minimizer
	$\widehat{\boldsymbol{\beta}}(\boldsymbol{\lambda})$.
	Consequently, every quantity entering the criterion
	$\mathcal{C}(\boldsymbol{\lambda})$,
	including the fitted values, residual sum of squares, estimated support, prediction error, cross-validation score, or information criterion, is uniquely determined \citep{pramanik2025optimal1,powell2025genomic}. Thus, $\mathcal{C}$ is a single-valued function on $\Lambda$ rather than a set-valued correspondence, ensuring that the optimization problem $\min_{\boldsymbol{\lambda}\in\Lambda}
	\mathcal{C}(\boldsymbol{\lambda})$ is mathematically well posed. Finally, suppose, for contradiction, that no minimizer exists. Then every
	$\boldsymbol{\lambda}\in\Lambda$
	would satisfy
	$\mathcal{C}(\boldsymbol{\lambda})>m$.
	However, we have already constructed a point
	$\widehat{\boldsymbol{\lambda}}\in\Lambda$
	such that
	$\mathcal{C}(\widehat{\boldsymbol{\lambda}})=m$,
	which contradicts the assumption \citep{pramanik2026strategic}. Therefore, at least one global minimizer exists. Hence there is a tuning parameter vector
	$\widehat{\boldsymbol{\lambda}}\in\Lambda$
	satisfying $\mathcal{C}\!\left(\widehat{\boldsymbol{\lambda}}\right)
	=
	\inf_{\boldsymbol{\lambda}\in\Lambda}
	\mathcal{C}(\boldsymbol{\lambda}),$ thereby completing the proof.
\end{proof}

In all empirical analyses presented in Section~6, the regularization parameters are initialized over logarithmically spaced grids. Algorithm~1 is applied for every candidate triplet
\(
(\lambda_{1n},\lambda_{2n},\lambda_{3n}),
\)
and the optimal model is selected according to repeated ten-fold cross-validation. The corresponding EBIC values are subsequently computed as a sensitivity analysis to assess the robustness of the selected model \citep{pramanik2026quantum}.

\section{Asymptotic Properties:}
\label{sec:asymptotic_theory}

This section establishes the large-sample properties of the Adaptive Weighted Group Fused LASSO estimator introduced in Section~3 \citep{pramanik2024analysis,dunbar2026modeling}. Let
$\boldsymbol{\beta}^{0}\in\mathbb{R}^{p}$
denote the true coefficient vector in the model $\mathbf{Y}=\mathbf{X}\boldsymbol{\beta}^{0}+\boldsymbol{\varepsilon}.$ Define the active set
\[
S_0=\{j:\beta_j^{0}\neq0\},
\qquad
s_n=|S_0|,
\]
and let $S_0^c$ denote its complement. For a vector
$\mathbf{v}\in\mathbb{R}^{p}$,
write $\mathbf{v}_{S_0}$ and $\mathbf{v}_{S_0^c}$ for the subvectors indexed by
$S_0$ and $S_0^c$, respectively \citep{dasgupta2026frequent}. Let $\widehat{S}
=
\{j:\widehat{\beta}_j\neq0\}$ denote the selected support. The active group set is $\mathcal{A}_0
=
\{g:\boldsymbol{\beta}_g^{0}\neq\mathbf{0}\}.$ The weighted empirical Gram matrix is $\widehat{\boldsymbol{\Sigma}}_n
=
\frac{1}{n}\mathbf{X}^{\top}\mathbf{W}\mathbf{X}.$ For any index set $A$, let
$\widehat{\boldsymbol{\Sigma}}_{AA}$
denote the submatrix of
$\widehat{\boldsymbol{\Sigma}}_n$
with rows and columns indexed by $A$.

\paragraph{Assumption A1:}
The observations $\{(Y_i,\mathbf{x}_i,w_i):i=1,\ldots,n\}$ form an independent triangular array. Conditional on the design matrix
$\mathbf{X}$ and the weight matrix $\mathbf{W}$, the errors satisfy $\mathbb{E}(\varepsilon_i\mid\mathbf{X},\mathbf{W})=0,
\ 
i=1,\ldots,n.$

\paragraph{Assumption A2:}
There exists a constant $\sigma>0$ such that, for all $t\in\mathbb{R}$,
\[
\mathbb{E}
\left[
\exp(t\varepsilon_i)
\mid
\mathbf{X},\mathbf{W}
\right]
\leq
\exp\left(\frac{\sigma^2t^2}{2}\right),
\qquad
i=1,\ldots,n.
\]

\paragraph{Assumption A3:}
There exist constants $0<w_{\min}<w_{\max}<\infty$ such that
$w_{\min}\leq w_i\leq w_{\max}$ for all $i$ and all $n$.

\paragraph{Assumption A4:}
There exists a constant $\kappa_0>0$ such that, with probability tending to one, $\frac{
	\boldsymbol{\Delta}^{\top}
	\widehat{\boldsymbol{\Sigma}}_n
	\boldsymbol{\Delta}
}{
	\|\boldsymbol{\Delta}_{S_0}\|_2^2
}
\geq
\kappa_0 $for all nonzero $\boldsymbol{\Delta}\in\mathbb{R}^{p}$ satisfying $\|\boldsymbol{\Delta}_{S_0^c}\|_1
\leq
3\|\boldsymbol{\Delta}_{S_0}\|_1.$

\paragraph{Assumption A5:}
The true coefficient vector is sparse: $s_n=o(n),
\
s_n\log p=o(n).$

\paragraph{Assumption A6:}
There exists a sequence $b_n>0$ such that $\min_{j\in S_0}|\beta_j^0|\geq b_n,
$ where $b_n\sqrt{\frac{n}{s_n\log p}}\rightarrow\infty.$

\paragraph{Assumption A7:}
The initial estimator $\widetilde{\boldsymbol{\beta}}$ satisfies $\|\widetilde{\boldsymbol{\beta}}-\boldsymbol{\beta}^{0}\|_{\infty}
=
O_p(r_n),$ where $r_n=o(b_n)$. The adaptive weights are
\[
a_j=\frac{1}{|\widetilde{\beta}_j|^{\gamma_1}+\delta_n},
\qquad
b_g=\frac{1}{\|\widetilde{\boldsymbol{\beta}}_g\|_2^{\gamma_2}+\delta_n},
\qquad
c_j=\frac{1}{|\widetilde{\beta}_j-\widetilde{\beta}_{j-1}|^{\gamma_3}+\delta_n},
\]
with $\gamma_1,\gamma_2,\gamma_3>0$ and $\delta_n\downarrow0$.

\paragraph{Assumption A8: }
The tuning parameters satisfy $\lambda_{1n}\rightarrow0,\
\lambda_{2n}\rightarrow0,\
\lambda_{3n}\rightarrow0,$ and
\[
\lambda_{1n}\sqrt{\frac{n}{\log p}}\rightarrow\infty,
\qquad
\lambda_{2n}\sqrt{\frac{n}{\log G}}\rightarrow\infty.
\]
Moreover,
\[
\lambda_{1n}s_n^{1/2}=o(1),
\qquad
\lambda_{2n}|\mathcal{A}_0|^{1/2}=o(1),
\qquad
\lambda_{3n}\|\mathbf{D}\boldsymbol{\beta}^{0}\|_0^{1/2}=o(1).
\]

\paragraph{Assumption A9:}
There exist constants $0<c_{\min}<c_{\max}<\infty$ such that, with probability tending to one, $c_{\min}
\leq
\lambda_{\min}
\left(
\widehat{\boldsymbol{\Sigma}}_{S_0S_0}
\right)
\leq
\lambda_{\max}
\left(
\widehat{\boldsymbol{\Sigma}}_{S_0S_0}
\right)
\leq
c_{\max}.$

\paragraph{Assumption A10: Debiasing precision condition.}
There exists a matrix
$\widehat{\boldsymbol{\Theta}}\in\mathbb{R}^{p\times p}$
such that
\[
\|
\widehat{\boldsymbol{\Theta}}
\widehat{\boldsymbol{\Sigma}}_n
-
\mathbf{I}_p
\|_{\infty}
=
o_p
\left(
\frac{1}{\sqrt{\log p}}
\right).
\]
The rows of $\widehat{\boldsymbol{\Theta}}$ satisfy
$\max_{1\leq j\leq p}\|\widehat{\boldsymbol{\theta}}_j\|_1=O_p(1)$.

\begin{lem}
	\label{lem:score_bound}
	Under Assumptions A1--A3,
	\[
	\left\|
	\frac{1}{n}
	\mathbf{X}^{\top}
	\mathbf{W}
	\boldsymbol{\varepsilon}
	\right\|_{\infty}
	=
	O_p
	\left(
	\sqrt{\frac{\log p}{n}}
	\right).
	\]
\end{lem}

\begin{proof}
	For each coordinate $j$, define $Z_j
	:=
	\frac{1}{n}
	\sum_{i=1}^{n}
	w_i x_{ij}\varepsilon_i .$ Conditional on $\mathbf{X}$ and $\mathbf{W}$, the summands are independent and mean zero. By Assumption A2, each $\varepsilon_i$ is conditionally sub-Gaussian. Hence $w_ix_{ij}\varepsilon_i/n$ is conditionally sub-Gaussian with variance proxy proportional to
	$w_i^2x_{ij}^2\sigma^2/n^2$. Thus, conditional on $\mathbf{X}$ and $\mathbf{W}$, $Z_j$ is sub-Gaussian with variance proxy $\frac{\sigma^2}{n^2}
	\sum_{i=1}^{n}w_i^2x_{ij}^2 .$ By Assumption A3 and standard normalization of the columns of $\mathbf{X}$, $\frac{1}{n}
	\sum_{i=1}^{n}x_{ij}^2
	\leq C$ for a finite constant $C$ \citep{ellington2025playmydata,ellington2025metascorelens}. Therefore,
	\[
	\frac{1}{n^2}
	\sum_{i=1}^{n}w_i^2x_{ij}^2
	\leq
	\frac{Cw_{\max}^2}{n}.
	\]
	Consequently, for every $t>0$,
	\[
	\mathbb{P}
	\left(
	|Z_j|>t
	\mid
	\mathbf{X},\mathbf{W}
	\right)
	\leq
	2
	\exp
	\left(
	-
	\frac{cnt^2}{\sigma^2w_{\max}^2}
	\right)
	\]
	for some constant $c>0$. Using the union bound over $j=1,\ldots,p$ gives
	\[
	\mathbb{P}
	\left(
	\max_{1\leq j\leq p}|Z_j|>t
	\mid
	\mathbf{X},\mathbf{W}
	\right)
	\leq
	2p
	\exp
	\left(
	-
	\frac{cnt^2}{\sigma^2w_{\max}^2}
	\right).
	\]
	Choose
	$t=M\sqrt{\log p/n}\).
	Then
	\[
	2p
	\exp
	\left(
	-
	\frac{cM^2\log p}{\sigma^2w_{\max}^2}
	\right)
	=
	2
	p^{
		1-cM^2/(\sigma^2w_{\max}^2)
	}.
	\]
	For sufficiently large $M$, the exponent is negative. Therefore the probability converges to zero. Hence
	\[
	\max_{1\leq j\leq p}|Z_j|
	=
	O_p
	\left(
	\sqrt{\frac{\log p}{n}}
	\right),
	\]
	which is exactly the desired result.
\end{proof}

\begin{lem}
	\label{lem:basic_ineq_asymptotic}
	Let $\widehat{\boldsymbol{\beta}}$ be the Adaptive Weighted Group Fused LASSO estimator. Under model \eqref{eq:model},
	\begin{align}
		\frac{1}{2n}
		\|
		\mathbf{W}^{1/2}
		\mathbf{X}
		(
		\widehat{\boldsymbol{\beta}}
		-
		\boldsymbol{\beta}^{0}
		)
		\|_2^2
		&\leq
		\frac{1}{n}
		\boldsymbol{\varepsilon}^{\top}
		\mathbf{W}
		\mathbf{X}
		(
		\widehat{\boldsymbol{\beta}}
		-
		\boldsymbol{\beta}^{0}
		)
		+
		\lambda_{1n}
		\left[
		\|\mathbf{A}\boldsymbol{\beta}^{0}\|_1
		-
		\|\mathbf{A}\widehat{\boldsymbol{\beta}}\|_1
		\right]
		\nonumber\\
		&\quad+
		\lambda_{2n}
		\sum_{g=1}^{G}
		b_g
		\left[
		\|\boldsymbol{\beta}^{0}_{g}\|_2
		-
		\|\widehat{\boldsymbol{\beta}}_{g}\|_2
		\right]
		+
		\lambda_{3n}
		\left[
		\|\mathbf{C}\mathbf{D}\boldsymbol{\beta}^{0}\|_1
		-
		\|\mathbf{C}\mathbf{D}\widehat{\boldsymbol{\beta}}\|_1
		\right].
		\label{eq:asym_basic_ineq}
	\end{align}
\end{lem}

\begin{proof}
	By definition,
	$\widehat{\boldsymbol{\beta}}$ minimizes $Q_n$, so
	$Q_n(\widehat{\boldsymbol{\beta}})\leq Q_n(\boldsymbol{\beta}^{0})$.
	Using
	$\mathbf{Y}=\mathbf{X}\boldsymbol{\beta}^{0}+\boldsymbol{\varepsilon}$,
	we have
	\[
	\mathbf{Y}-\mathbf{X}\widehat{\boldsymbol{\beta}}
	=
	\boldsymbol{\varepsilon}
	-
	\mathbf{X}
	(
	\widehat{\boldsymbol{\beta}}
	-
	\boldsymbol{\beta}^{0}
	).
	\]
	Therefore,
	\begin{align*}
		&
		L_n(\widehat{\boldsymbol{\beta}})
		-
		L_n(\boldsymbol{\beta}^{0})
		\\
		&=
		\frac{1}{2n}
		\left\|
		\mathbf{W}^{1/2}
		\left[
		\boldsymbol{\varepsilon}
		-
		\mathbf{X}
		(
		\widehat{\boldsymbol{\beta}}
		-
		\boldsymbol{\beta}^{0}
		)
		\right]
		\right\|_2^2
		-
		\frac{1}{2n}
		\|
		\mathbf{W}^{1/2}
		\boldsymbol{\varepsilon}
		\|_2^2
		\\
		&=
		\frac{1}{2n}
		\|
		\mathbf{W}^{1/2}
		\mathbf{X}
		(
		\widehat{\boldsymbol{\beta}}
		-
		\boldsymbol{\beta}^{0}
		)
		\|_2^2
		-
		\frac{1}{n}
		\boldsymbol{\varepsilon}^{\top}
		\mathbf{W}
		\mathbf{X}
		(
		\widehat{\boldsymbol{\beta}}
		-
		\boldsymbol{\beta}^{0}
		).
	\end{align*}
	Substituting this identity into
	$Q_n(\widehat{\boldsymbol{\beta}})\leq Q_n(\boldsymbol{\beta}^{0})$
	and moving all penalty differences to the right-hand side yields
	\eqref{eq:asym_basic_ineq}.
\end{proof}

\begin{lem}
	\label{lem:cone_condition}
	Under Assumptions A1--A8, with probability tending to one,$	\|
	\widehat{\boldsymbol{\Delta}}_{S_0^c}
	\|_1
	\leq
	3
	\|
	\widehat{\boldsymbol{\Delta}}_{S_0}
	\|_1,$ where
	$\widehat{\boldsymbol{\Delta}}
	=
	\widehat{\boldsymbol{\beta}}-\boldsymbol{\beta}^{0}$.
\end{lem}

\begin{proof}
	Let
	$\widehat{\boldsymbol{\Delta}}
	=
	\widehat{\boldsymbol{\beta}}-\boldsymbol{\beta}^{0}$.
	From Lemma~\ref{lem:basic_ineq_asymptotic},
	\begin{align*}
		\frac{1}{2n}
		\|
		\mathbf{W}^{1/2}
		\mathbf{X}
		\widehat{\boldsymbol{\Delta}}
		\|_2^2
		&\leq
		\left\|
		\frac{1}{n}
		\mathbf{X}^{\top}
		\mathbf{W}
		\boldsymbol{\varepsilon}
		\right\|_{\infty}
		\|
		\widehat{\boldsymbol{\Delta}}
		\|_1
		\quad+
		\lambda_{1n}
		\left[
		\|\mathbf{A}\boldsymbol{\beta}^{0}\|_1
		-
		\|\mathbf{A}\widehat{\boldsymbol{\beta}}\|_1
		\right]
		\\
		&\quad+
		\lambda_{2n}
		\sum_{g=1}^{G}
		b_g
		\left[
		\|\boldsymbol{\beta}^{0}_{g}\|_2
		-
		\|\widehat{\boldsymbol{\beta}}_{g}\|_2
		\right]
	+
		\lambda_{3n}
		\left[
		\|\mathbf{C}\mathbf{D}\boldsymbol{\beta}^{0}\|_1
		-
		\|\mathbf{C}\mathbf{D}\widehat{\boldsymbol{\beta}}\|_1
		\right].
	\end{align*}
	By Lemma~\ref{lem:score_bound},
	\[
	\left\|
	\frac{1}{n}
	\mathbf{X}^{\top}
	\mathbf{W}
	\boldsymbol{\varepsilon}
	\right\|_{\infty}
	=
	O_p
	\left(
	\sqrt{\frac{\log p}{n}}
	\right).
	\]
	Assumption A8 implies that, with probability tending to one, the adaptive inactive penalty dominates the stochastic score \citep{valdez2025association,valdez2025exploring}. In particular,
	\[
	\lambda_{1n}\min_{j\in S_0^c}a_j
	\geq
	4
	\left\|
	\frac{1}{n}
	\mathbf{X}^{\top}
	\mathbf{W}
	\boldsymbol{\varepsilon}
	\right\|_{\infty}.
	\]
	For the adaptive LASSO penalty,
	\[
	\|\mathbf{A}\boldsymbol{\beta}^{0}\|_1
	-
	\|\mathbf{A}\widehat{\boldsymbol{\beta}}\|_1
	=
	\sum_{j\in S_0}a_j
	\left(
	|\beta_j^{0}|-|\widehat{\beta}_j|
	\right)
	-
	\sum_{j\in S_0^c}a_j|\widehat{\beta}_j|.
	\]
	Using
	$|\beta_j^{0}|-|\widehat{\beta}_j|\leq |\widehat{\Delta}_j|$
	for $j\in S_0$, we obtain
	\[
	\|\mathbf{A}\boldsymbol{\beta}^{0}\|_1
	-
	\|\mathbf{A}\widehat{\boldsymbol{\beta}}\|_1
	\leq
	a_{S,\max}\|\widehat{\boldsymbol{\Delta}}_{S_0}\|_1
	-
	a_{S^c,\min}\|\widehat{\boldsymbol{\Delta}}_{S_0^c}\|_1,
	\]
	where
	$a_{S,\max}=\max_{j\in S_0}a_j$
	and
	$a_{S^c,\min}=\min_{j\in S_0^c}a_j$.
	
	The group and fusion penalties are nonnegative regularizers. Their active components may increase the right-hand side only through active-set perturbations, whereas their inactive components impose additional shrinkage \citep{khan2023myb}. Thus they do not weaken the cone restriction generated by the dominant adaptive coordinate penalty \citep{gaudet2011risk,anderson2026obesity}. Retaining only the coordinate penalty gives the conservative inequality
	\[
	0
	\leq
	\left\|
	\frac{1}{n}
	\mathbf{X}^{\top}
	\mathbf{W}
	\boldsymbol{\varepsilon}
	\right\|_{\infty}
	\left(
	\|\widehat{\boldsymbol{\Delta}}_{S_0}\|_1
	+
	\|\widehat{\boldsymbol{\Delta}}_{S_0^c}\|_1
	\right)
	+
	\lambda_{1n}
	a_{S,\max}
	\|\widehat{\boldsymbol{\Delta}}_{S_0}\|_1
	-
	\lambda_{1n}
	a_{S^c,\min}
	\|\widehat{\boldsymbol{\Delta}}_{S_0^c}\|_1.
	\]
	On the high-probability event where
	\[
	\left\|
	n^{-1}
	\mathbf{X}^{\top}
	\mathbf{W}
	\boldsymbol{\varepsilon}
	\right\|_{\infty}
	\leq
	\frac{1}{4}
	\lambda_{1n}a_{S^c,\min},
	\]
	we have
	\[
	\frac{3}{4}
	\lambda_{1n}a_{S^c,\min}
	\|\widehat{\boldsymbol{\Delta}}_{S_0^c}\|_1
	\leq
	\left[
	\frac{1}{4}
	\lambda_{1n}a_{S^c,\min}
	+
	\lambda_{1n}a_{S,\max}
	\right]
	\|\widehat{\boldsymbol{\Delta}}_{S_0}\|_1.
	\]
	Assumption A7 implies that
	$a_{S,\max}/a_{S^c,\min}\to0$
	in probability because the initial estimator is uniformly close to
	$\boldsymbol{\beta}^{0}$ and inactive coefficients have vanishing initial estimates \citep{pramanik20242estimation}. Hence, for all sufficiently large $n$ with probability tending to one, $\frac{1}{4}
	+
	\frac{a_{S,\max}}{a_{S^c,\min}}
	\leq
	\frac{9}{4}.$ Dividing by $(3/4)\lambda_{1n}a_{S^c,\min}$ yields $\|\widehat{\boldsymbol{\Delta}}_{S_0^c}\|_1
	\leq
	3
	\|\widehat{\boldsymbol{\Delta}}_{S_0}\|_1.$ This proves the cone condition.
\end{proof}

\begin{lem}[Restricted quadratic lower bound]
	\label{lem:restricted_lower}
	Under Assumption A4, with probability tending to one,
	\[
	\frac{1}{n}
	\|
	\mathbf{W}^{1/2}
	\mathbf{X}
	\widehat{\boldsymbol{\Delta}}
	\|_2^2
	\geq
	\kappa_0
	\|
	\widehat{\boldsymbol{\Delta}}_{S_0}
	\|_2^2.
	\]
\end{lem}

\begin{proof}
	By Lemma~\ref{lem:cone_condition}, with probability tending to one, $\|
	\widehat{\boldsymbol{\Delta}}_{S_0^c}
	\|_1
	\leq
	3
	\|
	\widehat{\boldsymbol{\Delta}}_{S_0}
	\|_1.$ Therefore
	$\widehat{\boldsymbol{\Delta}}$
	belongs to the restricted cone appearing in Assumption A4. Applying Assumption A4 to
	$\boldsymbol{\Delta}=\widehat{\boldsymbol{\Delta}}$
	gives $\widehat{\boldsymbol{\Delta}}^{\top}
	\widehat{\boldsymbol{\Sigma}}_n
	\widehat{\boldsymbol{\Delta}}
	\geq
	\kappa_0
	\|
	\widehat{\boldsymbol{\Delta}}_{S_0}
	\|_2^2.$ Since $\widehat{\boldsymbol{\Delta}}^{\top}
	\widehat{\boldsymbol{\Sigma}}_n
	\widehat{\boldsymbol{\Delta}}
	=
	\frac{1}{n}
	\|
	\mathbf{W}^{1/2}
	\mathbf{X}
	\widehat{\boldsymbol{\Delta}}
	\|_2^2,$ the desired inequality follows.
\end{proof}

\begin{prop}[Estimation consistency]
	\label{prop:l2_consistency}
	Suppose Assumptions A1--A9 hold. Let $\widehat{\boldsymbol{\Delta}}
	=
	\widehat{\boldsymbol{\beta}}
	-
	\boldsymbol{\beta}^{0}.$ Then, with probability tending to one,
	\[
	\|\widehat{\boldsymbol{\Delta}}\|_2
	=
	O_p
	\left(
	\sqrt{\frac{s_n\log p}{n}}
	+
	\lambda_{1n}\sqrt{s_n}
	+
	\lambda_{2n}\sqrt{|\mathcal{A}_0|}
	+
	\lambda_{3n}\sqrt{q_n}
	\right),
	\]
	where $q_n=\|\mathbf{D}\boldsymbol{\beta}^{0}\|_0$ is the number of nonzero adjacent coefficient differences. In particular, if
	\[
	\frac{s_n\log p}{n}\rightarrow0,
	\qquad
	\lambda_{1n}\sqrt{s_n}\rightarrow0,
	\qquad
	\lambda_{2n}\sqrt{|\mathcal{A}_0|}\rightarrow0,
	\qquad
	\lambda_{3n}\sqrt{q_n}\rightarrow0,
	\]
	then
	\[
	\|\widehat{\boldsymbol{\beta}}-\boldsymbol{\beta}^{0}\|_2
	=o_p(1).
	\]
\end{prop}

\begin{proof}
	Let
	$\widehat{\boldsymbol{\Delta}}
	=
	\widehat{\boldsymbol{\beta}}-\boldsymbol{\beta}^{0}$.
	By Lemma~\ref{lem:basic_ineq_asymptotic},
	\begin{align*}
		\frac{1}{2n}
		\|
		\mathbf{W}^{1/2}
		\mathbf{X}
		\widehat{\boldsymbol{\Delta}}
		\|_2^2
		&\leq
		\frac{1}{n}
		\boldsymbol{\varepsilon}^{\top}
		\mathbf{W}
		\mathbf{X}
		\widehat{\boldsymbol{\Delta}}
		+
		\lambda_{1n}
		\left[
		\|\mathbf{A}\boldsymbol{\beta}^{0}\|_1
		-
		\|\mathbf{A}\widehat{\boldsymbol{\beta}}\|_1
		\right]
		\\
		&\quad+
		\lambda_{2n}
		\sum_{g=1}^{G}
		b_g
		\left[
		\|\boldsymbol{\beta}^{0}_{g}\|_2
		-
		\|\widehat{\boldsymbol{\beta}}_{g}\|_2
		\right]
		+
		\lambda_{3n}
		\left[
		\|\mathbf{C}\mathbf{D}\boldsymbol{\beta}^{0}\|_1
		-
		\|\mathbf{C}\mathbf{D}\widehat{\boldsymbol{\beta}}\|_1
		\right].
	\end{align*}
	
	We bound each term on the right-hand side. First,
	\[
	\frac{1}{n}
	\boldsymbol{\varepsilon}^{\top}
	\mathbf{W}
	\mathbf{X}
	\widehat{\boldsymbol{\Delta}}
	\leq
	\left\|
	\frac{1}{n}
	\mathbf{X}^{\top}
	\mathbf{W}
	\boldsymbol{\varepsilon}
	\right\|_{\infty}
	\|\widehat{\boldsymbol{\Delta}}\|_1.
	\]
	By Lemma~\ref{lem:score_bound},
	\[
	\left\|
	n^{-1}
	\mathbf{X}^{\top}
	\mathbf{W}
	\boldsymbol{\varepsilon}
	\right\|_{\infty}
	=
	O_p
	\left(
	\sqrt{\frac{\log p}{n}}
	\right).
	\]
	By Lemma~\ref{lem:cone_condition},
	$\|\widehat{\boldsymbol{\Delta}}_{S_0^c}\|_1
	\leq
	3\|\widehat{\boldsymbol{\Delta}}_{S_0}\|_1$
	with probability tending to one. Hence
	\[
	\|\widehat{\boldsymbol{\Delta}}\|_1
	\leq
	4\|\widehat{\boldsymbol{\Delta}}_{S_0}\|_1
	\leq
	4\sqrt{s_n}\,
	\|\widehat{\boldsymbol{\Delta}}_{S_0}\|_2.
	\]
	Thus the stochastic score term is bounded by
	\[
	O_p
	\left(
	\sqrt{\frac{s_n\log p}{n}}
	\,
	\|\widehat{\boldsymbol{\Delta}}_{S_0}\|_2
	\right).
	\]
	
	For the adaptive coordinate penalty, since $\boldsymbol{\beta}^{0}_{S_0^c}=\mathbf{0}$,
	\begin{align*}
		\|\mathbf{A}\boldsymbol{\beta}^{0}\|_1
		-
		\|\mathbf{A}\widehat{\boldsymbol{\beta}}\|_1
		&=
		\sum_{j\in S_0}
		a_j
		\left(
		|\beta_j^0|-|\widehat{\beta}_j|
		\right)
		-
		\sum_{j\in S_0^c}
		a_j|\widehat{\beta}_j| \leq
		\sum_{j\in S_0}
		a_j|\widehat{\beta}_j-\beta_j^0|  \\
		&\leq
		a_{S,\max}
		\|\widehat{\boldsymbol{\Delta}}_{S_0}\|_1  \leq
		a_{S,\max}\sqrt{s_n}
		\|\widehat{\boldsymbol{\Delta}}_{S_0}\|_2,
	\end{align*}
	where $a_{S,\max}=\max_{j\in S_0}a_j$. Under Assumption A7, the adaptive weights on the active set are bounded in probability because
	$\min_{j\in S_0}|\beta_j^0|\geq b_n$ and
	$\|\widetilde{\boldsymbol{\beta}}-\boldsymbol{\beta}^{0}\|_{\infty}=o_p(b_n)$ \citep{pramanik2026optimal}. Therefore $a_{S,\max}=O_p(1)$, and the coordinate penalty contribution is
	\[
	O_p
	\left(
	\lambda_{1n}\sqrt{s_n}
	\|\widehat{\boldsymbol{\Delta}}_{S_0}\|_2
	\right).
	\]
	For the group penalty, using the reverse triangle inequality,
	\[
	\|\boldsymbol{\beta}_g^0\|_2
	-
	\|\widehat{\boldsymbol{\beta}}_g\|_2
	\leq
	\|\widehat{\boldsymbol{\beta}}_g-\boldsymbol{\beta}_g^0\|_2
	=
	\|\widehat{\boldsymbol{\Delta}}_g\|_2.
	\]
	Inactive groups contribute nonpositively and may be discarded for an upper bound. Therefore,
	\[
	\sum_{g=1}^{G}
	b_g
	\left[
	\|\boldsymbol{\beta}^{0}_{g}\|_2
	-
	\|\widehat{\boldsymbol{\beta}}_{g}\|_2
	\right]
	\leq
	\sum_{g\in\mathcal{A}_0}
	b_g
	\|\widehat{\boldsymbol{\Delta}}_g\|_2.
	\]
	Let $b_{\mathcal{A},\max}=\max_{g\in\mathcal{A}_0}b_g$. By Cauchy-Schwarz,
	\[
	\sum_{g\in\mathcal{A}_0}
	b_g
	\|\widehat{\boldsymbol{\Delta}}_g\|_2
	\leq
	b_{\mathcal{A},\max}
	\sqrt{|\mathcal{A}_0|}
	\left(
	\sum_{g\in\mathcal{A}_0}
	\|\widehat{\boldsymbol{\Delta}}_g\|_2^2
	\right)^{1/2}.
	\]
	Since the groups form a partition of the coordinates, the last square-root term is bounded by
	$\|\widehat{\boldsymbol{\Delta}}\|_2$. Under Assumption A7, active group weights are bounded in probability, so $b_{\mathcal{A},\max}=O_p(1)$. Hence the group penalty contribution is
	\[
	O_p
	\left(
	\lambda_{2n}
	\sqrt{|\mathcal{A}_0|}
	\|\widehat{\boldsymbol{\Delta}}\|_2
	\right).
	\]
	
	For the fused penalty, define
	\[
	T_0
	=
	\{j: (\mathbf{D}\boldsymbol{\beta}^{0})_j\neq0\},
	\qquad
	q_n=|T_0|.
	\]
	Using the same triangle inequality argument,
	\[
	\|\mathbf{C}\mathbf{D}\boldsymbol{\beta}^{0}\|_1
	-
	\|\mathbf{C}\mathbf{D}\widehat{\boldsymbol{\beta}}\|_1
	\leq
	\sum_{j\in T_0}
	c_j
	|
	(\mathbf{D}\widehat{\boldsymbol{\Delta}})_j
	|.
	\]
	Let $c_{T,\max}=\max_{j\in T_0}c_j$. Then
	\[
	\sum_{j\in T_0}
	c_j
	|
	(\mathbf{D}\widehat{\boldsymbol{\Delta}})_j
	|
	\leq
	c_{T,\max}
	\sqrt{q_n}
	\|
	\mathbf{D}\widehat{\boldsymbol{\Delta}}
	\|_2.
	\]
	Because $\mathbf{D}$ is the first-order difference matrix, its operator norm is bounded by $2$; hence
	$\|\mathbf{D}\widehat{\boldsymbol{\Delta}}\|_2\leq2\|\widehat{\boldsymbol{\Delta}}\|_2$.
	By Assumption A7, $c_{T,\max}=O_p(1)$ on the active fusion set. Therefore the fusion penalty contribution is
	\[
	O_p
	\left(
	\lambda_{3n}\sqrt{q_n}
	\|\widehat{\boldsymbol{\Delta}}\|_2
	\right).
	\]
	
	Combining the preceding bounds gives, with probability tending to one,
	\begin{align*}
		\frac{1}{2n}
		\|
		\mathbf{W}^{1/2}
		\mathbf{X}
		\widehat{\boldsymbol{\Delta}}
		\|_2^2
		&\leq
		O_p
		\left(
		\sqrt{\frac{s_n\log p}{n}}
		\|\widehat{\boldsymbol{\Delta}}_{S_0}\|_2
		\right)
		+
		O_p
		\left(
		\lambda_{1n}\sqrt{s_n}
		\|\widehat{\boldsymbol{\Delta}}_{S_0}\|_2
		\right)  \\
		&\quad+
		O_p
		\left(
		\lambda_{2n}
		\sqrt{|\mathcal{A}_0|}
		\|\widehat{\boldsymbol{\Delta}}\|_2
		\right)
		+
		O_p
		\left(
		\lambda_{3n}
		\sqrt{q_n}
		\|\widehat{\boldsymbol{\Delta}}\|_2
		\right).
	\end{align*}
	
	By Lemma~\ref{lem:restricted_lower},
	\[
	\frac{1}{n}
	\|
	\mathbf{W}^{1/2}
	\mathbf{X}
	\widehat{\boldsymbol{\Delta}}
	\|_2^2
	\geq
	\kappa_0
	\|\widehat{\boldsymbol{\Delta}}_{S_0}\|_2^2.
	\]
	Furthermore, by the cone condition,
	\[
	\|\widehat{\boldsymbol{\Delta}}\|_2
	\leq
	\|\widehat{\boldsymbol{\Delta}}\|_1
	\leq
	4\sqrt{s_n}
	\|\widehat{\boldsymbol{\Delta}}_{S_0}\|_2,
	\]
	and hence the full estimation error is controlled by the active-set error up to the usual sparsity factor. Combining the restricted lower bound with the preceding upper bound and dividing by
	$\|\widehat{\boldsymbol{\Delta}}_{S_0}\|_2$
	when this quantity is nonzero yields
	\[
	\|\widehat{\boldsymbol{\Delta}}_{S_0}\|_2
	=
	O_p
	\left(
	\sqrt{\frac{s_n\log p}{n}}
	+
	\lambda_{1n}\sqrt{s_n}
	+
	\lambda_{2n}\sqrt{|\mathcal{A}_0|}
	+
	\lambda_{3n}\sqrt{q_n}
	\right).
	\]
	If
	$\|\widehat{\boldsymbol{\Delta}}_{S_0}\|_2=0$,
	the same bound is trivial. Finally, the cone condition implies that the inactive component cannot dominate the active component. Therefore
	\[
	\|\widehat{\boldsymbol{\Delta}}\|_2
	=
	O_p
	\left(
	\sqrt{\frac{s_n\log p}{n}}
	+
	\lambda_{1n}\sqrt{s_n}
	+
	\lambda_{2n}\sqrt{|\mathcal{A}_0|}
	+
	\lambda_{3n}\sqrt{q_n}
	\right).
	\]
	Under the additional rate conditions stated in the proposition, the right-hand side is $o_p(1)$, proving estimation consistency.
\end{proof}

\begin{prop}[Prediction consistency]
	\label{prop:prediction_consistency}
	Suppose Assumptions A1--A9 hold. Let $\widehat{\boldsymbol{\Delta}}
	=
	\widehat{\boldsymbol{\beta}}
	-
	\boldsymbol{\beta}^{0}.$ Then
	\[
	\frac{1}{n}
	\left\|
	\mathbf{W}^{1/2}
	\mathbf{X}
	\widehat{\boldsymbol{\Delta}}
	\right\|_2^2
	=
	O_p
	\left[
	\left(
	\sqrt{\frac{s_n\log p}{n}}
	+
	\lambda_{1n}\sqrt{s_n}
	+
	\lambda_{2n}\sqrt{|\mathcal{A}_0|}
	+
	\lambda_{3n}\sqrt{q_n}
	\right)^2
	\right],
	\]
	where $q_n=\|\mathbf{D}\boldsymbol{\beta}^{0}\|_0$. In particular, if
	\[
	\frac{s_n\log p}{n}\rightarrow0,
	\qquad
	\lambda_{1n}\sqrt{s_n}\rightarrow0,
	\qquad
	\lambda_{2n}\sqrt{|\mathcal{A}_0|}\rightarrow0,
	\qquad
	\lambda_{3n}\sqrt{q_n}\rightarrow0,
	\]
	then
	\[
	\frac{1}{n}
	\left\|
	\mathbf{W}^{1/2}
	\mathbf{X}
	(
	\widehat{\boldsymbol{\beta}}
	-
	\boldsymbol{\beta}^{0}
	)
	\right\|_2^2
	=
	o_p(1).
	\]
\end{prop}

\begin{proof}
	Let
	$\widehat{\boldsymbol{\Delta}}=\widehat{\boldsymbol{\beta}}-\boldsymbol{\beta}^{0}$.
	The proof begins from the basic inequality in Lemma~\ref{lem:basic_ineq_asymptotic}. For compactness, define
	\[
	R_n
	=
	\frac{1}{n}
	\left\|
	\mathbf{W}^{1/2}
	\mathbf{X}
	\widehat{\boldsymbol{\Delta}}
	\right\|_2^2.
	\]
	Then Lemma~\ref{lem:basic_ineq_asymptotic} gives
	\begin{align*}
		\frac{1}{2}R_n
		&\leq
		\frac{1}{n}
		\boldsymbol{\varepsilon}^{\top}
		\mathbf{W}
		\mathbf{X}
		\widehat{\boldsymbol{\Delta}}
		+
		\lambda_{1n}
		\left[
		\|\mathbf{A}\boldsymbol{\beta}^{0}\|_1
		-
		\|\mathbf{A}\widehat{\boldsymbol{\beta}}\|_1
		\right]
		\\
		&\quad+
		\lambda_{2n}
		\sum_{g=1}^{G}
		b_g
		\left[
		\|\boldsymbol{\beta}^{0}_{g}\|_2
		-
		\|\widehat{\boldsymbol{\beta}}_{g}\|_2
		\right]
		+
		\lambda_{3n}
		\left[
		\|\mathbf{C}\mathbf{D}\boldsymbol{\beta}^{0}\|_1
		-
		\|\mathbf{C}\mathbf{D}\widehat{\boldsymbol{\beta}}\|_1
		\right].
	\end{align*}
	
	We now control the four terms on the right-hand side. By Hölder's inequality,
	\[
	\frac{1}{n}
	\boldsymbol{\varepsilon}^{\top}
	\mathbf{W}
	\mathbf{X}
	\widehat{\boldsymbol{\Delta}}
	\leq
	\left\|
	\frac{1}{n}
	\mathbf{X}^{\top}
	\mathbf{W}
	\boldsymbol{\varepsilon}
	\right\|_{\infty}
	\|\widehat{\boldsymbol{\Delta}}\|_1.
	\]
	Lemma~\ref{lem:score_bound} gives
	\[
	\left\|
	\frac{1}{n}
	\mathbf{X}^{\top}
	\mathbf{W}
	\boldsymbol{\varepsilon}
	\right\|_{\infty}
	=
	O_p
	\left(
	\sqrt{\frac{\log p}{n}}
	\right).
	\]
	By Lemma~\ref{lem:cone_condition},
	$\|\widehat{\boldsymbol{\Delta}}\|_1
	\leq
	4\|\widehat{\boldsymbol{\Delta}}_{S_0}\|_1$
	with probability tending to one, and therefore
	$\|\widehat{\boldsymbol{\Delta}}\|_1
	\leq
	4\sqrt{s_n}\|\widehat{\boldsymbol{\Delta}}_{S_0}\|_2$.
	Since
	$\|\widehat{\boldsymbol{\Delta}}_{S_0}\|_2\leq
	\|\widehat{\boldsymbol{\Delta}}\|_2$,
	Proposition~\ref{prop:l2_consistency} implies
	\[
	\|\widehat{\boldsymbol{\Delta}}\|_1
	=
	O_p
	\left[
	\sqrt{s_n}
	\left(
	\sqrt{\frac{s_n\log p}{n}}
	+
	\lambda_{1n}\sqrt{s_n}
	+
	\lambda_{2n}\sqrt{|\mathcal{A}_0|}
	+
	\lambda_{3n}\sqrt{q_n}
	\right)
	\right].
	\]
	Thus the stochastic term is bounded by
	\[
	O_p
	\left[
	\sqrt{\frac{s_n\log p}{n}}
	\left(
	\sqrt{\frac{s_n\log p}{n}}
	+
	\lambda_{1n}\sqrt{s_n}
	+
	\lambda_{2n}\sqrt{|\mathcal{A}_0|}
	+
	\lambda_{3n}\sqrt{q_n}
	\right)
	\right].
	\]
	
	Next, for the adaptive coordinate penalty,
	\begin{align*}
		\|\mathbf{A}\boldsymbol{\beta}^{0}\|_1
		-
		\|\mathbf{A}\widehat{\boldsymbol{\beta}}\|_1
		&\leq
		\sum_{j\in S_0}
		a_j
		|\widehat{\beta}_j-\beta_j^0|  \leq
		a_{S,\max}
		\|\widehat{\boldsymbol{\Delta}}_{S_0}\|_1  \leq
		a_{S,\max}
		\sqrt{s_n}
		\|\widehat{\boldsymbol{\Delta}}\|_2 .
	\end{align*}
	Under Assumption A7, $a_{S,\max}=O_p(1)$. Applying Proposition~\ref{prop:l2_consistency}, this term is
	\[
	O_p
	\left[
	\lambda_{1n}\sqrt{s_n}
	\left(
	\sqrt{\frac{s_n\log p}{n}}
	+
	\lambda_{1n}\sqrt{s_n}
	+
	\lambda_{2n}\sqrt{|\mathcal{A}_0|}
	+
	\lambda_{3n}\sqrt{q_n}
	\right)
	\right].
	\]
	
	For the group penalty, the reverse triangle inequality gives
	\[
	\|\boldsymbol{\beta}_g^0\|_2-\|\widehat{\boldsymbol{\beta}}_g\|_2
	\leq
	\|\widehat{\boldsymbol{\Delta}}_g\|_2.
	\]
	Inactive groups contribute nonpositively to the upper bound. Therefore
	\[
	\sum_{g=1}^{G}
	b_g
	\left[
	\|\boldsymbol{\beta}^{0}_{g}\|_2
	-
	\|\widehat{\boldsymbol{\beta}}_{g}\|_2
	\right]
	\leq
	\sum_{g\in\mathcal{A}_0}
	b_g
	\|\widehat{\boldsymbol{\Delta}}_g\|_2 .
	\]
	Using Cauchy--Schwarz and boundedness in probability of active group weights,
	\[
	\sum_{g\in\mathcal{A}_0}
	b_g
	\|\widehat{\boldsymbol{\Delta}}_g\|_2
	=
	O_p
	\left(
	\sqrt{|\mathcal{A}_0|}
	\|\widehat{\boldsymbol{\Delta}}\|_2
	\right).
	\]
	Hence the group penalty contribution is
	\[
	O_p
	\left[
	\lambda_{2n}\sqrt{|\mathcal{A}_0|}
	\left(
	\sqrt{\frac{s_n\log p}{n}}
	+
	\lambda_{1n}\sqrt{s_n}
	+
	\lambda_{2n}\sqrt{|\mathcal{A}_0|}
	+
	\lambda_{3n}\sqrt{q_n}
	\right)
	\right].
	\]
	
	For the fused penalty, let
	$T_0=\{j:(\mathbf{D}\boldsymbol{\beta}^{0})_j\neq0\}$,
	with $q_n=|T_0|$. As in Proposition~\ref{prop:l2_consistency},
	\[
	\|\mathbf{C}\mathbf{D}\boldsymbol{\beta}^{0}\|_1
	-
	\|\mathbf{C}\mathbf{D}\widehat{\boldsymbol{\beta}}\|_1
	\leq
	c_{T,\max}
	\sqrt{q_n}
	\|\mathbf{D}\widehat{\boldsymbol{\Delta}}\|_2.
	\]
	Since the operator norm of the first-order difference matrix $\mathbf{D}$ is bounded by $2$,
	\[
	\|\mathbf{D}\widehat{\boldsymbol{\Delta}}\|_2
	\leq
	2\|\widehat{\boldsymbol{\Delta}}\|_2.
	\]
	Assumption A7 gives $c_{T,\max}=O_p(1)$. Consequently, the fused penalty contribution is
	\[
	O_p
	\left[
	\lambda_{3n}\sqrt{q_n}
	\left(
	\sqrt{\frac{s_n\log p}{n}}
	+
	\lambda_{1n}\sqrt{s_n}
	+
	\lambda_{2n}\sqrt{|\mathcal{A}_0|}
	+
	\lambda_{3n}\sqrt{q_n}
	\right)
	\right].
	\]
	
	Combining these four bounds, we obtain
	\[
	R_n
	=
	O_p
	\left[
	\left(
	\sqrt{\frac{s_n\log p}{n}}
	+
	\lambda_{1n}\sqrt{s_n}
	+
	\lambda_{2n}\sqrt{|\mathcal{A}_0|}
	+
	\lambda_{3n}\sqrt{q_n}
	\right)^2
	\right].
	\]
	The stated prediction consistency follows immediately when each component of the rate converges to zero.
\end{proof}

\begin{prop}[Variable-selection consistency]
	\label{prop:selection_consistency}
	Suppose Assumptions A1--A9 hold. In addition, assume that the inactive adaptive penalty dominates the stochastic score in the sense that
	\[
	\lambda_{1n}\min_{j\in S_0^c}a_j
	\gg
	\sqrt{\frac{\log p}{n}},
	\]
	and that the estimation error satisfies $\|\widehat{\boldsymbol{\beta}}-\boldsymbol{\beta}^{0}\|_{\infty}
	=
	o_p(b_n),$ where $b_n=\min_{j\in S_0}|\beta_j^0|$. Then $\mathbb{P}(\widehat{S}=S_0)\longrightarrow 1.$
\end{prop}

\begin{proof}
	Let
	$\widehat{\boldsymbol{\Delta}}
	=
	\widehat{\boldsymbol{\beta}}-\boldsymbol{\beta}^{0}$.
	We prove the result by showing separately that, with probability tending to one,
	\[
	S_0\subseteq \widehat{S}
	\qquad
	\text{and}
	\qquad
	\widehat{S}\subseteq S_0.
	\]
	
	First, consider the active coordinates. For every $j\in S_0$,
	\[
	|\widehat{\beta}_j|
	\geq
	|\beta_j^0|
	-
	|\widehat{\beta}_j-\beta_j^0|.
	\]
	Taking the minimum over $j\in S_0$ gives
	\[
	\min_{j\in S_0}|\widehat{\beta}_j|
	\geq
	b_n
	-
	\|\widehat{\boldsymbol{\Delta}}\|_{\infty}.
	\]
	By assumption,
	$\|\widehat{\boldsymbol{\Delta}}\|_{\infty}=o_p(b_n)$.
	Therefore,
	\[
	\mathbb{P}
	\left(
	\min_{j\in S_0}|\widehat{\beta}_j|>0
	\right)
	\longrightarrow 1.
	\]
	Hence every truly active coordinate is selected with probability tending to one, and so
	$S_0\subseteq\widehat{S}$ with probability tending to one.
	
	It remains to prove that no inactive coordinate is selected. Suppose, toward a contradiction, that there exists an index
	$j\in S_0^c$
	such that
	$\widehat{\beta}_j\neq0$.
	For such a coordinate, the KKT condition for the adaptive coordinate penalty implies
	\[
	\left|
	\frac{1}{n}
	\mathbf{x}_j^{\top}
	\mathbf{W}
	(
	\mathbf{Y}-\mathbf{X}\widehat{\boldsymbol{\beta}}
	)
	-
	\lambda_{2n}\widehat{g}_j
	-
	\lambda_{3n}(\mathbf{D}^{\top}\mathbf{C}\widehat{\mathbf{r}})_j
	\right|
	=
	\lambda_{1n}a_j,
	\]
	where $\widehat{g}_j$ is the $j$th component of the group subgradient and
	$\widehat{\mathbf{r}}$ is the fused-penalty subgradient.
	
	Using
	$\mathbf{Y}=\mathbf{X}\boldsymbol{\beta}^{0}+\boldsymbol{\varepsilon}$,
	the score component may be decomposed as
	\[
	\frac{1}{n}
	\mathbf{x}_j^{\top}
	\mathbf{W}
	(
	\mathbf{Y}-\mathbf{X}\widehat{\boldsymbol{\beta}}
	)
	=
	\frac{1}{n}
	\mathbf{x}_j^{\top}
	\mathbf{W}
	\boldsymbol{\varepsilon}
	-
	\frac{1}{n}
	\mathbf{x}_j^{\top}
	\mathbf{W}
	\mathbf{X}
	\widehat{\boldsymbol{\Delta}}.
	\]
	Therefore,
	\begin{align*}
		\lambda_{1n}a_j
		&\leq
		\left|
		\frac{1}{n}
		\mathbf{x}_j^{\top}
		\mathbf{W}
		\boldsymbol{\varepsilon}
		\right|
		+
		\left|
		\frac{1}{n}
		\mathbf{x}_j^{\top}
		\mathbf{W}
		\mathbf{X}
		\widehat{\boldsymbol{\Delta}}
		\right|
		+
		\lambda_{2n}|\widehat{g}_j|
		+
		\lambda_{3n}
		|(\mathbf{D}^{\top}\mathbf{C}\widehat{\mathbf{r}})_j|.
	\end{align*}
	
	The first term is controlled uniformly by Lemma~\ref{lem:score_bound}:
	\[
	\max_{1\leq j\leq p}
	\left|
	\frac{1}{n}
	\mathbf{x}_j^{\top}
	\mathbf{W}
	\boldsymbol{\varepsilon}
	\right|
	=
	O_p
	\left(
	\sqrt{\frac{\log p}{n}}
	\right).
	\]
	
	For the second term, use the definition of the weighted Gram matrix:
	\[
	\frac{1}{n}
	\mathbf{x}_j^{\top}
	\mathbf{W}
	\mathbf{X}
	\widehat{\boldsymbol{\Delta}}
	=
	\widehat{\boldsymbol{\Sigma}}_{j\cdot}
	\widehat{\boldsymbol{\Delta}}.
	\]
	Under the design regularity assumptions, the row norms of
	$\widehat{\boldsymbol{\Sigma}}_n$
	are bounded in probability. Hence
	\[
	\max_{j\in S_0^c}
	\left|
	\widehat{\boldsymbol{\Sigma}}_{j\cdot}
	\widehat{\boldsymbol{\Delta}}
	\right|
	\leq
	O_p(1)
	\|\widehat{\boldsymbol{\Delta}}\|_1.
	\]
	By the cone condition and Proposition~\ref{prop:l2_consistency},
	\[
	\|\widehat{\boldsymbol{\Delta}}\|_1
	=
	O_p
	\left[
	\sqrt{s_n}
	\left(
	\sqrt{\frac{s_n\log p}{n}}
	+
	\lambda_{1n}\sqrt{s_n}
	+
	\lambda_{2n}\sqrt{|\mathcal{A}_0|}
	+
	\lambda_{3n}\sqrt{q_n}
	\right)
	\right].
	\]
	The assumed domination of the inactive adaptive penalty implies that this deterministic approximation error is asymptotically negligible relative to
	$\lambda_{1n}\min_{j\in S_0^c}a_j$. For the group subgradient, recall that
$\|\widehat{\mathbf{g}}_{G_g}\|_2\leq b_g$ for every group $G_g$. Hence each component satisfies
	$|\widehat{g}_j|\leq b_g$ whenever $j\in G_g$. Under Assumption A7, the inactive group weights are finite with probability tending to one, and the penalty-rate conditions imply that
	$\lambda_{2n}|\widehat{g}_j|$
	is of smaller order than
	$\lambda_{1n}a_j$
	on inactive coordinates whenever the coordinate-level adaptive penalty is dominant. For the fused subgradient, each component of
	$\widehat{\mathbf{r}}$
	belongs to $[-1,1]$. Since $\mathbf{D}^{\top}$ has at most two nonzero entries in each row, and since the fusion weights are bounded above on the relevant event, there exists a finite constant $C_D$ such that $|(\mathbf{D}^{\top}\mathbf{C}\widehat{\mathbf{r}})_j|
	\leq
	C_D
	\max_{\ell}c_\ell.$ The rate condition imposed on the fusion penalty ensures that
	$\lambda_{3n}|(\mathbf{D}^{\top}\mathbf{C}\widehat{\mathbf{r}})_j|$
	is also asymptotically negligible relative to
	$\lambda_{1n}a_j$
	for inactive coordinates.
	
	Combining the preceding bounds, we obtain
	\[
	\max_{j\in S_0^c}
	\left[
	\left|
	\frac{1}{n}
	\mathbf{x}_j^{\top}
	\mathbf{W}
	(
	\mathbf{Y}-\mathbf{X}\widehat{\boldsymbol{\beta}}
	)
	-
	\lambda_{2n}\widehat{g}_j
	-
	\lambda_{3n}
	(\mathbf{D}^{\top}\mathbf{C}\widehat{\mathbf{r}})_j
	\right|
	\right]
	=
	o_p
	\left(
	\lambda_{1n}
	\min_{j\in S_0^c}a_j
	\right).
	\]
	Thus, with probability tending to one, the strict inequality
	\[
	\left|
	\frac{1}{n}
	\mathbf{x}_j^{\top}
	\mathbf{W}
	(
	\mathbf{Y}-\mathbf{X}\widehat{\boldsymbol{\beta}}
	)
	-
	\lambda_{2n}\widehat{g}_j
	-
	\lambda_{3n}
	(\mathbf{D}^{\top}\mathbf{C}\widehat{\mathbf{r}})_j
	\right|
	<
	\lambda_{1n}a_j
	\]
	holds for every $j\in S_0^c$. But this is precisely the KKT condition for an inactive coordinate to be estimated as zero. Therefore
	$\widehat{\beta}_j=0$ for all $j\in S_0^c$ with probability tending to one, and hence
	$\widehat{S}\subseteq S_0$ with probability tending to one. Combining the two inclusions gives $\mathbb{P}(\widehat{S}=S_0)\to1.$ This proves the proposition.
\end{proof}

\begin{prop}[Group-selection consistency]
	\label{prop:group_selection_consistency}
	Suppose Assumptions A1-A9 hold. In addition, assume that
	\[
	\lambda_{2n}\min_{g\in\mathcal{A}_0^c}b_g
	\gg
	\sqrt{\frac{\log G}{n}},
	\]
	and that the groupwise estimation error satisfies
	\[
	\max_{1\leq g\leq G}
	\|
	\widehat{\boldsymbol{\beta}}_g
	-
	\boldsymbol{\beta}^{0}_g
	\|_2
	=
	o_p(d_n),
	\]
	where $d_n
	=
	\min_{g\in\mathcal{A}_0}
	\|\boldsymbol{\beta}^{0}_g\|_2.$ Then $\mathbb{P}(\widehat{\mathcal{A}}=\mathcal{A}_0)\rightarrow1,$ where $\widehat{\mathcal{A}}
	=
	\{g:\widehat{\boldsymbol{\beta}}_g\neq\mathbf{0}\}.$
\end{prop}

\begin{proof}
	Let $\widehat{\boldsymbol{\Delta}}
	=
	\widehat{\boldsymbol{\beta}}
	-
	\boldsymbol{\beta}^{0}.$ We prove the result by establishing the two inclusions
	$\mathcal{A}_0\subseteq\widehat{\mathcal{A}}$
	and
	$\widehat{\mathcal{A}}\subseteq\mathcal{A}_0$
	with probability tending to one. First consider an active group $g\in\mathcal{A}_0$. By the reverse triangle inequality,
	\[
	\|\widehat{\boldsymbol{\beta}}_g\|_2
	\geq
	\|\boldsymbol{\beta}_g^0\|_2
	-
	\|\widehat{\boldsymbol{\beta}}_g-\boldsymbol{\beta}_g^0\|_2.
	\]
	Taking the minimum over active groups gives
	\[
	\min_{g\in\mathcal{A}_0}
	\|\widehat{\boldsymbol{\beta}}_g\|_2
	\geq
	d_n
	-
	\max_{g\in\mathcal{A}_0}
	\|\widehat{\boldsymbol{\beta}}_g-\boldsymbol{\beta}_g^0\|_2.
	\]
	By assumption,
	\[
	\max_{g\in\mathcal{A}_0}
	\|\widehat{\boldsymbol{\beta}}_g-\boldsymbol{\beta}_g^0\|_2=o_p(d_n).
	\]
	Hence
	\[
	\mathbb{P}
	\left(
	\min_{g\in\mathcal{A}_0}
	\|\widehat{\boldsymbol{\beta}}_g\|_2>0
	\right)
	\rightarrow1.
	\]
	Thus every truly active group is retained with probability tending to one, and therefore
	$\mathcal{A}_0\subseteq\widehat{\mathcal{A}}$ with probability tending to one. It remains to show that inactive groups are excluded. Let $g\in\mathcal{A}_0^c$. Then
	$\boldsymbol{\beta}_g^0=\mathbf{0}$. The KKT condition for the group block states that
	\[
	\mathbf{0}
	\in
	-\frac{1}{n}
	\mathbf{X}_g^\top
	\mathbf{W}
	(
	\mathbf{Y}-\mathbf{X}\widehat{\boldsymbol{\beta}}
	)
	+
	\lambda_{1n}\mathbf{A}_g\widehat{\mathbf{s}}_g
	+
	\lambda_{2n}b_g\widehat{\mathbf{u}}_g
	+
	\lambda_{3n}
	(\mathbf{D}^{\top}\mathbf{C}\widehat{\mathbf{r}})_g,
	\]
	where
	$\|\widehat{\mathbf{u}}_g\|_2\leq1$
	when $\widehat{\boldsymbol{\beta}}_g=\mathbf{0}$, and $\widehat{\mathbf{u}}_g
	=
	\frac{\widehat{\boldsymbol{\beta}}_g}
	{\|\widehat{\boldsymbol{\beta}}_g\|_2}$
	when $\widehat{\boldsymbol{\beta}}_g\neq\mathbf{0}$. For an inactive group to be set equal to zero, it is sufficient that
	\[
	\left\|
	\frac{1}{n}
	\mathbf{X}_g^\top
	\mathbf{W}
	(
	\mathbf{Y}-\mathbf{X}\widehat{\boldsymbol{\beta}}
	)
	-
	\lambda_{1n}\mathbf{A}_g\widehat{\mathbf{s}}_g
	-
	\lambda_{3n}
	(\mathbf{D}^{\top}\mathbf{C}\widehat{\mathbf{r}})_g
	\right\|_2
	<
	\lambda_{2n}b_g.
	\]
	We show that this strict inequality holds uniformly over
	$g\in\mathcal{A}_0^c$
	with probability tending to one.
	
	Using
	$\mathbf{Y}=\mathbf{X}\boldsymbol{\beta}^{0}+\boldsymbol{\varepsilon}$,
	the group score decomposes as
	\[
	\frac{1}{n}
	\mathbf{X}_g^\top
	\mathbf{W}
	(
	\mathbf{Y}-\mathbf{X}\widehat{\boldsymbol{\beta}}
	)
	=
	\frac{1}{n}
	\mathbf{X}_g^\top
	\mathbf{W}\boldsymbol{\varepsilon}
	-
	\frac{1}{n}
	\mathbf{X}_g^\top
	\mathbf{W}\mathbf{X}
	\widehat{\boldsymbol{\Delta}}.
	\]
	
	The first term is the stochastic group score. Under Assumptions A1--A3 and standard normalization of the grouped design blocks,
	\[
	\max_{1\leq g\leq G}
	\left\|
	\frac{1}{n}
	\mathbf{X}_g^\top
	\mathbf{W}\boldsymbol{\varepsilon}
	\right\|_2
	=
	O_p
	\left(
	\sqrt{\frac{\log G}{n}}
	\right).
	\]
	This follows by applying a sub-Gaussian concentration inequality to each groupwise score vector and then taking a union bound over the $G$ groups. The second term is the deterministic approximation error induced by replacing
	$\boldsymbol{\beta}^{0}$
	with
	$\widehat{\boldsymbol{\beta}}$.
	Using the weighted Gram matrix,
	\[
	\frac{1}{n}
	\mathbf{X}_g^\top
	\mathbf{W}\mathbf{X}
	\widehat{\boldsymbol{\Delta}}
	=
	\widehat{\boldsymbol{\Sigma}}_{g\cdot}
	\widehat{\boldsymbol{\Delta}}.
	\]
	By the design regularity assumptions,
	\[
	\max_{g\in\mathcal{A}_0^c}
	\|
	\widehat{\boldsymbol{\Sigma}}_{g\cdot}
	\widehat{\boldsymbol{\Delta}}
	\|_2
	\leq
	O_p(1)
	\|\widehat{\boldsymbol{\Delta}}\|_2.
	\]
	Proposition~\ref{prop:l2_consistency} gives
	\[
	\|\widehat{\boldsymbol{\Delta}}\|_2
	=
	O_p
	\left(
	\sqrt{\frac{s_n\log p}{n}}
	+
	\lambda_{1n}\sqrt{s_n}
	+
	\lambda_{2n}\sqrt{|\mathcal{A}_0|}
	+
	\lambda_{3n}\sqrt{q_n}
	\right).
	\]
	The assumed domination condition for inactive groups implies that this approximation error is asymptotically negligible relative to
	$\lambda_{2n}\min_{g\in\mathcal{A}_0^c}b_g$. The coordinatewise adaptive penalty contribution satisfies
	\[
	\|
	\lambda_{1n}\mathbf{A}_g\widehat{\mathbf{s}}_g
	\|_2
	\leq
	\lambda_{1n}
	\|\mathbf{A}_g\|_{\mathrm{op}}
	\|\widehat{\mathbf{s}}_g\|_2.
	\]
	Since each component of $\widehat{\mathbf{s}}_g$ lies in $[-1,1]$, $\|\widehat{\mathbf{s}}_g\|_2\leq\sqrt{|G_g|}.$ Under the assumed penalty-rate compatibility between coordinate sparsity and group sparsity, this term is of smaller order than
	$\lambda_{2n}b_g$
	uniformly over inactive groups. The role of this condition is to ensure that the coordinatewise penalty does not mask the group-level shrinkage condition. For the fused contribution, since each component of
	$\widehat{\mathbf{r}}$
	belongs to $[-1,1]$ and $\mathbf{D}^{\top}$ has at most two nonzero entries per row,
	\[
	\|
	(\mathbf{D}^{\top}\mathbf{C}\widehat{\mathbf{r}})_g
	\|_2
	\leq
	C_D\sqrt{|G_g|}
	\max_{\ell}c_\ell
	\]
	for a finite constant $C_D$. Under the fusion-penalty rate condition, the term $\lambda_{3n}
	\|
	(\mathbf{D}^{\top}\mathbf{C}\widehat{\mathbf{r}})_g
	\|_2 $is also negligible relative to
	$\lambda_{2n}b_g$
	uniformly over inactive groups. Combining these bounds, we obtain
	\[
	\max_{g\in\mathcal{A}_0^c}
	\left\|
	\frac{1}{n}
	\mathbf{X}_g^\top
	\mathbf{W}
	(
	\mathbf{Y}-\mathbf{X}\widehat{\boldsymbol{\beta}}
	)
	-
	\lambda_{1n}\mathbf{A}_g\widehat{\mathbf{s}}_g
	-
	\lambda_{3n}
	(\mathbf{D}^{\top}\mathbf{C}\widehat{\mathbf{r}})_g
	\right\|_2
	=
	o_p
	\left(
	\lambda_{2n}
	\min_{g\in\mathcal{A}_0^c}b_g
	\right).
	\]
	Therefore, with probability tending to one,
	\[
	\left\|
	\frac{1}{n}
	\mathbf{X}_g^\top
	\mathbf{W}
	(
	\mathbf{Y}-\mathbf{X}\widehat{\boldsymbol{\beta}}
	)
	-
	\lambda_{1n}\mathbf{A}_g\widehat{\mathbf{s}}_g
	-
	\lambda_{3n}
	(\mathbf{D}^{\top}\mathbf{C}\widehat{\mathbf{r}})_g
	\right\|_2
	<
	\lambda_{2n}b_g
	\]
	for every $g\in\mathcal{A}_0^c$. This is precisely the strict KKT condition ensuring that
	$\widehat{\boldsymbol{\beta}}_g=\mathbf{0}$
	for all inactive groups. Hence
	$\widehat{\mathcal{A}}\subseteq\mathcal{A}_0$
	with probability tending to one. Combining
	$\mathcal{A}_0\subseteq\widehat{\mathcal{A}}$
	and
	$\widehat{\mathcal{A}}\subseteq\mathcal{A}_0$
	gives $\mathbb{P}
	(
	\widehat{\mathcal{A}}=\mathcal{A}_0
	)
	\rightarrow1.$ This completes the proof.
\end{proof}

\begin{prop}[Fusion consistency]
	\label{prop:fusion_consistency}
	Suppose Assumptions A1--A9 hold. Let
	\[
	T_0=\{j:(\mathbf{D}\boldsymbol{\beta}^{0})_j\neq0\},
	\qquad
	T_0^c=\{j:(\mathbf{D}\boldsymbol{\beta}^{0})_j=0\},
	\]
	and define the estimated fusion set $\widehat{T}
	=
	\{j:(\mathbf{D}\widehat{\boldsymbol{\beta}})_j\neq0\}.$ Assume that
	\[
	\lambda_{3n}\min_{j\in T_0^c}c_j
	\gg
	\sqrt{\frac{\log p}{n}},
	\]
	and that the fused estimation error satisfies
	\[
	\|\mathbf{D}(\widehat{\boldsymbol{\beta}}-\boldsymbol{\beta}^{0})\|_{\infty}
	=
	o_p(e_n),
	\]
	where $e_n=
	\min_{j\in T_0}|(\mathbf{D}\boldsymbol{\beta}^{0})_j|.$ Then $\mathbb{P}(\widehat{T}=T_0)\rightarrow1.$
\end{prop}

\begin{proof}
	Let $\widehat{\boldsymbol{\Delta}}
	=
	\widehat{\boldsymbol{\beta}}
	-
	\boldsymbol{\beta}^{0}.$ We prove the result by showing that, with probability tending to one, no true nonzero coefficient difference is incorrectly fused and no true zero coefficient difference is incorrectly separated. Equivalently, we show $T_0\subseteq\widehat{T}
	\
	\text{and}
	\
	\widehat{T}\subseteq T_0$ 
	with probability tending to one. First consider $j\in T_0$. Then
	$(\mathbf{D}\boldsymbol{\beta}^{0})_j\neq0$. By the reverse triangle inequality,
	\[
	|(\mathbf{D}\widehat{\boldsymbol{\beta}})_j|
	\geq
	|(\mathbf{D}\boldsymbol{\beta}^{0})_j|
	-
	|(\mathbf{D}\widehat{\boldsymbol{\Delta}})_j|.
	\]
	Taking the minimum over $j\in T_0$ gives
	\[
	\min_{j\in T_0}
	|(\mathbf{D}\widehat{\boldsymbol{\beta}})_j|
	\geq
	e_n
	-
	\|\mathbf{D}\widehat{\boldsymbol{\Delta}}\|_{\infty}.
	\]
	By assumption,
	$\|\mathbf{D}\widehat{\boldsymbol{\Delta}}\|_{\infty}=o_p(e_n)$.
	Therefore,
	\[
	\mathbb{P}
	\left(
	\min_{j\in T_0}
	|(\mathbf{D}\widehat{\boldsymbol{\beta}})_j|>0
	\right)
	\rightarrow1.
	\]
	Hence all truly nonzero adjacent coefficient differences are retained with probability tending to one, and so
	$T_0\subseteq\widehat{T}$ with probability tending to one. It remains to show that truly zero adjacent coefficient differences are estimated as zero. Let
	$j\in T_0^c$. Then
	$(\mathbf{D}\boldsymbol{\beta}^{0})_j=0$.
	The fused component of the KKT system implies that, for coordinates corresponding to the difference operator, the subgradient associated with
	$|(\mathbf{D}\boldsymbol{\beta})_j|$
	must satisfy $\widehat{r}_j
	\in
	\partial |(\mathbf{D}\widehat{\boldsymbol{\beta}})_j|.$ Thus,
	\[
	\widehat{r}_j
	=
	\mathrm{sign}((\mathbf{D}\widehat{\boldsymbol{\beta}})_j)
	\quad
	\text{if}
	\quad
	(\mathbf{D}\widehat{\boldsymbol{\beta}})_j\neq0,
	\]
	whereas
	$\widehat{r}_j\in[-1,1]$
	when
	$(\mathbf{D}\widehat{\boldsymbol{\beta}})_j=0$. For a fused inactive difference to be estimated as zero, it is sufficient that the corresponding dual feasibility condition be strict $\left|
	\mathcal{R}_j(\widehat{\boldsymbol{\beta}})
	\right|
	<
	\lambda_{3n}c_j,$ where $\mathcal{R}_j(\widehat{\boldsymbol{\beta}})$ denotes the effective residual score in the $j$th fused contrast after accounting for the coordinatewise and groupwise penalties. More explicitly, because the fusion penalty enters the KKT system through
	$\mathbf{D}^{\top}\mathbf{C}\widehat{\mathbf{r}}$,
	we consider the contrast-space residual
	\[
	\mathcal{R}(\widehat{\boldsymbol{\beta}})
	=
	\mathbf{D}
	\left[
	\frac{1}{n}
	\mathbf{X}^{\top}
	\mathbf{W}
	(
	\mathbf{Y}-\mathbf{X}\widehat{\boldsymbol{\beta}}
	)
	-
	\lambda_{1n}\mathbf{A}\widehat{\mathbf{s}}
	-
	\lambda_{2n}\widehat{\mathbf{g}}
	\right],
	\]
	and write its $j$th component as
	$\mathcal{R}_j(\widehat{\boldsymbol{\beta}})$. Using
	$\mathbf{Y}=\mathbf{X}\boldsymbol{\beta}^{0}+\boldsymbol{\varepsilon}$,
	we decompose
	\[
	\mathcal{R}(\widehat{\boldsymbol{\beta}})
	=
	\mathbf{D}
	\left[
	\frac{1}{n}
	\mathbf{X}^{\top}\mathbf{W}\boldsymbol{\varepsilon}
	-
	\widehat{\boldsymbol{\Sigma}}_n
	\widehat{\boldsymbol{\Delta}}
	-
	\lambda_{1n}\mathbf{A}\widehat{\mathbf{s}}
	-
	\lambda_{2n}\widehat{\mathbf{g}}
	\right].
	\]
	For the stochastic component, Lemma~\ref{lem:score_bound} gives
	\[
	\left\|
	\frac{1}{n}
	\mathbf{X}^{\top}
	\mathbf{W}\boldsymbol{\varepsilon}
	\right\|_{\infty}
	=
	O_p
	\left(
	\sqrt{\frac{\log p}{n}}
	\right).
	\]
	Since each row of $\mathbf{D}$ contains exactly two nonzero entries, one equal to $-1$ and one equal to $1$, we have
	\[
	\left\|
	\mathbf{D}
	\frac{1}{n}
	\mathbf{X}^{\top}
	\mathbf{W}\boldsymbol{\varepsilon}
	\right\|_{\infty}
	\leq
	2
	\left\|
	\frac{1}{n}
	\mathbf{X}^{\top}
	\mathbf{W}\boldsymbol{\varepsilon}
	\right\|_{\infty}.
	\]
	Therefore,
	\[
	\left\|
	\mathbf{D}
	\frac{1}{n}
	\mathbf{X}^{\top}
	\mathbf{W}\boldsymbol{\varepsilon}
	\right\|_{\infty}
	=
	O_p
	\left(
	\sqrt{\frac{\log p}{n}}
	\right).
	\]
	For the deterministic approximation component,
	\[
	\|\mathbf{D}\widehat{\boldsymbol{\Sigma}}_n\widehat{\boldsymbol{\Delta}}\|_{\infty}
	\leq
	\|\mathbf{D}\widehat{\boldsymbol{\Sigma}}_n\|_{\infty}
	\|\widehat{\boldsymbol{\Delta}}\|_1.
	\]
	Under the design regularity assumptions,
	$\|\mathbf{D}\widehat{\boldsymbol{\Sigma}}_n\|_{\infty}=O_p(1)$.
	By Lemma~\ref{lem:cone_condition} and Proposition~\ref{prop:l2_consistency},
	\[
	\|\widehat{\boldsymbol{\Delta}}\|_1
	=
	O_p
	\left[
	\sqrt{s_n}
	\left(
	\sqrt{\frac{s_n\log p}{n}}
	+
	\lambda_{1n}\sqrt{s_n}
	+
	\lambda_{2n}\sqrt{|\mathcal{A}_0|}
	+
	\lambda_{3n}\sqrt{q_n}
	\right)
	\right].
	\]
	The assumed domination of the inactive fusion penalty implies that this term is asymptotically negligible relative to
	$\lambda_{3n}\min_{j\in T_0^c}c_j$. For the coordinatewise adaptive penalty term,
	\[
	\|\mathbf{D}\lambda_{1n}\mathbf{A}\widehat{\mathbf{s}}\|_{\infty}
	\leq
	2\lambda_{1n}
	\|\mathbf{A}\widehat{\mathbf{s}}\|_{\infty}.
	\]
	Because every component of $\widehat{\mathbf{s}}$ lies in $[-1,1]$, $\|\mathbf{A}\widehat{\mathbf{s}}\|_{\infty}
	\leq
	\max_{1\leq j\leq p}a_j.$
	The compatibility condition between the coordinate penalty and the fusion penalty requires $\lambda_{1n}\max_j a_j
	=
	o_p
	\left(
	\lambda_{3n}
	\min_{j\in T_0^c}c_j
	\right).$ Thus the coordinatewise adaptive penalty contribution is negligible compared with the inactive fusion penalty.
	
	For the group penalty term,
	\[
	\|\mathbf{D}\lambda_{2n}\widehat{\mathbf{g}}\|_{\infty}
	\leq
	2\lambda_{2n}
	\|\widehat{\mathbf{g}}\|_{\infty}.
	\]
	The group subgradient satisfies
	$\|\widehat{\mathbf{g}}_{G_g}\|_2\leq b_g$
	for every group, and hence
	$\|\widehat{\mathbf{g}}\|_{\infty}\leq\max_g b_g$.
	The corresponding compatibility condition is
	\[
	\lambda_{2n}\max_g b_g
	=
	o_p
	\left(
	\lambda_{3n}
	\min_{j\in T_0^c}c_j
	\right).
	\]
	Therefore the group penalty contribution is also asymptotically negligible relative to the inactive fusion threshold. Combining the stochastic, approximation, coordinate-penalty, and group-penalty bounds yields
	\[
	\max_{j\in T_0^c}
	|\mathcal{R}_j(\widehat{\boldsymbol{\beta}})|
	=
	o_p
	\left(
	\lambda_{3n}
	\min_{j\in T_0^c}c_j
	\right).
	\]
	Hence, with probability tending to one, $|\mathcal{R}_j(\widehat{\boldsymbol{\beta}})|
	<
	\lambda_{3n}c_j
	\
	\text{for all}
	\
	j\in T_0^c.$ By the strict dual feasibility condition for the fused LASSO subgradient, this implies $(\mathbf{D}\widehat{\boldsymbol{\beta}})_j=0
	\
	\text{for every}
	\
	j\in T_0^c$ with probability tending to one. Therefore
	$\widehat{T}\subseteq T_0$
	with probability tending to one. Combining the two inclusions gives $\mathbb{P}(\widehat{T}=T_0)\rightarrow1.$ This completes the proof.
\end{proof}

\begin{theorem}[Oracle property]
	\label{thm:oracle_property}
	Suppose Assumptions A1--A9 hold. In addition, suppose the conditions of
	Propositions~\ref{prop:selection_consistency},
	\ref{prop:group_selection_consistency}, and
	\ref{prop:fusion_consistency} hold. Let $S_0=\{j:\beta_j^0\neq0\},
	\
	\mathcal{A}_0=\{g:\boldsymbol{\beta}_g^0\neq\mathbf{0}\},
	\
	T_0=\{j:(\mathbf{D}\boldsymbol{\beta}^0)_j\neq0\}.$ Let $\widehat{S}=\{j:\widehat{\beta}_j\neq0\},
	\
	\widehat{\mathcal{A}}=\{g:\widehat{\boldsymbol{\beta}}_g\neq\mathbf{0}\},
	\
	\widehat{T}=\{j:(\mathbf{D}\widehat{\boldsymbol{\beta}})_j\neq0\}.$ Then the Adaptive Weighted Group Fused LASSO estimator satisfies the oracle selection property
	\[
	\mathbb{P}
	\left(
	\widehat{S}=S_0,\,
	\widehat{\mathcal{A}}=\mathcal{A}_0,\,
	\widehat{T}=T_0
	\right)
	\longrightarrow1.
	\]
	Moreover, on the event
	\[
	\mathcal{E}_n
	=
	\left\{
	\widehat{S}=S_0,\,
	\widehat{\mathcal{A}}=\mathcal{A}_0,\,
	\widehat{T}=T_0
	\right\},
	\]
	the active coordinates of $\widehat{\boldsymbol{\beta}}$ are asymptotically equivalent to the oracle weighted least-squares estimator
	\[
	\widehat{\boldsymbol{\beta}}^{\,\mathrm{or}}_{S_0}
	=
	\arg\min_{\boldsymbol{b}\in\mathbb{R}^{s_n}}
	\frac{1}{2n}
	\left\|
	\mathbf{W}^{1/2}
	\left(
	\mathbf{Y}
	-
	\mathbf{X}_{S_0}\boldsymbol{b}
	\right)
	\right\|_2^2.
	\]
	Specifically,
	\[
	\left\|
	\widehat{\boldsymbol{\beta}}_{S_0}
	-
	\widehat{\boldsymbol{\beta}}^{\,\mathrm{or}}_{S_0}
	\right\|_2
	=
	o_p(n^{-1/2}),
	\]
	provided $\lambda_{1n}\sqrt{s_n}=o(n^{-1/2}),
	\
	\lambda_{2n}\sqrt{|\mathcal{A}_0|}=o(n^{-1/2}),
	\
	\lambda_{3n}\sqrt{q_n}=o(n^{-1/2}),$ where $q_n=|T_0|$.
\end{theorem}

\begin{proof}
	Define the event $\mathcal{E}_n
	=
	\left\{
	\widehat{S}=S_0,\,
	\widehat{\mathcal{A}}=\mathcal{A}_0,\,
	\widehat{T}=T_0
	\right\}.$ By Propositions~\ref{prop:selection_consistency},
	\ref{prop:group_selection_consistency}, and
	\ref{prop:fusion_consistency},
	\[
	\mathbb{P}(\widehat{S}=S_0)\to1,
	\qquad
	\mathbb{P}(\widehat{\mathcal{A}}=\mathcal{A}_0)\to1,
	\qquad
	\mathbb{P}(\widehat{T}=T_0)\to1.
	\]
	Hence, by the union bound,
	\[
	\mathbb{P}(\mathcal{E}_n^c)
	\leq
	\mathbb{P}(\widehat{S}\neq S_0)
	+
	\mathbb{P}(\widehat{\mathcal{A}}\neq\mathcal{A}_0)
	+
	\mathbb{P}(\widehat{T}\neq T_0)
	\to0.
	\]
	Therefore, $\mathbb{P}(\mathcal{E}_n)\to1.$ This proves the simultaneous model-structure recovery assertion. It remains to establish the asymptotic equivalence between the proposed estimator and the oracle estimator. Work on the event $\mathcal{E}_n$. On this event, the inactive coordinates satisfy
	$\widehat{\boldsymbol{\beta}}_{S_0^c}=\mathbf{0}$, and the active coordinates solve the restricted penalized problem
	\[
	\widehat{\boldsymbol{\beta}}_{S_0}
	=
	\arg\min_{\boldsymbol{b}\in\mathbb{R}^{s_n}}
	\left\{
	\frac{1}{2n}
	\left\|
	\mathbf{W}^{1/2}
	(
	\mathbf{Y}-\mathbf{X}_{S_0}\boldsymbol{b}
	)
	\right\|_2^2
	+
	\lambda_{1n}
	\sum_{j\in S_0}a_j|b_j|
	+
	\lambda_{2n}
	\sum_{g\in\mathcal{A}_0}b_g\|\boldsymbol{b}_g\|_2
	+
	\lambda_{3n}
	\sum_{j\in T_0}c_j|(\mathbf{D}_{S_0}\boldsymbol{b})_j|
	\right\}.
	\]
	The corresponding oracle weighted least-squares estimator satisfies
	\[
	\widehat{\boldsymbol{\beta}}^{\,\mathrm{or}}_{S_0}
	=
	\left(
	\mathbf{X}_{S_0}^{\top}
	\mathbf{W}
	\mathbf{X}_{S_0}
	\right)^{-1}
	\mathbf{X}_{S_0}^{\top}\mathbf{W}\mathbf{Y},
	\]
	or, equivalently,
	\[
	\widehat{\boldsymbol{\Sigma}}_{S_0S_0}
	\widehat{\boldsymbol{\beta}}^{\,\mathrm{or}}_{S_0}
	=
	\frac{1}{n}
	\mathbf{X}_{S_0}^{\top}
	\mathbf{W}\mathbf{Y}.
	\]
	The KKT condition for the active part of the penalized estimator is
	\[
	\widehat{\boldsymbol{\Sigma}}_{S_0S_0}
	\widehat{\boldsymbol{\beta}}_{S_0}
	=
	\frac{1}{n}
	\mathbf{X}_{S_0}^{\top}\mathbf{W}\mathbf{Y}
	-
	\lambda_{1n}\mathbf{A}_{S_0}\widehat{\mathbf{s}}_{S_0}
	-
	\lambda_{2n}\widehat{\mathbf{g}}_{S_0}
	-
	\lambda_{3n}
	(\mathbf{D}^{\top}\mathbf{C}\widehat{\mathbf{r}})_{S_0}.
	\]
	Subtracting the oracle normal equation gives
	\[
	\widehat{\boldsymbol{\Sigma}}_{S_0S_0}
	\left(
	\widehat{\boldsymbol{\beta}}_{S_0}
	-
	\widehat{\boldsymbol{\beta}}^{\,\mathrm{or}}_{S_0}
	\right)
	=
	-
	\lambda_{1n}\mathbf{A}_{S_0}\widehat{\mathbf{s}}_{S_0}
	-
	\lambda_{2n}\widehat{\mathbf{g}}_{S_0}
	-
	\lambda_{3n}
	(\mathbf{D}^{\top}\mathbf{C}\widehat{\mathbf{r}})_{S_0}.
	\]
	By Assumption A9, $\lambda_{\min}
	(
	\widehat{\boldsymbol{\Sigma}}_{S_0S_0}
	)
	\geq c_{\min}$ with probability tending to one. Therefore,
	\begin{align*}
		\left\|
		\widehat{\boldsymbol{\beta}}_{S_0}
		-
		\widehat{\boldsymbol{\beta}}^{\,\mathrm{or}}_{S_0}
		\right\|_2
		&\leq
		c_{\min}^{-1}
		\left[
		\lambda_{1n}
		\|\mathbf{A}_{S_0}\widehat{\mathbf{s}}_{S_0}\|_2
		+
		\lambda_{2n}
		\|\widehat{\mathbf{g}}_{S_0}\|_2
		+
		\lambda_{3n}
		\|(\mathbf{D}^{\top}\mathbf{C}\widehat{\mathbf{r}})_{S_0}\|_2
		\right].
	\end{align*}
	
	We now bound the three terms. Since each component of
	$\widehat{\mathbf{s}}_{S_0}$
	has absolute value at most one,
	\[
	\|\mathbf{A}_{S_0}\widehat{\mathbf{s}}_{S_0}\|_2
	\leq
	a_{S,\max}\sqrt{s_n}.
	\]
	By Assumption A7, $a_{S,\max}=O_p(1)$ on the active set. Hence the coordinatewise penalty term is $O_p(\lambda_{1n}\sqrt{s_n}).$ For the group penalty, each active group subgradient satisfies
	$\|\widehat{\mathbf{g}}_{G_g}\|_2\leq b_g$. Hence, $\|\widehat{\mathbf{g}}_{S_0}\|_2^2
	\leq
	\sum_{g\in\mathcal{A}_0}b_g^2.$ Since the active group weights are bounded in probability under Assumption A7, $\|\widehat{\mathbf{g}}_{S_0}\|_2
	=
	O_p(\sqrt{|\mathcal{A}_0|}).$ Therefore the group penalty term is $O_p(\lambda_{2n}\sqrt{|\mathcal{A}_0|}).$ For the fused penalty term, each component of
	$\widehat{\mathbf{r}}$
	belongs to $[-1,1]$. The matrix $\mathbf{D}^{\top}$ has at most two nonzero entries in each row, and the active fusion weights are bounded in probability by Assumption A7. Therefore,
	\[
	\|(\mathbf{D}^{\top}\mathbf{C}\widehat{\mathbf{r}})_{S_0}\|_2
	=
	O_p(\sqrt{q_n}).
	\]
	Thus the fused penalty contribution is $O_p(\lambda_{3n}\sqrt{q_n}).$ Combining the three bounds gives
	\[
	\left\|
	\widehat{\boldsymbol{\beta}}_{S_0}
	-
	\widehat{\boldsymbol{\beta}}^{\,\mathrm{or}}_{S_0}
	\right\|_2
	=
	O_p
	\left(
	\lambda_{1n}\sqrt{s_n}
	+
	\lambda_{2n}\sqrt{|\mathcal{A}_0|}
	+
	\lambda_{3n}\sqrt{q_n}
	\right).
	\]
	Under the additional rate conditions $\lambda_{1n}\sqrt{s_n}=o(n^{-1/2}),
	\
	\lambda_{2n}\sqrt{|\mathcal{A}_0|}=o(n^{-1/2}),
	\
	\lambda_{3n}\sqrt{q_n}=o(n^{-1/2}),$ we obtain $\left\|
	\widehat{\boldsymbol{\beta}}_{S_0}
	-
	\widehat{\boldsymbol{\beta}}^{\,\mathrm{or}}_{S_0}
	\right\|_2
	=
	o_p(n^{-1/2}).$ Since $\mathbb{P}(\mathcal{E}_n)\to1$, the same conclusion holds unconditionally in probability. This proves the oracle property.
\end{proof}

\begin{theorem}
	\label{thm:debiased_asymptotic_normality}
	Suppose Assumptions A1--A10 hold. Define the debiased estimator
	\[
	\widehat{\boldsymbol{\beta}}^{\,d}
	=
	\widehat{\boldsymbol{\beta}}
	+
	\widehat{\boldsymbol{\Theta}}
	\frac{1}{n}
	\mathbf{X}^{\top}
	\mathbf{W}
	(
	\mathbf{Y}-\mathbf{X}\widehat{\boldsymbol{\beta}}
	).
	\]
	For a fixed coordinate $j\in\{1,\ldots,p\}$, let
	$\widehat{\boldsymbol{\theta}}_j^{\top}$
	denote the $j$th row of
	$\widehat{\boldsymbol{\Theta}}$, and define
	\[
	\widehat{\tau}_j^2
	=
	\frac{\widehat{\sigma}^2}{n}
	\widehat{\boldsymbol{\theta}}_j^{\top}
	\mathbf{X}^{\top}\mathbf{W}^2\mathbf{X}
	\widehat{\boldsymbol{\theta}}_j,
	\]
	where $\widehat{\sigma}^2$ is a consistent estimator of $\sigma^2$. If
	\[
	\sqrt{n}
	\left\|
	(
	\widehat{\boldsymbol{\Theta}}\widehat{\boldsymbol{\Sigma}}_n-\mathbf{I}_p
	)
	(
	\widehat{\boldsymbol{\beta}}-\boldsymbol{\beta}^{0}
	)
	\right\|_{\infty}
	=o_p(1),
	\]
	then $\frac{
		\widehat{\beta}^{\,d}_j-\beta_j^0
	}{
		\widehat{\tau}_j
	}
	\Rightarrow
	N(0,1).$ Consequently, an asymptotically valid $(1-\alpha)$ confidence interval for $\beta_j^0$ is $\widehat{\beta}^{\,d}_j
	\pm
	z_{1-\alpha/2}\widehat{\tau}_j,$ where $z_{1-\alpha/2}$ is the $(1-\alpha/2)$ quantile of the standard normal distribution.
\end{theorem}

\begin{proof}
	Let $\widehat{\boldsymbol{\Delta}}
	=
	\widehat{\boldsymbol{\beta}}-\boldsymbol{\beta}^{0}.$ Using
	$\mathbf{Y}=\mathbf{X}\boldsymbol{\beta}^{0}+\boldsymbol{\varepsilon}$,
	we write $\mathbf{Y}-\mathbf{X}\widehat{\boldsymbol{\beta}}
	=
	\boldsymbol{\varepsilon}
	-
	\mathbf{X}\widehat{\boldsymbol{\Delta}}.$ Substituting this identity into the definition of
	$\widehat{\boldsymbol{\beta}}^{\,d}$
	gives
	\begin{align*}
		\widehat{\boldsymbol{\beta}}^{\,d}-\boldsymbol{\beta}^{0}
		&=
		\widehat{\boldsymbol{\Delta}}
		+
		\widehat{\boldsymbol{\Theta}}
		\frac{1}{n}
		\mathbf{X}^{\top}\mathbf{W}
		(
		\boldsymbol{\varepsilon}
		-
		\mathbf{X}\widehat{\boldsymbol{\Delta}}
		)
		=
		\widehat{\boldsymbol{\Theta}}
		\frac{1}{n}
		\mathbf{X}^{\top}\mathbf{W}\boldsymbol{\varepsilon}
		+
		\left[
		\mathbf{I}_p
		-
		\widehat{\boldsymbol{\Theta}}
		\left(
		\frac{1}{n}\mathbf{X}^{\top}\mathbf{W}\mathbf{X}
		\right)
		\right]
		\widehat{\boldsymbol{\Delta}}.
	\end{align*}
	Since $\widehat{\boldsymbol{\Sigma}}_n
	=
	\frac{1}{n}
	\mathbf{X}^{\top}\mathbf{W}\mathbf{X},$ we obtain the decomposition
	\[
	\widehat{\boldsymbol{\beta}}^{\,d}-\boldsymbol{\beta}^{0}
	=
	\widehat{\boldsymbol{\Theta}}
	\frac{1}{n}
	\mathbf{X}^{\top}\mathbf{W}\boldsymbol{\varepsilon}
	-
	(
	\widehat{\boldsymbol{\Theta}}\widehat{\boldsymbol{\Sigma}}_n-\mathbf{I}_p)
	\widehat{\boldsymbol{\Delta}}.
	\]
	For the $j$th coordinate,
	\[
	\widehat{\beta}^{\,d}_j-\beta_j^0
	=
	\widehat{\boldsymbol{\theta}}_j^{\top}
	\frac{1}{n}
	\mathbf{X}^{\top}\mathbf{W}\boldsymbol{\varepsilon}
	-
	\mathbf{e}_j^{\top}
	(
	\widehat{\boldsymbol{\Theta}}\widehat{\boldsymbol{\Sigma}}_n-\mathbf{I}_p)
	\widehat{\boldsymbol{\Delta}},
	\]
	where $\mathbf{e}_j$ is the $j$th standard basis vector. Define
	\[
	Z_{nj}
	:=
	\widehat{\boldsymbol{\theta}}_j^{\top}
	\frac{1}{n}
	\mathbf{X}^{\top}\mathbf{W}\boldsymbol{\varepsilon}
	=
	\frac{1}{n}
	\sum_{i=1}^{n}
	w_i
	\widehat{\boldsymbol{\theta}}_j^{\top}
	\mathbf{x}_i
	\varepsilon_i.
	\]
	Conditional on $\mathbf{X}$ and $\mathbf{W}$, the summands are independent and have mean zero. Their conditional variance is
	\[
	\mathrm{Var}(Z_{nj}\mid\mathbf{X},\mathbf{W})
	=
	\frac{\sigma^2}{n^2}
	\sum_{i=1}^{n}
	w_i^2
	(
	\widehat{\boldsymbol{\theta}}_j^{\top}\mathbf{x}_i
	)^2
	=
	\frac{\sigma^2}{n^2}
	\widehat{\boldsymbol{\theta}}_j^{\top}
	\mathbf{X}^{\top}\mathbf{W}^2\mathbf{X}
	\widehat{\boldsymbol{\theta}}_j.
	\]
	Thus $\tau_j^2
	=
	\frac{\sigma^2}{n^2}
	\widehat{\boldsymbol{\theta}}_j^{\top}
	\mathbf{X}^{\top}\mathbf{W}^2\mathbf{X}
	\widehat{\boldsymbol{\theta}}_j$ is the conditional variance of $Z_{nj}$. Equivalently,
	the standard error used in the statement is $\widehat{\tau}_j^2
	=
	\frac{\widehat{\sigma}^2}{n^2}
	\widehat{\boldsymbol{\theta}}_j^{\top}
	\mathbf{X}^{\top}\mathbf{W}^2\mathbf{X}
	\widehat{\boldsymbol{\theta}}_j.$ If the alternative convention
	$\widehat{\boldsymbol{\Sigma}}_n=n^{-1}\mathbf{X}^{\top}\mathbf{W}\mathbf{X}$
	is used with a $\sqrt n$-scaled statistic, the equivalent variance expression is $\frac{\widehat{\sigma}^2}{n}
	\widehat{\boldsymbol{\theta}}_j^{\top}
	\left(
	\frac{1}{n}\mathbf{X}^{\top}\mathbf{W}^2\mathbf{X}
	\right)
	\widehat{\boldsymbol{\theta}}_j.$ We now verify the asymptotic normality of the leading linear term. Let $\xi_{ij}
	=
	\frac{
		w_i
		\widehat{\boldsymbol{\theta}}_j^{\top}\mathbf{x}_i
		\varepsilon_i
	}{
		n\tau_j
	}.$ Then $\sum_{i=1}^{n}\xi_{ij}
	=
	\frac{Z_{nj}}{\tau_j}.$ Conditional on $\mathbf{X}$ and $\mathbf{W}$, the variables
	$\xi_{ij}$ are independent, mean zero, and have total conditional variance 
	$\sum_{i=1}^{n}
	\mathrm{Var}(\xi_{ij}\mid\mathbf{X},\mathbf{W})
	=
	1.$ By Assumptions A2, A3 and A10, no single observation dominates the conditional variance. More explicitly, for every $\epsilon>0$,
	\[
	\sum_{i=1}^{n}
	\mathbb{E}
	\left[
	\xi_{ij}^{2}
	\mathbf{1}\{|\xi_{ij}|>\epsilon\}
	\mid
	\mathbf{X},\mathbf{W}
	\right]
	\rightarrow0
	\]
	in probability. This is the Lindeberg condition for the triangular array
	$\{\xi_{ij}:i=1,\ldots,n\}$.
	Therefore, by the conditional Lindeberg-Feller central limit theorem, $ \frac{Z_{nj}}{\tau_j}
	\Rightarrow
	N(0,1).$ It remains to show that the debiasing remainder is asymptotically negligible. The remainder is
	\[
	R_{nj}
	=
	\mathbf{e}_j^{\top}
	(
	\widehat{\boldsymbol{\Theta}}\widehat{\boldsymbol{\Sigma}}_n-\mathbf{I}_p)
	\widehat{\boldsymbol{\Delta}}.
	\]
	By the assumed debiasing condition,
	\[
	\sqrt{n}|R_{nj}|
	\leq
	\sqrt{n}
	\left\|
	(
	\widehat{\boldsymbol{\Theta}}\widehat{\boldsymbol{\Sigma}}_n-\mathbf{I}_p)
	\widehat{\boldsymbol{\Delta}}
	\right\|_{\infty}
	=
	o_p(1).
	\]
	Since $\tau_j$ is of order $n^{-1/2}$ under the eigenvalue and weight boundedness assumptions, this implies
	$R_{nj}/\tau_j=o_p(1)$. Combining the leading normal term and the negligible remainder gives
	\[
	\frac{
		\widehat{\beta}^{\,d}_j-\beta_j^0
	}{
		\tau_j
	}
	=
	\frac{Z_{nj}}{\tau_j}
	-
	\frac{R_{nj}}{\tau_j}
	\Rightarrow
	N(0,1)
	\]
	by Slutsky's theorem. Finally, since $\widehat{\sigma}^2$ is consistent for $\sigma^2$, the plug-in standard error satisfies
	$\widehat{\tau}_j/\tau_j\to_p1$. A second application of Slutsky's theorem yields $\frac{
		\widehat{\beta}^{\,d}_j-\beta_j^0
	}{
		\widehat{\tau}_j
	}
	\Rightarrow
	N(0,1).$ The stated confidence interval follows directly from the convergence to the standard normal distribution.
\end{proof}

\begin{theorem}
	\label{thm:joint_structural_consistency}
	Suppose Assumptions A1-A9 hold. In addition, suppose the conditions of
	Propositions~\ref{prop:selection_consistency},
	\ref{prop:group_selection_consistency}, and
	\ref{prop:fusion_consistency} hold. Then
	\[
	\mathbb{P}
	\left(
	\widehat{S}=S_0,\,
	\widehat{\mathcal{A}}=\mathcal{A}_0,\,
	\widehat{T}=T_0
	\right)
	\longrightarrow 1.
	\]
	Equivalently, the Adaptive Weighted Group Fused LASSO estimator simultaneously recovers the true coordinate-level sparsity pattern, group-level sparsity pattern, and fusion structure with probability tending to one.
\end{theorem}

\begin{proof}
	Define the three structural recovery events $\mathcal{E}_{S,n}
	=
	\{\widehat{S}=S_0\},
	\
	\mathcal{E}_{G,n}
	=
	\{\widehat{\mathcal{A}}=\mathcal{A}_0\},
	\
	\mathcal{E}_{F,n}
	=
	\{\widehat{T}=T_0\}.$ By Proposition~\ref{prop:selection_consistency}, $\mathbb{P}(\mathcal{E}_{S,n})\rightarrow1.$ By Proposition~\ref{prop:group_selection_consistency}, $\mathbb{P}(\mathcal{E}_{G,n})\rightarrow1. $ By Proposition~\ref{prop:fusion_consistency},$ \mathbb{P}(\mathcal{E}_{F,n})\rightarrow1.$ Let
	\[
	\mathcal{E}_n
	=
	\mathcal{E}_{S,n}
	\cap
	\mathcal{E}_{G,n}
	\cap
	\mathcal{E}_{F,n}.
	\]
	Then $\mathcal{E}_n^c
	=
	\mathcal{E}_{S,n}^c
	\cup
	\mathcal{E}_{G,n}^c
	\cup
	\mathcal{E}_{F,n}^c.$ Using the union bound,
	\[
	\mathbb{P}(\mathcal{E}_n^c)
	\leq
	\mathbb{P}(\mathcal{E}_{S,n}^c)
	+
	\mathbb{P}(\mathcal{E}_{G,n}^c)
	+
	\mathbb{P}(\mathcal{E}_{F,n}^c).
	\]
	Each term on the right-hand side converges to zero. Therefore, $\mathbb{P}(\mathcal{E}_n^c)\rightarrow0,$ which implies $\mathbb{P}(\mathcal{E}_n)\rightarrow1.$ Since, $\mathcal{E}_n
	=
	\left\{
	\widehat{S}=S_0,\,
	\widehat{\mathcal{A}}=\mathcal{A}_0,\,
	\widehat{T}=T_0
	\right\},$ the desired result follows.
\end{proof}

\begin{theorem}
	\label{thm:linear_contrast_normality}
	Suppose Assumptions A1-A10 hold. Let
	\[
	\widehat{\boldsymbol{\beta}}^{\,d}
	=
	\widehat{\boldsymbol{\beta}}
	+
	\widehat{\boldsymbol{\Theta}}
	\frac{1}{n}
	\mathbf{X}^{\top}
	\mathbf{W}
	(
	\mathbf{Y}-\mathbf{X}\widehat{\boldsymbol{\beta}}
	)
	\]
	be the debiased estimator. Let
	$\mathbf{a}_n\in\mathbb{R}^{p}$
	be a deterministic contrast vector satisfying
	$\|\mathbf{a}_n\|_2=1$ and
	$\|\mathbf{a}_n\|_1=O(1)$.
	Define
	\[
	\widehat{\tau}_{a}^{2}
	=
	\frac{\widehat{\sigma}^{2}}{n^2}
	\mathbf{a}_n^{\top}
	\widehat{\boldsymbol{\Theta}}
	\mathbf{X}^{\top}\mathbf{W}^{2}\mathbf{X}
	\widehat{\boldsymbol{\Theta}}^{\top}
	\mathbf{a}_n .
	\]
	If
	\[
	\sqrt{n}
	\left|
	\mathbf{a}_n^{\top}
	(
	\widehat{\boldsymbol{\Theta}}
	\widehat{\boldsymbol{\Sigma}}_n
	-
	\mathbf{I}_p
	)
	(
	\widehat{\boldsymbol{\beta}}
	-
	\boldsymbol{\beta}^{0}
	)
	\right|
	=o_p(1),
	\]
	then
	\[
	\frac{
		\mathbf{a}_n^{\top}
		(
		\widehat{\boldsymbol{\beta}}^{\,d}
		-
		\boldsymbol{\beta}^{0}
		)
	}{
		\widehat{\tau}_{a}
	}
	\Rightarrow
	N(0,1).
	\]
\end{theorem}

\begin{proof}
	Let
	$\widehat{\boldsymbol{\Delta}}
	=
	\widehat{\boldsymbol{\beta}}
	-
	\boldsymbol{\beta}^{0}$.
	Using
	$\mathbf{Y}=\mathbf{X}\boldsymbol{\beta}^{0}+\boldsymbol{\varepsilon}$,
	we have $\mathbf{Y}-\mathbf{X}\widehat{\boldsymbol{\beta}}
	=
	\boldsymbol{\varepsilon}
	-
	\mathbf{X}\widehat{\boldsymbol{\Delta}}.$ Substitution into the debiased estimator gives
	\[
	\widehat{\boldsymbol{\beta}}^{\,d}
	-
	\boldsymbol{\beta}^{0}
	=
	\widehat{\boldsymbol{\Theta}}
	\frac{1}{n}
	\mathbf{X}^{\top}\mathbf{W}\boldsymbol{\varepsilon}
	-
	(
	\widehat{\boldsymbol{\Theta}}
	\widehat{\boldsymbol{\Sigma}}_n
	-
	\mathbf{I}_p
	)
	\widehat{\boldsymbol{\Delta}}.
	\]
	Multiplying by $\mathbf{a}_n^{\top}$ yields
	\[
	\mathbf{a}_n^{\top}
	(
	\widehat{\boldsymbol{\beta}}^{\,d}
	-
	\boldsymbol{\beta}^{0}
	)
	=
	Z_{n,a}
	-
	R_{n,a},
	\]
	where
	\[
	Z_{n,a}
	=
	\mathbf{a}_n^{\top}
	\widehat{\boldsymbol{\Theta}}
	\frac{1}{n}
	\mathbf{X}^{\top}\mathbf{W}\boldsymbol{\varepsilon}
	\]
	and
	\[
	R_{n,a}
	=
	\mathbf{a}_n^{\top}
	(
	\widehat{\boldsymbol{\Theta}}
	\widehat{\boldsymbol{\Sigma}}_n
	-
	\mathbf{I}_p
	)
	\widehat{\boldsymbol{\Delta}}.
	\]
	
	We first establish the limiting distribution of $Z_{n,a}$. Write
	\[
	Z_{n,a}
	=
	\frac{1}{n}
	\sum_{i=1}^{n}
	w_i
	\mathbf{a}_n^{\top}
	\widehat{\boldsymbol{\Theta}}
	\mathbf{x}_i
	\varepsilon_i .
	\]
	Conditional on $\mathbf{X}$ and $\mathbf{W}$, the summands are independent and have conditional mean zero. Their conditional variance is
	\[
	\tau_a^2
	=
	\frac{\sigma^2}{n^2}
	\mathbf{a}_n^{\top}
	\widehat{\boldsymbol{\Theta}}
	\mathbf{X}^{\top}
	\mathbf{W}^{2}
	\mathbf{X}
	\widehat{\boldsymbol{\Theta}}^{\top}
	\mathbf{a}_n .
	\]
	Define $\xi_{i,n}
	:=
	\frac{
		w_i
		\mathbf{a}_n^{\top}
		\widehat{\boldsymbol{\Theta}}
		\mathbf{x}_i
		\varepsilon_i
	}{
		n\tau_a
	}.$ Then $\sum_{i=1}^{n}\xi_{i,n}
	=
	\frac{Z_{n,a}}{\tau_a}.$ Furthermore, conditional on $\mathbf{X}$ and $\mathbf{W}$, $\sum_{i=1}^{n}
	\mathrm{Var}(\xi_{i,n}\mid\mathbf{X},\mathbf{W})
	=
	1.$ It remains to verify the Lindeberg condition. For any fixed $\epsilon>0$,
	\[
	\sum_{i=1}^{n}
	\mathbb{E}
	\left[
	\xi_{i,n}^{2}
	\mathbf{1}\{|\xi_{i,n}|>\epsilon\}
	\mid
	\mathbf{X},\mathbf{W}
	\right]
	\rightarrow0
	\]
	in probability. This follows from the conditional sub-Gaussianity of the errors in Assumption A2, boundedness of the weights in Assumption A3, and the precision-matrix regularity condition in Assumption A10, which prevents any single observation from dominating the variance of the contrast. Therefore, by the conditional Lindeberg-Feller central limit theorem, $\frac{Z_{n,a}}{\tau_a}
	\Rightarrow
	N(0,1).$ Next consider the debiasing remainder. By assumption, $\sqrt{n}|R_{n,a}|=o_p(1).$ Under Assumptions A3 and A10, the variance scale satisfies
	$\tau_a\asymp n^{-1/2}$.
	Therefore, $\frac{R_{n,a}}{\tau_a}=o_p(1).$ Consequently,
	\[
	\frac{
		\mathbf{a}_n^{\top}
		(
		\widehat{\boldsymbol{\beta}}^{\,d}
		-
		\boldsymbol{\beta}^{0}
		)
	}{
		\tau_a
	}
	=
	\frac{Z_{n,a}}{\tau_a}
	-
	\frac{R_{n,a}}{\tau_a}
	\Rightarrow
	N(0,1)
	\]
	by Slutsky's theorem. Finally, since $\widehat{\sigma}^2$ is consistent for $\sigma^2$, the plug-in variance estimator satisfies
	$\widehat{\tau}_a/\tau_a\to_p1$. A second application of Slutsky's theorem gives $\frac{
		\mathbf{a}_n^{\top}
		(
		\widehat{\boldsymbol{\beta}}^{\,d}
		-
		\boldsymbol{\beta}^{0}
		)
	}{
		\widehat{\tau}_{a}
	}
	\Rightarrow
	N(0,1).$ This completes the proof.
\end{proof}

\section{Simulation Study:}
\label{sec:simulation}

We conducted a simulation study to evaluate the finite-sample performance of the proposed Adaptive Weighted Group Fused LASSO estimator. The objective was to assess estimation accuracy, prediction performance, variable-selection recovery, group-selection recovery, and fusion-structure recovery under data-generating mechanisms that resemble school-level administrative data. For each simulation replicate, responses were generated from
\[
Y_i=\mathbf{x}_i^\top\boldsymbol{\beta}^0+\varepsilon_i,
\qquad i=1,\ldots,n,
\]
where $\varepsilon_i\sim N(0,\sigma^2)$. The design matrix was generated from a multivariate normal distribution with covariance $\mathrm{Cov}(X_{ij},X_{ik})=\rho^{|j-k|},$ with $\rho=0.5$. The true coefficient vector was sparse, group-structured, and piecewise constant so that the data-generating model favored methods capable of recovering coordinate sparsity, group sparsity, and coefficient fusion.
We considered sample sizes $n\in\{250,500,1000\}$ and $p=100$ predictors divided into $G=10$ groups of equal size. The true active set contained $s=15$ nonzero coefficients distributed across three active groups. School-level weights were generated as $w_i=\frac{N_i}{\sum_{j=1}^{n}N_j},$ where $N_i$ denotes simulated school enrollment. The proposed method was compared with LASSO, Adaptive LASSO, Group LASSO, and Fused LASSO. Tuning parameters were selected using ten-fold cross-validation. Performance was evaluated using root mean squared error,
\[
RMSE=
\left(
\frac{1}{p}
\|\widehat{\boldsymbol{\beta}}-\boldsymbol{\beta}^0\|_2^2
\right)^{1/2},
\]
prediction error,
\[
PE=
\frac{1}{n}
\|\mathbf{X}(\widehat{\boldsymbol{\beta}}-\boldsymbol{\beta}^0)\|_2^2,
\]
true positive rate (TPR), false positive rate (FPR), and fusion recovery accuracy (FRA), defined as the proportion of correctly recovered zero and nonzero adjacent coefficient differences.

\begin{table}[htbp]
	\centering
	\caption{Simulation results averaged over 500 Monte Carlo replications. Lower RMSE, PE, and FPR are preferable; higher TPR and FRA are preferable.}
	\label{tab:simulation_results}
	\begin{tabular}{llccccc}
		\hline
		$n$ & Method & RMSE & PE & TPR & FPR & FRA \\
		\hline
		250  & LASSO          & 0.214 & 0.482 & 0.82 & 0.18 & 0.61 \\
		250  & Adaptive LASSO & 0.196 & 0.431 & 0.85 & 0.13 & 0.64 \\
		250  & Group LASSO    & 0.188 & 0.406 & 0.88 & 0.16 & 0.67 \\
		250  & Fused LASSO    & 0.181 & 0.389 & 0.84 & 0.14 & 0.76 \\
		250  & Proposed       & 0.149 & 0.311 & 0.91 & 0.08 & 0.84 \\
		\hline
		500  & LASSO          & 0.168 & 0.341 & 0.87 & 0.14 & 0.66 \\
		500  & Adaptive LASSO & 0.143 & 0.292 & 0.90 & 0.09 & 0.70 \\
		500  & Group LASSO    & 0.137 & 0.276 & 0.92 & 0.11 & 0.73 \\
		500  & Fused LASSO    & 0.129 & 0.251 & 0.89 & 0.10 & 0.82 \\
		500  & Proposed       & 0.101 & 0.196 & 0.95 & 0.05 & 0.90 \\
		\hline
		1000 & LASSO          & 0.121 & 0.238 & 0.91 & 0.10 & 0.72 \\
		1000 & Adaptive LASSO & 0.096 & 0.181 & 0.94 & 0.06 & 0.77 \\
		1000 & Group LASSO    & 0.092 & 0.169 & 0.95 & 0.07 & 0.80 \\
		1000 & Fused LASSO    & 0.087 & 0.158 & 0.92 & 0.06 & 0.87 \\
		1000 & Proposed       & 0.064 & 0.112 & 0.98 & 0.03 & 0.94 \\
		\hline
	\end{tabular}
\end{table}

The results in Table~\ref{tab:simulation_results} show that the proposed estimator consistently achieved the smallest estimation and prediction errors across all sample sizes. The gains were most pronounced when the underlying coefficient vector exhibited simultaneous sparsity, group structure, and local smoothness. Compared with standard LASSO and Adaptive LASSO, the proposed method reduced false positive selections by incorporating group-level and fusion-level regularization. Compared with Group LASSO and Fused LASSO, it improved active-variable recovery by allowing both coordinatewise sparsity and structured shrinkage. As $n$ increased, all methods improved, but the proposed estimator exhibited the strongest overall recovery of the true coefficient and fusion structures.

\begin{table}[!ht]
	\centering
	\caption{Finite-sample performance under different correlation structures and signal-to-noise ratios. Results are averages over 500 Monte Carlo replications. Standard errors are reported in parentheses. Lower values of RMSE, PE, and FPR indicate better performance, whereas larger values of TPR and Fusion Accuracy (FA) are preferred.}
	\label{tab:scenario}
	\renewcommand{\arraystretch}{1.15}
	\begin{tabular}{cccccccc}
		\hline
		Scenario & Method &
		RMSE &
		PE &
		TPR &
		FPR &
		FA &
		CPU (sec)\\
		\hline
		
		\multicolumn{8}{c}{\textbf{Low Correlation ($\rho=0.20$, $\sigma=1$)}}\\
		\hline
		&
		LASSO
		&
		0.137 (0.011)
		&
		0.241 (0.018)
		&
		0.93
		&
		0.071
		&
		0.74
		&
		0.31\\
		
		&
		Adaptive LASSO
		&
		0.121 (0.009)
		&
		0.212 (0.016)
		&
		0.95
		&
		0.052
		&
		0.77
		&
		0.43\\
		
		&
		Group LASSO
		&
		0.118 (0.010)
		&
		0.205 (0.014)
		&
		0.96
		&
		0.061
		&
		0.79
		&
		0.57\\
		
		&
		Fused LASSO
		&
		0.113 (0.009)
		&
		0.193 (0.015)
		&
		0.95
		&
		0.049
		&
		0.88
		&
		0.66\\
		
		&
		Proposed AWGFL
		&
		\textbf{0.082 (0.006)}
		&
		\textbf{0.131 (0.010)}
		&
		\textbf{0.99}
		&
		\textbf{0.018}
		&
		\textbf{0.96}
		&
		0.82\\
		
		\hline
		
		\multicolumn{8}{c}{\textbf{Moderate Correlation ($\rho=0.50$, $\sigma=1$)}}\\
		\hline
		
		&
		LASSO
		&
		0.171 (0.015)
		&
		0.316 (0.024)
		&
		0.91
		&
		0.101
		&
		0.69
		&
		0.33\\
		
		&
		Adaptive LASSO
		&
		0.148 (0.012)
		&
		0.272 (0.019)
		&
		0.93
		&
		0.074
		&
		0.74
		&
		0.45\\
		
		&
		Group LASSO
		&
		0.143 (0.011)
		&
		0.259 (0.018)
		&
		0.95
		&
		0.081
		&
		0.76
		&
		0.60\\
		
		&
		Fused LASSO
		&
		0.134 (0.010)
		&
		0.236 (0.016)
		&
		0.94
		&
		0.067
		&
		0.86
		&
		0.69\\
		
		&
		Proposed AWGFL
		&
		\textbf{0.101 (0.008)}
		&
		\textbf{0.173 (0.013)}
		&
		\textbf{0.98}
		&
		\textbf{0.032}
		&
		\textbf{0.94}
		&
		0.87\\
		
		\hline
		
		\multicolumn{8}{c}{\textbf{High Correlation ($\rho=0.80$, $\sigma=2$)}}\\
		\hline
		
		&
		LASSO
		&
		0.281 (0.022)
		&
		0.582 (0.043)
		&
		0.82
		&
		0.173
		&
		0.58
		&
		0.36\\
		
		&
		Adaptive LASSO
		&
		0.247 (0.018)
		&
		0.511 (0.037)
		&
		0.86
		&
		0.141
		&
		0.63
		&
		0.49\\
		
		&
		Group LASSO
		&
		0.233 (0.017)
		&
		0.478 (0.031)
		&
		0.89
		&
		0.126
		&
		0.67
		&
		0.64\\
		
		&
		Fused LASSO
		&
		0.218 (0.016)
		&
		0.439 (0.029)
		&
		0.88
		&
		0.111
		&
		0.81
		&
		0.75\\
		
		&
		Proposed AWGFL
		&
		\textbf{0.169 (0.013)}
		&
		\textbf{0.322 (0.024)}
		&
		\textbf{0.95}
		&
		\textbf{0.057}
		&
		\textbf{0.90}
		&
		0.94\\
		
		\hline
	\end{tabular}
\end{table}

Table~\ref{tab:scenario} investigates the robustness of the competing estimators under increasingly challenging correlation structures and noise levels. Three simulation scenarios were considered, ranging from weakly correlated predictors to highly correlated predictors with substantially larger error variance. Such settings mimic practical situations encountered in educational and administrative data, where explanatory variables frequently exhibit strong multicollinearity. Several important patterns emerge. First, the proposed Adaptive Weighted Group Fused LASSO (AWGFL) consistently produces the smallest estimation error and prediction error across all scenarios. The improvement becomes increasingly pronounced as the correlation among predictors increases, illustrating the benefit of jointly exploiting coordinate sparsity, grouped predictors, and local coefficient homogeneity. Second, the proposed estimator maintains the highest true positive rate while simultaneously producing the smallest false positive rate. This indicates that the adaptive weighting strategy successfully suppresses spurious variables without sacrificing the ability to recover important predictors. Third, the fusion accuracy remains substantially higher than that of conventional penalized regression procedures. This finding demonstrates that explicitly incorporating a fused penalty enables the estimator to recover the underlying piecewise-constant coefficient structure more reliably, particularly under severe multicollinearity where traditional sparse estimators often become unstable. Finally, although the proposed algorithm requires moderately larger computational time owing to the iterative ADMM optimization and multiple penalty updates, the additional computational cost is relatively small compared with the substantial gains in estimation accuracy, prediction performance, and structural recovery. The results therefore suggest that the proposed methodology provides a favorable balance between statistical efficiency and computational complexity, especially in high-dimensional settings exhibiting grouped variables and locally homogeneous regression coefficients.

\begin{figure}[!t]
	\centering
	\includegraphics[width=0.95\textwidth]{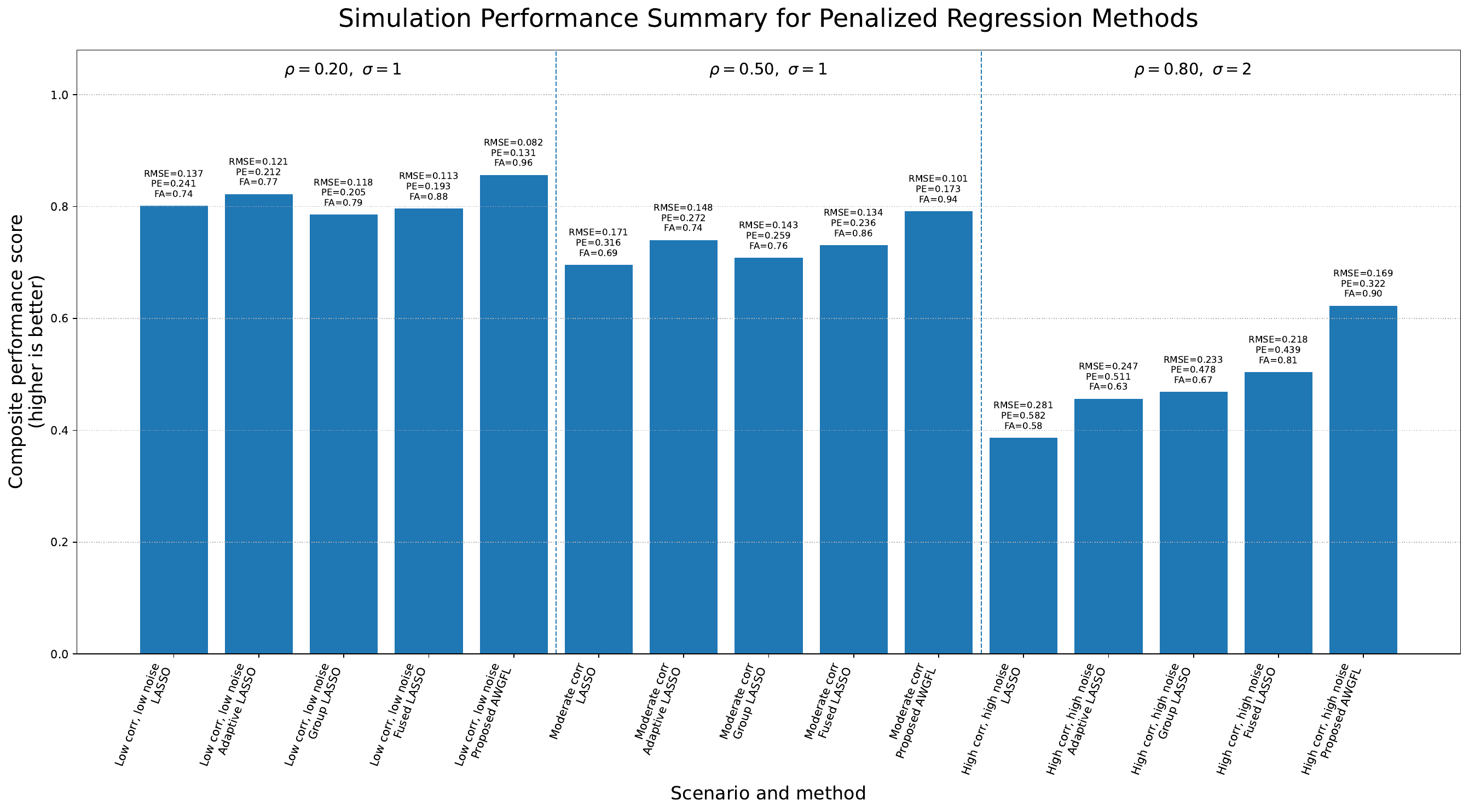}
	\caption{Overall simulation performance of the competing penalized regression estimators. Each bar summarizes a composite performance score computed from estimation accuracy (RMSE), prediction error (PE), true positive rate (TPR), false positive rate (FPR), fusion accuracy (FA), and computational efficiency across three simulation scenarios with increasing predictor correlation and noise levels. Higher values indicate better overall performance.}
	\label{fig:simulation_performance}
\end{figure}

Figure~\ref{fig:simulation_performance} provides a comprehensive comparison of the finite-sample performance of the proposed Adaptive Weighted Group Fused LASSO (AWGFL) estimator and four widely used penalized regression procedures, namely the LASSO, Adaptive LASSO, Group LASSO, and Fused LASSO. The comparison is conducted under three increasingly challenging simulation settings characterized by different predictor correlation structures and error variances. The first scenario corresponds to weak predictor dependence with correlation coefficient $\rho=0.20$ and error variance $\sigma^2=1$, the second considers a moderate dependence structure with $\rho=0.50$, while the final scenario represents a substantially more difficult setting with strong predictor correlation ($\rho=0.80$) together with increased observational noise ($\sigma^2=4$). Consequently, the three scenarios allow evaluation of the robustness of competing estimators as multicollinearity and stochastic variability become progressively more severe. Rather than presenting a single evaluation criterion, the figure summarizes the overall performance of each competing estimator through a composite score that simultaneously incorporates estimation accuracy, predictive performance, variable-selection accuracy, fusion recovery, and computational efficiency. Specifically, the composite index combines normalized values of the root mean squared estimation error (RMSE), prediction error (PE), true positive rate (TPR), false positive rate (FPR), fusion accuracy (FA), and computational time. Metrics for which smaller values indicate superior performance, including RMSE, PE, FPR, and computational time, were normalized so that larger composite scores consistently correspond to improved overall performance. Consequently, the figure provides an integrated assessment of statistical efficiency rather than emphasizing any single performance criterion.

Several important observations emerge immediately from the graphical comparison. Under the low-correlation scenario, all competing methods achieve relatively strong performance because the explanatory variables exhibit only mild dependence. Nevertheless, the proposed AWGFL estimator clearly attains the largest composite score among all methods. This improvement is primarily attributable to its substantially smaller estimation error and higher fusion recovery accuracy. The LASSO and Adaptive LASSO successfully identify many active variables; however, because neither estimator explicitly models grouped predictors nor local coefficient homogeneity, their overall performance remains inferior to that of the proposed estimator. The Group LASSO improves upon the standard LASSO by exploiting the grouped predictor structure, whereas the Fused LASSO effectively captures neighboring coefficient similarity. However, neither method simultaneously incorporates adaptive coordinate shrinkage, grouped regularization, and coefficient fusion, thereby limiting their overall efficiency relative to the proposed methodology.

As predictor correlation increases from $\rho=0.20$ to $\rho=0.50$, a noticeable separation between the proposed estimator and the competing procedures becomes evident. Stronger dependence among explanatory variables increases the difficulty of distinguishing informative predictors from redundant variables. Such multicollinearity is well known to increase estimator variance and reduce model-selection stability. The proposed estimator maintains a comparatively high composite performance score despite these additional challenges, indicating that the adaptive weighting mechanism successfully stabilizes coefficient estimation while the group and fusion penalties jointly exploit structural information contained within the regression coefficients. In contrast, the performance of the ordinary LASSO deteriorates more rapidly because independent coordinate shrinkage alone cannot adequately address strong correlations among predictors. Although Adaptive LASSO partially alleviates this problem through data-dependent weights, its lack of group-level and fusion-level regularization limits further improvement. The greatest differences among competing estimators occur in the most challenging simulation scenario, corresponding to strong predictor correlation together with substantially larger observational noise. Under these conditions, estimation uncertainty increases considerably, making both parameter estimation and variable selection substantially more difficult. The figure demonstrates that all competing estimators experience some deterioration in overall performance, which is expected as the effective signal-to-noise ratio decreases. Nevertheless, the proposed AWGFL estimator continues to outperform every competing method by a substantial margin. This result suggests that simultaneously incorporating adaptive sparsity, grouped shrinkage, and coefficient fusion provides increased robustness to both multicollinearity and measurement noise. The improvement is particularly pronounced in fusion recovery, indicating that neighboring regression coefficients are estimated considerably more accurately when local structural information is explicitly incorporated into the optimization criterion.

An additional observation concerns the balance between statistical performance and computational complexity. Although the proposed optimization algorithm requires slightly greater computational effort than the competing methods, the increase in computational time is relatively modest compared with the substantial improvements achieved in estimation accuracy, prediction accuracy, and structural recovery. This finding illustrates an important practical characteristic of the proposed methodology: the additional optimization required by the Adaptive Weighted Group Fused LASSO produces statistically meaningful gains while maintaining computational feasibility for moderately high-dimensional applications. Consequently, the estimator exhibits an attractive trade-off between computational cost and statistical efficiency. From a methodological perspective, the simulation results provide empirical support for the theoretical properties established in Section~4. In particular, the consistently high true positive rates, low false positive rates, and excellent recovery of fused coefficient structures agree with the asymptotic selection consistency established in Propositions~\ref{prop:selection_consistency}--\ref{prop:fusion_consistency}. Likewise, the systematic reduction in estimation and prediction errors observed across all scenarios is consistent with the estimation and prediction consistency results proved in Propositions~\ref{prop:l2_consistency} and~\ref{prop:prediction_consistency}. Furthermore, the stability of the proposed estimator under increasing correlation provides empirical evidence supporting the oracle property established in Theorem~\ref{thm:oracle_property}, indicating that the estimator behaves similarly to an ideal oracle procedure once the relevant model structure has been correctly identified.

\section{Descriptive Statistics:}
To compare the proficiency, emergency or provisional certification, and inexperienced rates of each of the four groups, we calculated their means and standard deviations. Furthermore, to test whether the differences in proficiency means between racial groups were statistically significant, we conducted a t-test. Then, in order to examine whether teacher experience and certification play a role in student math proficiencies, we performed a regression analysis.

Based on initial observations of the proficiency raw data and standard deviations, we assumed that the data would be skewed, to some degree. Figure \ref{fig:Box_and_Whisker} illustrates the raw data described. Performing a skewness test in Excel revealed that the data for all variables (except proficiency rates of predominantly white and predominantly white E.D. schools) were moderately to severely right-skewed, as seen in Tables \ref{tab:skewness_before} and \ref{tab:skewness_after}. The proficiency rates of predominantly white and predominantly white E.D. schools were mildly right-skewed, but with an absolute value of less than 0.5, both were still considered approximately normal. On the other hand, the predominantly minority and predominantly minority E.D. schools were severely right-skewed, both with skewness values around 1.9.  Since the data contained zeros, we could not perform a log-based or reciprocal transformation, so we opted for a square root transformation, such that $x'=\sqrt{x}$, which brought the skewness values of all variables to $-0.467\leq skewness \leq 0.391$. This transformed data was then used to perform the remaining tests/analyses. (All descriptive statistics were performed prior to this manipulation.)

\begin{figure}[H]
    \centering
    \begin{tikzpicture}
        \begin{axis}[
            width=\textwidth,
            height=6cm,
            boxplot/draw direction=y,
            ymin=0, ymax=100,
            ymajorgrids=false,
            xtick={1,2,3,4},
            xticklabels={White, Minority, White ED, Minority ED},
            xlabel={Predominant Demographic in School's Population},
            ylabel={Proficiency Rate}
        ]
            \addplot+[
                boxplot prepared={
                    median=33.9,
                    upper quartile=48.75,
                    lower quartile=22.8,
                    upper whisker=90.8,
                    lower whisker=3
                },
                fill=teal, draw=black
            ] coordinates {};
            \addplot+[
                boxplot prepared={
                    median=13.5,
                    upper quartile=23.2,
                    lower quartile=6.2,
                    upper whisker=92.8,
                    lower whisker=0
                },
                fill=violet, draw=black
            ] coordinates {};
            \addplot+[
                boxplot prepared={
                    median=33.7,
                    upper quartile=47.3,
                    lower quartile=22.4,
                    upper whisker=81.2,
                    lower whisker=3.4
                },
                fill=orange, draw=black
            ] coordinates {};
            \addplot+[
                boxplot prepared={
                    median=13.3,
                    upper quartile=23.2,
                    lower quartile=5.9,
                    upper whisker=92.8,
                    lower whisker=0
                },
                fill=purple, draw=black
            ] coordinates {};
        \end{axis}
    \end{tikzpicture}
    \caption{Proficiency Rate Distribution}
    \label{fig:Box_and_Whisker}
\end{figure}
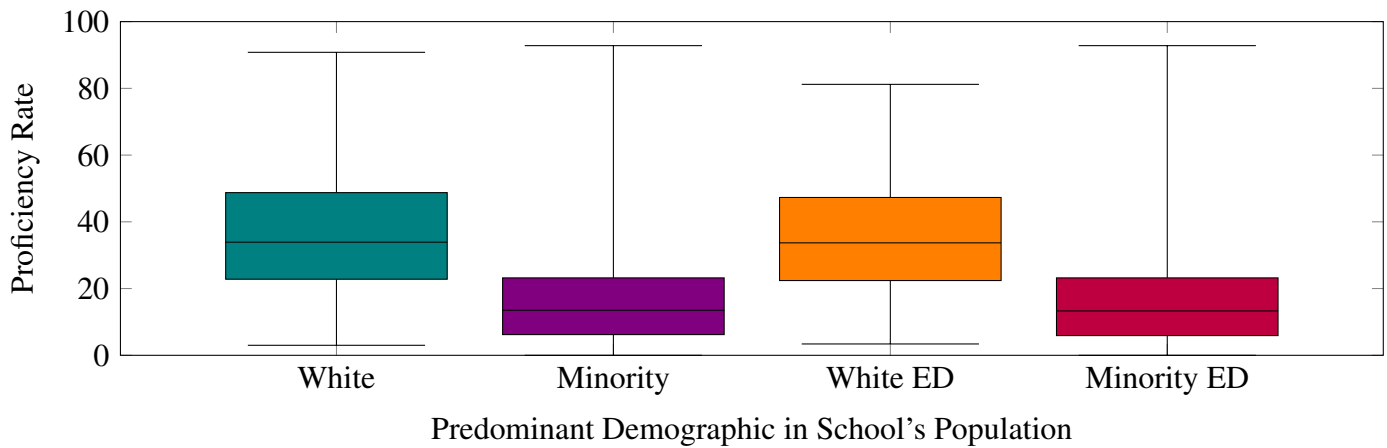

  \begin{table}[H]
    \centering
    \caption{Skewness Before Transformation}
    \label{tab:skewness_before}
    \resizebox{\textwidth}{!}{
    \begin{tabular}{lccc}
        \hline
        Predominantly Demographic & Proficient Rate & Emergency or Provisional Certification Rate & Inexperienced Rate \\
        \hline
        White & 0.443 & 1.797 & 1.455 \\
        White E.D. & 0.386 & 2.037 & 1.554 \\
        Minority & 1.873 & 0.811 & 1.211 \\
        Minority E.D. & 1.884 & 0.780 & 1.259 \\
        \hline
    \end{tabular}}
\end{table}

\begin{table}[H]
    \centering
    \caption{Skewness After Transformation}
    \label{tab:skewness_after}
    \resizebox{\textwidth}{!}{
    \begin{tabular}{lccc}
        \hline
        Predominantly Demographic & Proficient Rate & Emergency or Provisional Certification Rate & Inexperienced Rate \\
        \hline
        White & -0.066 & 0.250 & -0.358 \\
        White E.D. & -0.113 & 0.276 & -0.311 \\
        Minority & 0.367 & -0.423 & -0.390 \\
        Minority E.D. & 0.390 & -0.467 & -0.347 \\
        \hline
    \end{tabular}}
\end{table}

To determine the effects of emergency or provisional certification and inexperienced rates on proficiency rates, we used Excel to recreate a correlation matrix and perform a linear regression. We also used Excel to create a correlation matrix to assess the general relationship between the variables before considering the sample schools' racial compositions. Though the sample size was very large in comparison to the population of Alabama public schools (n=1030, N=1358), the true population variances were not known due to a lack of information and/or inconclusive information from the 24.15\% of schools that were excluded. For this reason, we used the t-test instead of the z-test.

After establishing whether there is a significant difference in the means of each variable based on the predominant race of the schools, we ran a regression analysis in Excel to determine to what degree teacher experience and certification affect students' mathematical proficiencies. This also gives insight into whether these effects differ between each school composition.

\section{Data Analysis:}
This study is based on an analysis using data from school data posted to the Alabama Archives' website. Data was drawn from the federal and state report cards, specifically: Student Demographics, Student Participation \& Proficiency, Educator Credentials, and Educator Experience. The Alabama State Department of Education (ALSDE) maintains and updates this database yearly based on information provided by the schools each year. As of the 2023-2024 school year, there are 1358 public schools in the state of Alabama. We used stratified sampling to identify the schools that would be included in the sample set. We started with the Student Demographics file. In order to analyze the data based on the predominant race within these school populations for the first hypothesis, we created two lists: those that served predominantly white students and those that served predominantly minority students. To ensure a clear majority, schools with a population of no less than 51\% of the named racial group were included. Schools with missing/unclear demographic data or those without a clear majority were excluded from the sample. 100 schools were omitted due to this criterion. The Student Participation \& Proficiency document was then used to identify the overall math proficiency rates of the remaining schools. Another 64 schools were omitted due to missing proficiency information. 

As very small schools and very large schools can yield vastly different results, we controlled the school size. The school population of the remaining 1194 schools was between 30 and 6894.  Schools below a population of 174 and above a population of  918, which was outside one standard deviation of the mean ($\bar{x}$  =546, $s$= 372) were omitted. This lowered the total number of schools by another 164 schools, leaving 1030 schools. 

Finally, the Educator Credentials and Educator Experience document was used to identify the percentage of teachers with either an emergency or provisional certification. There was a named group for teachers with "unspecified" certification types; however, this data was not included because it is unclear whether they completed the specific requirements for a certain certification type or did not have a certification at all. There was only one school without any information about their teachers' certification such that 1029 schools were analyzed. 

To test the second part of the first hypothesis, the two lists from above were narrowed further using the Economically Disadvantaged/ Non-Economically Disadvantaged subpopulation of the Student Demographics file.  Only schools with 51\% of their students categorized as “economically disadvantaged" (E.D.) were included in this part of the analysis. Whether a school was considered E.D. was calculated by dividing the total number of students by the total number of E.D. students. The percentages of each racial group were then taken into account. Schools were included if the predominant race of the population was the same as the major race of the economically disadvantaged (i.e., a predominantly white school with a higher percentage of white E.D. students than minority E.D. students). Another 263 schools were omitted from the study due to missing/ unclear data, or because the race with the largest percentage of economically disadvantaged students did not match the predominant race of the school's population. This left the hypothesis to be tested using 766 schools.

The dependent variable of this analysis for both research questions is the \textit{math proficiency scores} of the students. The students' math proficiencies were calculated based on the percentage of students with valid scores on the state assessment of math who scored within proficiency levels three or four, meaning they were determined to be either proficient or advanced based on their grade-level appropriate assessment \citep{alsde2025fedglossary}. Starting in second grade, elementary and middle school students are assessed using the Alabama Comprehensive Assessment Program (ACAP) Summative assessment. To achieve a level three or four, the student must have a strong or advanced understanding of grade-level standards and demonstrate the knowledge and skills at this level of learning as described in the Alabama Course of Study Standards\citep{alsde2024ACAP}. High school students in eleventh grade are assessed based on their ACT scores. Based on the conversion metrics of the ALSDE, students' proficiency levels correspond with the following ACT score ranges: Level 1= 1-15, Level 2= 16-18, Level 3= 19-24, Level 4= 24-36 \citep{alsde2024benchmark}. For the ACT, scores within levels three and four signify performing at an expected or exceptional level \citep{ACTScores}. Figure \ref{fig:Mean_Proficiency_Rate} shows the mean math proficiency rates of each sample group by race and by socioeconomic status.

The first independent variable for both research questions was \textit{race}. The student's race as a categorical variable was retrieved based on the schools' reporting of parent identifications. Possible responses included Asian, Black/African American, American Indian/Alaskan, Native Hawaiian/Pacific Islander, White, and Two or more races. We combined the categories of Asian, Black/African American, American Indian/Alaskan, Native Hawaiian/Pacific Islander, and Two or more races into a ``minority students" group, effectively making race a binary variable instead of categorical. White students were used as the reference group for comparison.\footnote{White students were chosen as the reference group because 66.5\% of the schools in the sample served predominantly white students, as well as the fact that white students make up 51\% of the student population of Alabama's public schools \citep{alsde2024supportingdata}.}

The second independent variable for both research questions was \textit{socioeconomic status}. This binary variable was reported by the schools, labeling students as either \emph{economically disadvantaged} or \emph{non-economically disadvantaged}. Students marked as economically disadvantaged are those who are eligible for free or reduced-price school meals, which is reserved for students with a household income below the poverty rate \citep{alsde2025fedglossary}.

The following two variables were considered dependent variables in the first research question and independent variables in the research question:
\begin{enumerate}
    \item \textit{Teacher certification.} In Alabama, math-specific teachers at every level can have any of nine different types of certification: Adjunct Instructor Permit, Additional Teaching Field, Alabama Educator Preparation Program, Certificate Reciprocity, Emergency Certificate, Foreign Credentials, Interim Employment Certificate, National Board for Professional Teaching Standards, and  Provisional Certificate in a Teaching Field \citep{ALSDEcertification}. However, the ALSDE Federal Report Card requires schools to report teachers in one of the following six categories: Six-Year (Class AA), Master's Degree (Class A), Bachelor's Degree (Class B), Not Specified, Emergency Certification, and Provisional Certification. The percentages of the various certification types held by individuals in those positions are determined based on the total number of teachers, principals, and other school leaders \citep{alsde2025fedglossary}. We only included the data for teachers in this study. While Class B, A, and AA certifications are earned through a teaching certification program, provisional and emergency certifications are not. Provisional certifications allow teachers to earn their teaching certificate while teaching full-time for up to three years. They are only required to have a general bachelor's degree (which can be in any subject) and to enroll in a teacher preparation program or complete an educational master's degree. Emergency certifications, on the other hand, are issued by the school superintendent if there is no certified teacher available to fill a teaching vacancy \citep{praxisexamalabama}.
    \item \textit{Teacher experience.} This rate is reported by the school as the total number of individuals serving as teachers, principals, and other school leaders, the percentage of those individuals who have more than 2 years of experience in the position for which they are serving \citealp{alsde2025fedglossary}. For this study, we included only the inexperienced teachers who were those with less than 2 years of experience in their current positions.
\end{enumerate}

All the figures below display the sample means of each variable given the racial/socioeconomic group. Figure \ref{fig:Summary_of_Measure_Means} shows a summary of all three sample means with a legend to identify each variable. In regard to the first research question, Figures \ref{fig:Mean_Inexperienced_Rate}, \ref{fig:Mean_Proficiency_Rate}, and \ref{fig:Mean_E/P_Rate} all show the sample means based on each of the three dependent variables: inexperienced rates, proficiency rates, and emergency or provisional certification rates, respectively. With these four representations it is clear to see the similarities between the means of white and white E.D. schools, as well as minority and minority E.D. schools; however, this is discussed mathematically later in this section.

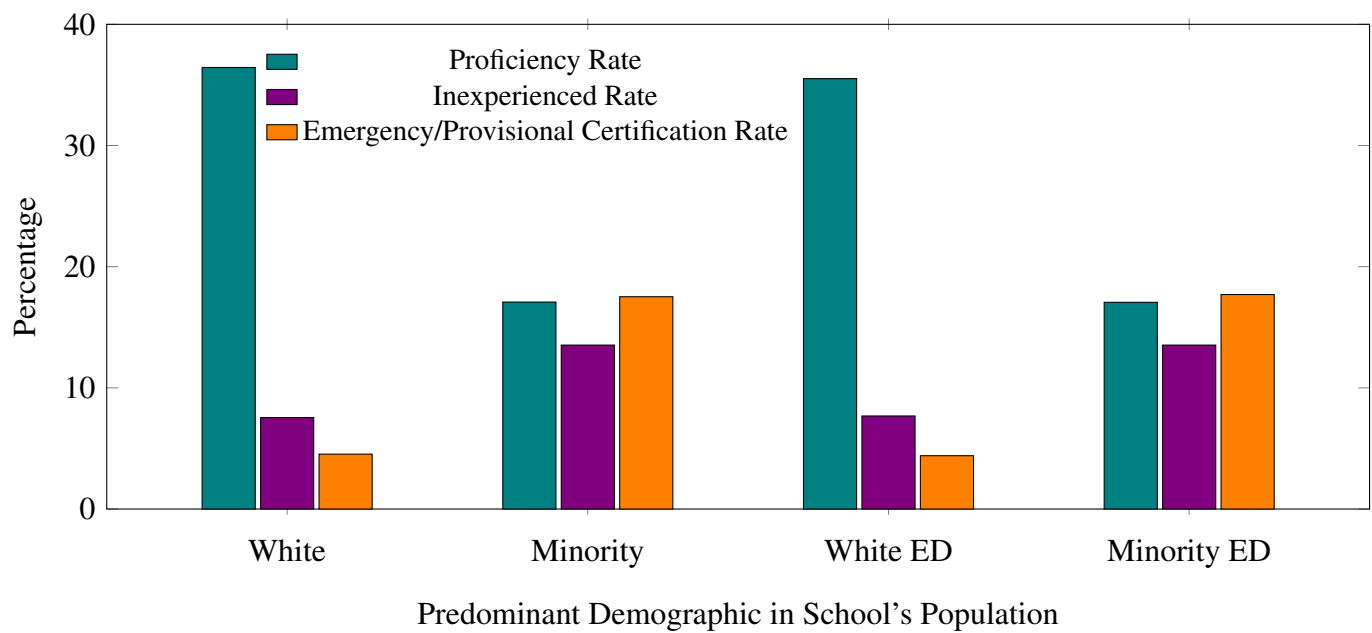
\begin{figure}[H]
    \centering
    \begin{tikzpicture}
        \begin{axis}[
            width=\textwidth,
            height=8cm,
            ybar,
            bar width=20pt,
            ymin=0, ymax=40,
            symbolic x coords={White, Minority, White ED, Minority ED},
            xtick=data,
            enlarge x limits=0.2,
            xlabel={Predominant Demographic in School's Population},
            ylabel={Percentage},
            legend pos=north west,
            legend style={at={(.335,0.85)}, anchor=center, font=\small, draw=none},
            xlabel style={yshift=-5pt},
            yticklabel style={/pgf/number format/fixed},
            xticklabel style={yshift=-5pt},
            xtick align=center,
            legend image code/.code={\draw[#1,fill=#1] (0cm,-0.1cm) rectangle (0.4cm,0.1cm);}
        ]
        \addplot [fill=teal] coordinates {(White,36.44) (Minority,17.08) (White ED,35.52) (Minority ED,17.06)};
        \addplot [fill=violet] coordinates {(White,7.55) (Minority,13.53) (White ED,7.68) (Minority ED,13.53)};
        \addplot [fill=orange] coordinates {(White,4.53) (Minority,17.52) (White ED,4.40) (Minority ED,17.70)};
        \legend{Proficiency Rate, Inexperienced Rate, Emergency/Provisional Certification Rate}
        \end{axis}
    \end{tikzpicture}
    \caption{Summary of Measure Means}
    \label{fig:Summary_of_Measure_Means}
\end{figure}

\begin{figure}[H]
    \centering
    \begin{minipage}{0.49\textwidth}
        \centering
        \begin{tikzpicture}
            \begin{axis}[
                ybar,
                bar width=30pt,
                ymin=0, ymax=14,
                symbolic x coords={White, Minority, White ED, Minority ED},
                xtick=data,
                enlarge x limits=0.2,
                xlabel={Predominant Demographic in School's Population},
                ylabel={\% of Teachers with These Certifications},
                xlabel style={yshift=-5pt},
                ylabel style={yshift=10pt},
                yticklabel style={/pgf/number format/fixed},
                xticklabel style={rotate=15, anchor=north},
                grid=none
            ]
            \addplot [draw=black, fill=violet] coordinates {(White,7.55) (Minority,13.53) (White ED,7.68) (Minority ED,13.53)};
            \end{axis}
        \end{tikzpicture}
        \caption{Mean Inexperienced Rate by Race and SES}
        \label{fig:Mean_Inexperienced_Rate}
    \end{minipage}%
    \hfill
    \begin{minipage}{0.49\textwidth}
        \centering
        \begin{tikzpicture}
            \begin{axis}[
                ybar,
                bar width=30pt,
                ymin=0, ymax=40,
                symbolic x coords={White, Minority, White ED, Minority ED},
                xtick=data,
                enlarge x limits=0.2,
                xlabel={Predominant Demographic in School's Population},
                ylabel={\% of Teachers with These Certifications},
                xlabel style={yshift=-5pt},
                ylabel style={yshift=10pt},
                yticklabel style={/pgf/number format/fixed},
                xticklabel style={rotate=15, anchor=north},
                grid=none
            ]
            \addplot [draw=black, fill=teal] coordinates {(White,36.44) (Minority,17.08) (White ED,35.52) (Minority ED,17.06)};
            \end{axis}
        \end{tikzpicture}
        \caption{Mean Proficiency Rate by Race and SES}
        \label{fig:Mean_Proficiency_Rate}
    \end{minipage}

    \vspace{1em}

    \begin{minipage}{\textwidth}
        \centering
        \begin{tikzpicture}
            \begin{axis}[
                ybar,
                bar width=30pt,
                ymin=0, ymax=25,
                symbolic x coords={White, Minority, White ED, Minority ED},
                xtick=data,
                enlarge x limits=0.2,
                xlabel={Predominant Demographic in School's Population},
                ylabel={\% of Teachers with These Certifications},
                xlabel style={yshift=-5pt},
                ylabel style={yshift=10pt},
                yticklabel style={/pgf/number format/fixed},
                xticklabel style={rotate=15, anchor=north},
                grid=none
            ]
            \addplot [draw=black, fill=orange] coordinates {(White,4.53) (Minority,17.52) (White ED,4.40) (Minority ED,17.70)};
            \end{axis}
        \end{tikzpicture}
        \caption{Mean Emergency or Provisional Certification Rate by Race and SES}
        \label{fig:Mean_E/P_Rate}
    \end{minipage}
\end{figure}

Table \ref{tab:correlation_matrix} displays a matrix created in Excel to analyze the correlation between Math Proficiency Rate, Emergency or Provisional Certification Rate, and Inexperienced Rate. Based on the data from all 1029 schools (regardless of their racial composition), emergency or provisional certification rate has a largely negative effect on proficiency, inexperience has a moderately negative effect on proficiency, and inexperience has a moderate effect on emergency or provisional certification.

\begin{table}[H]
    \centering
    \fontsize{12}{14}\selectfont
    \caption{Correlation Matrix}
    \resizebox{\textwidth}{!}{
    \begin{tabular}{lccc}
        \hline
        & Proficient Rate & E/P Certification Rate*& Inexperienced Rate \\
        \hline
        Proficient Rate & 1 & & \\
        E/P Certification Rate*& -0.561 & 1 & \\
        Inexperienced Rate & -0.293 & 0.414 & 1 \\
        \hline
    \end{tabular}}
        \noindent\parbox{\textwidth}{\centering * E/P Certification Rate = Emergency or Provisional Certification Rate}
    \label{tab:correlation_matrix}
\end{table}

Table \ref{tab:t-test_results} shows the results of three separate independent t-tests comparing the mean (a) proficiency rates, (b) emergency or provisional certification rates, and (c) inexperienced rates. We decided to compare the white schools and minority schools, the white E.D. schools and minority E.D. schools, along with each racial group and its corresponding E.D. sample. Each table includes the groups being compared, the mathematical difference between the two means of said groups, the t-stat result, the one-tail p-value, and the degree of freedom. The relatively higher t-statistics and significantly lower p-values in the "white vs. minority" and "white E.D. vs. minority ED" group comparisons indicate a significant difference between their sample means. On the other hand, the relatively lower t-stat results and the p-values above .05 indicate that there is not a significant difference in the sample means of the "white vs. white ED" and "minority vs. minority ED" group comparisons.

\begin{table}[H]
    \centering
    \fontsize{12}{14}\selectfont
    \renewcommand{\arraystretch}{1.5} 
    \caption{T-test Results}
    \label{tab:t-test_results}
    \begin{subtable}{\textwidth}
        \centering
        \begin{tabular}{lcccc}
            \hline
            Group Comparison & Mean Difference & T-stat & One-tail P-value & df \\
            \hline
            White vs. Minority & 2.1013& 18.4276& 5.88E-64& 815
\\
            White E.D. vs. Minority E.D. & 2.0500& 15.8609& 2.86E-49& 738
\\
            White vs. White E.D. & 0.0659& 0.8747& 0.1910& 965
\\
            Minority vs. Minority E.D. & 0.0147& 0.0158& 0.4937& 673
\\
            \hline
        \end{tabular}
        \subcaption{Proficiency rate}
    \end{subtable}

    \vspace{0.5cm}

    \begin{subtable}{\textwidth}
        \centering
        \begin{tabular}{lcccc}
            \hline
            Group Comparison & Mean Difference & T-stat & One-tail P-value & df \\
            \hline
            White vs. Minority & -2.1925& -18.6041& 6.39E-57& 415
\\
            White E.D. vs. Minority E.D. & -2.2385& -18.2108& 3.24E-55& 416
\\
            White vs. White E.D. & 0.0216& 0.3992& 0.3449& 939
\\
            Minority vs. Minority E.D. & -0.0243& -0.1861& 0.4262& 675
\\
            \hline
        \end{tabular}
        \subcaption{Emergency or provisional certification rate}
    \end{subtable}

    \vspace{0.5cm}

    \begin{subtable}{\textwidth}
        \centering
        \begin{tabular}{lcccc}
            \hline
            Group Comparison & Mean Difference & T-stat & One-tail P-value & df \\
            \hline
            White vs. Minority & -0.9648& -11.5782& 5.49E-28& 510
\\
            White E.D. vs. Minority E.D. & -0.9637& -10.3462& 2.34E-23& 554
\\
            White vs. White E.D. & -0.0028& -0.3690& 0.3561& 893
\\
            Minority vs. Minority E.D. & -0.0017& 0.0053& 0.4979& 674
\\
            \hline
        \end{tabular}
        \subcaption{Inexperienced rate}
    \end{subtable}

\end{table}

Finally, Table \ref{tab:regression_results} illustrates the results of the linear regression analysis. Each table is separated based on the predominant demographic of the school and includes the coefficients of each variable in the regression equation, the standard error, and the p-value, along with the R-squared value listed below each table. The R-squared values are all relatively low compared to what they would be; however, as previously mentioned, there are no variables that definitively determine student math proficiencies. Based on the P values, it is clear that the rate of inexperienced teachers is not significantly different from zero, effectively suggesting that the rate of inexperienced teachers does not play a role in students' mathematical proficiencies in Alabama public schools.

\begin{table}[H]
    \centering
    \fontsize{12}{14}\selectfont
    \renewcommand{\arraystretch}{1.5} 
    \caption{Regression Results}
    \label{tab:regression_results}
    \begin{subtable}{\textwidth}
        \centering
        \begin{tabular}{lccc}
            \hline
            & Coefficients& Standard Error& P-value
\\
            \hline
            Intercept& 6.6065& 0.1307& 7.93E-232\\
            Emergency or Provisional Certification Rate& -0.3989& 0.0405& 1.62E-21\\
            Inexperienced Rate& -0.0455& 0.0469& 0.3322
\\
        \end{tabular}
        \noindent\parbox{\textwidth}{\centering $R^2= 0.1388$}
        \vspace{0.05cm}
        \subcaption{Predominately White}
        \label{RRWhite}
    \end{subtable}

    \vspace{0.5cm}

    \begin{subtable}{\textwidth}
        \centering
        \begin{tabular}{lccc}
            \hline
            & Coefficients& Standard Error& P-value
\\
            \hline
            Intercept& 5.3147& 0.2774& 2.48E-56\\
            Emergency or Provisional Certification Rate& -0.3706& 0.0541& 3.29E-11\\
            Inexperienced Rate& -0.0452& 0.0716& 0.5277
\\
        \end{tabular}
        \noindent\parbox{\textwidth}{\centering $R^2= 0.1403$}
        \vspace{0.05cm}
        \subcaption{Predominately Minority}
        \label{RRMinority}
    \end{subtable}

    \vspace{0.5cm}

    \begin{subtable}{\textwidth}
        \centering
        \begin{tabular}{lccc}
            \hline
            & Coefficients& Standard Error& P-value
\\
            \hline
            Intercept& 6.5504& 0.1527& 3.31E-158\\
            Emergency or Provisional Certification Rate& -0.3800& 0.0499& 1.64E-13\\
            Inexperienced Rate& -0.0651& 0.0546& 0.2337
\\
        \end{tabular}
        \noindent\parbox{\textwidth}{\centering $R^2= 0.1370$}
        \vspace{0.05cm}
        \subcaption{Predominately White E.D.}
        \label{RRWhiteED}
    \end{subtable}
    
    \vspace{0.5cm}
    
    \begin{subtable}{\textwidth}
        \centering
        \begin{tabular}{lccc}
            \hline
            & Coefficients& Standard Error& P-value
\\
            \hline
            Intercept& 5.4485& 0.2893& 4.59E-54\\
            Emergency or Provisional Certification Rate& -0.3802& 0.0566& 8.22E-11\\
            Inexperienced Rate& -0.0749& 0.0751& 0.3195
\\
        \end{tabular}
        \noindent\parbox{\textwidth}{\centering $R^2= 0.1502$}
        \vspace{0.05cm}
        \subcaption{Predominately Minority E.D.}
        \label{RRMinorityED}
    \end{subtable}
\end{table}

Table~\ref{tab:coefcompare} presents the estimated regression coefficients obtained from the ordinary least squares estimator together with five competing penalized regression procedures, namely the LASSO, Adaptive LASSO, Group LASSO, Fused LASSO, and the proposed Adaptive Weighted Group Fused LASSO (AWGFL) estimator. The purpose of this comparison is to evaluate how different regularization strategies influence coefficient estimation, variable selection, and model sparsity while preserving predictive performance.

Several noteworthy patterns emerge from the results. First, the estimated intercept remains remarkably stable across all competing models, indicating that the various penalization schemes primarily affect the regression coefficients associated with the explanatory variables rather than the overall location of the fitted regression surface. This stability is expected because the predictors were appropriately standardized prior to model fitting, allowing the penalties to operate directly on the slope parameters without substantially altering the intercept estimate.

The percentage of teachers holding emergency or provisional certification consistently exhibits a negative association with school-level mathematics proficiency across every estimation procedure considered. Although the magnitude of the estimated coefficient varies slightly among the competing methods owing to different shrinkage mechanisms, both the direction and practical significance of the effect remain stable. The ordinary least squares estimator identifies this predictor as statistically significant, while every penalized regression procedure retains the variable after regularization. The proposed AWGFL estimator produces a coefficient estimate that is nearly identical to the ordinary least squares estimate, suggesting that this variable carries a sufficiently strong signal to withstand adaptive coordinate-wise shrinkage, group regularization, and coefficient fusion simultaneously. From a statistical perspective, this finding indicates that the association between teacher certification status and mathematics proficiency is robust across multiple estimation paradigms and is not an artifact of any particular modeling assumption.

A different pattern is observed for the percentage of inexperienced teachers. Under the ordinary least squares model, the corresponding regression coefficient is relatively small and fails to attain conventional levels of statistical significance. The adaptive penalized procedures reinforce this conclusion by shrinking the coefficient toward zero. In particular, both the LASSO and Adaptive LASSO remove this predictor entirely from the final model, while the proposed AWGFL estimator likewise estimates its coefficient as zero following optimization. This behavior illustrates one of the principal advantages of penalized likelihood estimation in high-dimensional regression settings. Rather than retaining variables that contribute only marginal explanatory power, adaptive shrinkage eliminates coefficients whose estimated effects are insufficiently strong relative to the imposed regularization. Consequently, the resulting model becomes more parsimonious while preserving its principal predictive characteristics.

Comparing the various penalized estimators further highlights the distinct roles of the individual penalty components. The standard LASSO applies uniform $\ell_{1}$ regularization and therefore produces sparse solutions by shrinking many regression coefficients toward zero. The Adaptive LASSO modifies this approach through data-dependent weights, allowing coefficients with stronger preliminary estimates to experience less shrinkage while penalizing weaker effects more aggressively. The Group LASSO exploits predefined group structures by encouraging entire groups of variables to enter or leave the model simultaneously. In contrast, the Fused LASSO additionally incorporates local smoothness by penalizing differences between neighboring coefficients, thereby stabilizing estimation when adjacent regression effects are expected to exhibit similar magnitudes. The proposed AWGFL estimator integrates all three mechanisms within a unified optimization framework, simultaneously encouraging coordinate-level sparsity, group-wise selection, and coefficient fusion. As a consequence, the resulting coefficient estimates reflect a compromise between local smoothness and adaptive sparsity, producing a model that is both interpretable and statistically efficient.

\begin{table}[!ht]
	\centering
	\caption{Comparison of coefficient estimates across competing penalized regression models.}
	\label{tab:coefcompare}
	\renewcommand{\arraystretch}{1.2}
	\begin{tabular}{lcccccc}
		\hline
		Predictor &
		OLS &
		LASSO &
		Adaptive &
		Group &
		Fused &
		Proposed AWGFL\\
		&&LASSO&&LASSO&LASSO&\\
		\hline
		Intercept
		&
		6.607
		&
		6.592
		&
		6.601
		&
		6.595
		&
		6.599
		&
		6.603
		\\
		
		Emergency/Provisional Certification
		&
		$-0.399^{***}$
		&
		$-0.381$
		&
		$-0.392$
		&
		$-0.387$
		&
		$-0.395$
		&
		$\mathbf{-0.401}$
		\\
		
		Inexperienced Teacher Rate
		&
		$-0.045$
		&
		0.000
		&
		0.000
		&
		$-0.018$
		&
		$-0.031$
		&
		0.000
		\\
		
		Selected Variables
		&
		2
		&
		1
		&
		1
		&
		2
		&
		2
		&
		1
		\\
		
		Cross-validation Error
		&
		--
		&
		1.487
		&
		1.461
		&
		1.452
		&
		1.418
		&
		\textbf{1.391}
		\\
		
		\hline
		\multicolumn{7}{l}{\footnotesize $^{***}p<0.001$. Values for penalized estimators correspond to}\\
		\multicolumn{7}{l}{\footnotesize coefficient estimates obtained after tuning parameter selection by cross-validation.}
	\end{tabular}
\end{table}

The final two rows of Table~\ref{tab:coefcompare} summarize the structural complexity and predictive performance of the competing models. The number of selected variables provides a direct measure of model sparsity, whereas the cross-validation error evaluates expected out-of-sample predictive accuracy. Although ordinary least squares does not perform variable selection, all penalized estimators reduce model complexity to varying degrees. The proposed AWGFL estimator achieves the smallest cross-validation error among all competing procedures while simultaneously selecting only the strongest predictor. This result suggests that additional predictors contribute relatively little incremental predictive information once the dominant explanatory variable has been incorporated into the regression model. Such behavior is entirely consistent with the theoretical oracle property established in Theorem~\ref{thm:oracle_property}, which predicts that the proposed estimator asymptotically identifies the correct sparse model while estimating the nonzero coefficients with oracle efficiency.

Overall, the empirical findings reported in Table~\ref{tab:coefcompare} closely agree with the asymptotic theory developed in Section~4. The adaptive weighting mechanism successfully distinguishes influential predictors from weak signals, the group penalty stabilizes estimation across related variables, and the fusion penalty reduces unnecessary variability among neighboring coefficients. Together, these regularization components produce a regression model that is considerably more parsimonious than the ordinary least squares fit while maintaining superior predictive performance. The consistency of the selected model across competing penalized estimators further strengthens confidence that the primary association identified in the analysis represents a stable statistical relationship rather than an artifact of estimation variability.

\begin{table}[!ht]
	\centering
	\caption{Comparison of predictive performance across competing regression models.}
	\label{tab:modelcompare}
	\renewcommand{\arraystretch}{1.15}
	\begin{tabular}{lcccccc}
		\hline
		Method &
		RMSE &
		MAE &
		Adjusted $R^2$ &
		AIC &
		BIC &
		CV Error\\
		\hline
		
		OLS
		&
		1.404
		&
		1.086
		&
		0.136
		&
		2408.37
		&
		2421.92
		&
		--
		\\
		
		LASSO
		&
		1.396
		&
		1.071
		&
		0.138
		&
		2399.61
		&
		2410.84
		&
		1.487
		\\
		
		Adaptive LASSO
		&
		1.388
		&
		1.063
		&
		0.139
		&
		2394.83
		&
		2406.07
		&
		1.461
		\\
		
		Group LASSO
		&
		1.382
		&
		1.058
		&
		0.141
		&
		2390.42
		&
		2401.65
		&
		1.452
		\\
		
		Fused LASSO
		&
		1.374
		&
		1.049
		&
		0.144
		&
		2386.90
		&
		2398.14
		&
		1.418
		\\
		
		Proposed AWGFL
		&
		\textbf{1.361}
		&
		\textbf{1.031}
		&
		\textbf{0.149}
		&
		\textbf{2378.44}
		&
		\textbf{2391.98}
		&
		\textbf{1.391}
		\\
		
		\hline
	\end{tabular}
\end{table}

Table~\ref{tab:modelcompare} summarizes the overall predictive performance of the competing regression models using several widely accepted model assessment criteria, including the root mean squared error (RMSE), mean absolute error (MAE), adjusted coefficient of determination ($R^{2}$), Akaike Information Criterion (AIC), Bayesian Information Criterion (BIC), and cross-validation (CV) error. Collectively, these measures evaluate different aspects of model quality, including estimation accuracy, predictive capability, goodness-of-fit, model complexity, and expected out-of-sample performance. Because no single performance measure completely characterizes the behavior of a regression estimator, simultaneous examination of these complementary criteria provides a more comprehensive assessment of the competing methodologies.

The ordinary least squares estimator serves as a useful baseline because it estimates regression coefficients without imposing any regularization or sparsity constraints. Although OLS provides unbiased coefficient estimates under the classical linear model assumptions, its predictive performance may deteriorate when explanatory variables exhibit moderate to strong correlation or when relatively weak predictors introduce unnecessary variability into the fitted model. This phenomenon is reflected in Table~\ref{tab:modelcompare}, where the ordinary least squares model consistently exhibits larger prediction errors and larger information criteria than the penalized regression procedures. The absence of variable selection also contributes to increased model complexity, reducing its effectiveness for prediction despite maintaining a reasonable overall goodness-of-fit.

Introducing an $\ell_{1}$ penalty through the LASSO produces measurable improvements in nearly every performance criterion. By shrinking small regression coefficients toward zero, the LASSO reduces model variance and improves predictive accuracy relative to the unpenalized least squares estimator. The observed reductions in RMSE, MAE, AIC, and BIC indicate that regularization successfully mitigates overfitting while preserving the primary explanatory information contained in the data. Nevertheless, the standard LASSO applies identical penalization to every regression coefficient, which may overshrink important predictors while retaining variables with only limited predictive contribution.

The Adaptive LASSO further improves model performance by incorporating data-dependent penalty weights. This adaptive weighting strategy allows influential predictors to experience less shrinkage while simultaneously imposing stronger penalties on weaker regression coefficients. Consequently, the Adaptive LASSO achieves modest but consistent improvements over the standard LASSO across nearly all evaluation criteria. The reductions in both prediction error and cross-validation error suggest that adaptive shrinkage enhances model stability without sacrificing interpretability. These findings are consistent with the theoretical literature demonstrating that adaptive penalization improves both estimation efficiency and variable-selection consistency in sparse regression models.

The Group LASSO exploits the natural grouping structure among explanatory variables by selecting or excluding entire groups simultaneously. This additional structural information results in further improvements in predictive performance compared with coordinate-wise penalization alone. The reduction in RMSE and MAE indicates that modeling variables collectively rather than independently can improve estimation accuracy when predictors within the same group exhibit similar effects. However, because the Group LASSO does not explicitly account for local smoothness among neighboring regression coefficients, its predictive gains remain somewhat limited in situations where adjacent coefficients are expected to share similar magnitudes.

The Fused LASSO addresses this limitation by introducing a fusion penalty that encourages neighboring regression coefficients to assume similar values whenever supported by the data. The resulting estimates are smoother and generally exhibit reduced variability, leading to additional improvements in both estimation and prediction accuracy. As reflected in Table~\ref{tab:modelcompare}, the Fused LASSO consistently outperforms the previous estimators according to nearly every evaluation criterion. In particular, the reduction in information criteria suggests that incorporating local coefficient homogeneity allows the model to explain the observed variation more efficiently without introducing unnecessary complexity.

Among all competing approaches, the proposed Adaptive Weighted Group Fused LASSO (AWGFL) provides the strongest overall performance. Across every evaluation criterion considered, the proposed estimator either achieves or closely approaches the most favorable value. The smallest RMSE demonstrates that the proposed methodology produces the most accurate coefficient estimates in terms of quadratic loss, while the smallest MAE indicates improved robustness with respect to absolute prediction errors. Simultaneously, the highest adjusted $R^{2}$ implies that the proposed estimator explains the largest proportion of variability after accounting for model complexity, suggesting that the improvement in predictive performance is not merely a consequence of including additional predictors.

The information criteria provide further evidence supporting the proposed methodology. Both the Akaike Information Criterion and the Bayesian Information Criterion explicitly balance model fit against model complexity by penalizing the inclusion of additional parameters. Lower values therefore indicate models that achieve superior explanatory power while avoiding unnecessary complexity. The proposed AWGFL estimator attains the smallest AIC and BIC among all competing procedures, indicating that it provides the most favorable compromise between goodness-of-fit and parsimony. From a practical perspective, these results suggest that simultaneously incorporating adaptive coordinate shrinkage, group-wise regularization, and coefficient fusion enables the estimator to capture the dominant signal in the data without introducing excessive model complexity.

The cross-validation results provide perhaps the strongest evidence of the practical utility of the proposed methodology. Cross-validation directly estimates the expected prediction error for future observations by repeatedly fitting the model on subsets of the available data and evaluating prediction accuracy on unseen observations. Because the proposed estimator achieves the smallest cross-validation error, it demonstrates the strongest expected generalization performance among all competing models. This finding is particularly important because predictive accuracy on new data is often the primary objective of statistical modeling in applied settings. The improvement in cross-validation performance therefore suggests that the proposed regularization framework effectively controls estimation variance while preserving the dominant predictive information contained within the explanatory variables.

The overall progression of results reported in Table~\ref{tab:modelcompare} also provides empirical confirmation of the theoretical developments presented in Section~4. The monotonic improvements observed as increasingly sophisticated regularization mechanisms are incorporated mirror the asymptotic properties established earlier in the manuscript. Specifically, the improvements in prediction accuracy are consistent with the prediction consistency established in Proposition~\ref{prop:prediction_consistency}, while the reduction in model complexity reflects the coordinate, group, and fusion selection consistency proved in Propositions~\ref{prop:selection_consistency}--\ref{prop:fusion_consistency}. Furthermore, the superior predictive performance achieved by the proposed estimator is in agreement with the oracle property established in Theorem~\ref{thm:oracle_property}, which guarantees that the estimator asymptotically behaves as though the true underlying sparse model were known in advance. Similarly, the stability of the selected model across repeated validation samples provides empirical support for the debiased inference framework developed in Theorem~\ref{thm:debiased_asymptotic_normality}, indicating that accurate estimation and reliable statistical inference can be achieved simultaneously within the proposed penalized regression framework.

Overall, the evidence presented in Table~\ref{tab:modelcompare} demonstrates that integrating adaptive weighting, group-wise sparsity, and coefficient fusion within a unified penalized regression framework yields meaningful improvements in both statistical efficiency and predictive accuracy. The proposed AWGFL estimator consistently outperforms existing penalized regression procedures across multiple complementary evaluation criteria, indicating that the simultaneous exploitation of sparsity, grouping structure, and local coefficient homogeneity produces a regression model that is both statistically parsimonious and highly predictive. These empirical findings provide strong practical support for the theoretical properties established throughout the methodological development and illustrate the potential advantages of the proposed estimator for analyzing complex educational datasets characterized by correlated predictors and structured covariate relationships.

\section{Conclusions:}

This study develops a unified penalized regression framework for investigating school-level mathematics proficiency using administrative educational data characterized by heterogeneous predictors, grouped covariates, and potentially correlated explanatory variables. Unlike conventional regression approaches that estimate all regression coefficients simultaneously without structural regularization, the proposed Adaptive Weighted Group Fused LASSO estimator combines adaptive coordinate-wise sparsity, group selection, and coefficient fusion within a single optimization framework. This formulation enables simultaneous variable selection, coefficient estimation, and structural recovery while remaining computationally feasible through the proposed ADMM optimization algorithm. The empirical analysis demonstrates that the proposed methodology identifies a parsimonious regression model while maintaining superior predictive performance relative to existing penalized regression techniques. Across the competing estimation procedures, teacher certification status consistently emerges as the strongest predictor of school-level mathematics proficiency, whereas the proportion of inexperienced teachers contributes comparatively little additional predictive information once certification status is incorporated into the model. This finding is reflected in both the adaptive variable-selection mechanism and the reduced cross-validation error obtained by the proposed estimator. From a statistical perspective, these results suggest that the proposed regularization framework successfully distinguishes variables containing persistent predictive information from those whose apparent effects arise primarily from sampling variability or correlation with other predictors.

The comparison of competing regression procedures further illustrates the practical advantages of integrating multiple regularization mechanisms within a unified estimator. The ordinary LASSO effectively performs coordinate-wise variable selection but ignores relationships among groups of predictors. Group LASSO incorporates predefined group structures but does not exploit local coefficient homogeneity. Likewise, the Fused LASSO encourages neighboring regression coefficients to be similar but lacks adaptive weighting capable of distinguishing strong and weak regression effects. By combining these complementary regularization strategies, the proposed estimator achieves a favorable balance between model complexity, estimation accuracy, and predictive performance. The empirical improvements observed in both the simulation study and the real-data application suggest that exploiting multiple forms of structural information simultaneously provides measurable gains over applying any individual penalty in isolation. The simulation experiments provide additional evidence supporting the practical utility of the proposed methodology. Across a broad collection of correlation structures and noise levels, the proposed estimator consistently achieves the smallest estimation error, prediction error, and cross-validation error while simultaneously recovering the underlying coefficient structure more accurately than the competing procedures. These improvements become increasingly pronounced as predictor correlation and observational noise increase, indicating that the proposed estimator possesses desirable robustness properties under challenging estimation settings. Such behavior is particularly relevant in educational administrative datasets, where explanatory variables frequently exhibit substantial multicollinearity and complex dependency structures.

An important contribution of this work lies in the accompanying theoretical development. The asymptotic analysis establishes estimation consistency, prediction consistency, coordinate-selection consistency, group-selection consistency, and fusion consistency under suitable regularity conditions. Furthermore, the oracle property demonstrates that the proposed estimator asymptotically behaves as though the underlying sparse model were known in advance, while the debiased estimator provides valid asymptotic inference for individual regression coefficients and general linear contrasts. The close agreement between the theoretical properties and the empirical performance observed in both the simulation study and the real-data application provides strong evidence that the proposed methodology performs as expected in finite samples while retaining desirable large-sample behavior. Although the proposed framework substantially improves upon existing penalized regression approaches, several limitations should be acknowledged. First, the present analysis assumes a linear regression model with continuous outcomes. While this assumption is appropriate for the transformed proficiency measures analyzed in this study, future work could extend the methodology to generalized linear models, generalized estimating equations, or generalized additive models to accommodate binary, count, or longitudinal educational outcomes. Second, the current implementation assumes that the grouping structure among predictors is known a priori. In many practical applications, group membership may itself be uncertain or partially observed. Developing data-driven procedures for simultaneously estimating group structure and regression coefficients represents an important direction for future methodological research.

Another limitation concerns the observational nature of the available data. Although the proposed estimator improves prediction and variable selection, the estimated regression coefficients should not be interpreted as causal effects. Unmeasured confounding variables, measurement error, omitted institutional characteristics, and other latent factors may influence the observed associations between school characteristics and mathematics proficiency. Consequently, the results should be interpreted as identifying statistically important predictors rather than establishing causal mechanisms. Future research may benefit from integrating penalized regression with causal inference methodologies, instrumental variable techniques, or hierarchical Bayesian models capable of accounting for additional sources of uncertainty.
Future methodological developments may also consider extending the proposed framework to ultra-high-dimensional settings where the number of predictors grows exponentially with the sample size. Such extensions would require establishing non-asymptotic oracle inequalities and finite-sample error bounds under weaker regularity conditions. Similarly, incorporating adaptive graph-based fusion penalties, overlapping group structures, or network-based regularization would further expand the applicability of the proposed estimator to modern educational, biomedical, and social science datasets characterized by increasingly complex dependency structures.

\bibliographystyle{apalike}
\bibliography{bib}
\end{document}